\RequirePackage{fix-cm}

\documentclass[twocolumn]{svjour3}       

\smartqed  

\usepackage{flushend}
\usepackage{graphicx}
\usepackage[linesnumbered,ruled, noend]{algorithm2e}
\usepackage[noend]{algpseudocode}
\usepackage{varwidth}
\usepackage{multirow}
\usepackage{graphicx}
\usepackage{amsmath}
\usepackage{amssymb}
\usepackage{caption}
\usepackage{subfig}
\graphicspath {{images/}}

\begin{document}

\title{No-But-Semantic-Match: Computing Semantically Matched XML Keyword Search Results
}

\author{Mehdi Naseriparsa  \and
        Md. Saiful Islam   \and
				Chengfei Liu       \and
				Irene Moser
}


\institute{M. Naseriparsa \and
              Md. Saiful Islam  \and
							C. Liu \and 
							I. Moser \\
              \email{\{mnaseriparsa,mdsaifulislam,cliu,imoser\}@swin.edu.au}\\
							Swinburne University of Technology
							}          

\date{Received: date / Accepted: date}

\maketitle

\begin{abstract}
Users are rarely familiar with the content of a data source they are querying, and therefore cannot avoid using keywords that do not exist in the data source. Traditional systems may respond with an empty result, causing dissatisfaction, while the data source in effect holds semantically related content. In this paper we study this no-but-semantic-match problem on XML keyword search and propose a solution which enables us to present the top-k semantically related results to the user. Our solution involves two steps: (a) extracting semantically related candidate queries from the original query and (b) processing candidate queries and retrieving the top-$k$ semantically related results. Candidate queries are generated by replacement of non-mapped keywords with candidate keywords obtained from an ontological knowledge base. Candidate results are scored using their cohesiveness and their similarity to the original query. Since the number of queries to process can be large, with each result having to be analyzed, we propose pruning techniques to retrieve the top-$k$ results efficiently. We develop two query processing algorithms based on our pruning techniques. Further, we exploit a property of the candidate queries to propose a technique for processing multiple queries in batch, which improves the performance substantially. Extensive experiments on two real datasets verify the effectiveness and efficiency of the proposed approaches.

\end{abstract}

\begin{figure*}

\centering
             \includegraphics[width=7in,height=4.4cm]{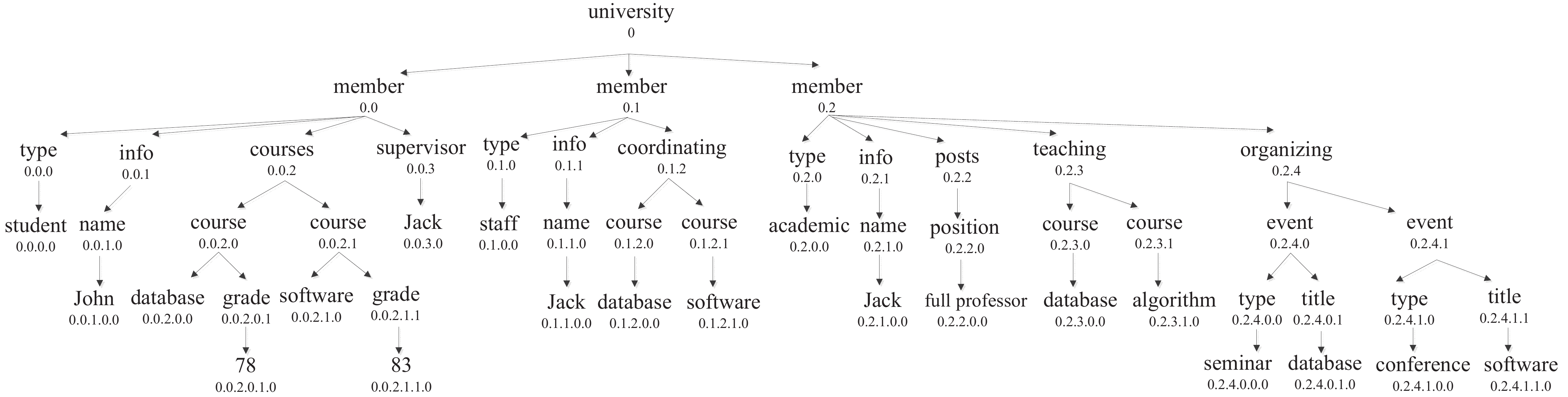}
						  	\caption{A part of XML data}
								
					\label{fig:Motivating}	
\end{figure*}

\section{Introduction}
\label{intro}
{U}{sers} who query data sources using keyword searches often are not familiar with the data source schema or the appropriate query language. For the query to succeed, the keywords have to have matches in the data source. Failing this, an empty result is returned even when semantically related content exists. When keywords have indirect mappings in a data source that cannot be found by traditional systems, the user faces the no-but-semantic-match problem.
 
\begin{example}.
Consider a user submitted a keyword query $q_0$ = \{$Jack$, $lecturer$, $class$\} on XML database given in Fig. \ref{fig:Motivating} and would like to find information about the professor $Jack$. Using conjunctive keyword search, traditional systems will show an empty result because there is no occurrence for the keywords $lecturer$ and $class$ in the data source. However, the keyword $lecturer$ has a semantic connection to $academic$ and $full$ $professor$ while the keyword $class$ is semantically related to $course$, $grade$ and $event$ which exist in the data source and could generate results that might interest the user.
\end{example}

The XML keyword search has been addressed by researchers before. 
The concept of Lowest Common Ancestor (LCA) was first proposed by Guo et al.~\cite{GuoSBS03} to extract XML nodes which contain all query keywords within the same subtree. 
Xu and Papakonstantinou \cite{XuP05} introduced the concept of Smallest Lowest Common Ancestor (SLCA) to reduce the query result to the smallest tree that contains all keywords. 
Sun, Chan and Goenka~\cite{SunCG07} extended this work by applying the SLCA principle to logical OR searches. 
Hristidis et al.~\cite{HristidisKPS06} explored the trees below LCA to provide information about the proximity of the keywords in the document. 
None of the existing studies use the SLCA semantics to provide a solution when one or more keywords do not exist in the database.  
In this paper, we adapt the widely-accepted SLCA semantics and algorithms to retrieve meaningful results when some non-mapped keywords are submitted to the system.        

When a query encounters the no-but-semantic-match problem, we need to find candidate keywords for the non-mapped keywords to produce non-empty results. Even though the non-mapped keywords may be semantically close to some items in data source, traditional systems do not attempt to discover them. To produce an answer to the user's initial query, the candidate keywords must be semantically  close to the non-mapped keywords. One way of fulfilling this requirement is to find substitutes for non-mapped keywords in an ontological knowledge base. Clearly, only candidate keywords that have a mapping in the data source can be selected as substitutes for a new query. Replacing each of the non-mapped keywords with one or more semantically related words that are known to exist in the database leads to a list of candidate queries. Depending on the number of available keywords, the number of potential queries and results can be impractically large. Hence the degree of semantic similarity with the original query is calculated for each candidate query before it is executed. Before the results can be presented to the user, results of poor quality in terms of cohesiveness must be eliminated to ensure all results are meaningful answers to the original query. 
Thus, to solve the no-but-semantic-match problem, two aspects are considered: (a) query similarity; and (b) result cohesiveness.

\begin{figure*}[tb]
\centering
\includegraphics[scale=0.6, angle=-90]{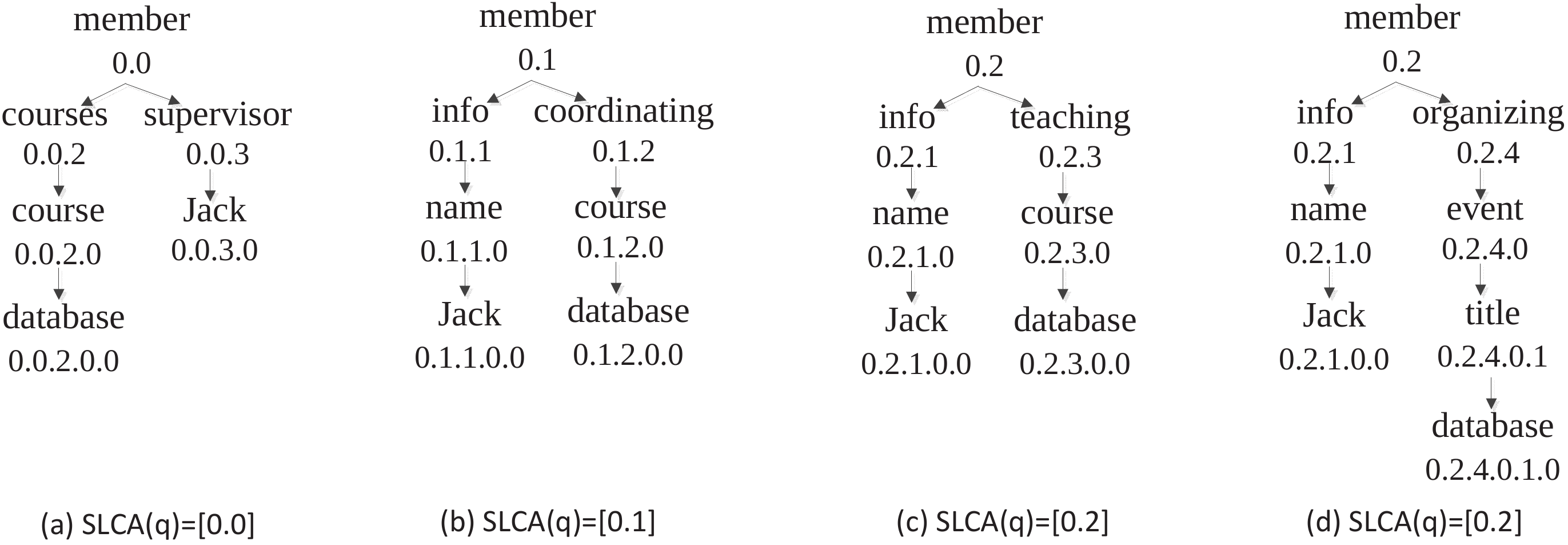}
	\caption{SLCA subtree results for query $q=\{Jack,database\}$ executed on data given in Fig.~\ref{fig:Motivating}}
\label{fig:Preliminaries}	
\end{figure*}

\begin{example}.
Consider the keyword query $q_0$ = \{$Jack$, $lecturer$, $class$\} presented in Example 1 on the database shown in Fig. \ref{fig:Motivating}. Keywords $lecturer$ and $class$ do not have a mapping in the data source and the traditional system generates an empty result for it. The ontological knowledge base \cite{Miller90} has $44$ semantic counterparts for $lecturer$ and $39$ for $class$. All possible substitutions and their combinations are considered. In the extreme case when all candidate keywords are available in the data source, $44 \times 39$ = $1716$ queries are generated and each query may have several answers that have to be considered. When a high number of keywords have to be replaced and these keywords have many semantic counterparts, we may face an unmanageably large number of combinations that have to be analyzed for semantic similarity with the original query. Hence, there is a need to identify and remove less promising candidate queries early.
\end{example}

In this paper, we present a novel two-step solution to the no-but-semantic-match problem in XML keyword search. In the first step, semantically related candidate queries are created by replacing non-mapped keywords in the original queries with semantic counterparts and in the second step, the queries are processed and the top-k semantically related results retrieved. In order to present the top-k results to the user for evaluation, each result retrieved from the queries is separately analyzed in terms of its similarity to the original query and its cohesiveness in data source. Since there may be a large number of semantically related results, retrieving the top-k results is potentially costly. Therefore, we propose two pruning techniques, inter-query and intra-query pruning. Since the candidate queries are generated by replacing non-mapped keywords, some keywords are shared between the candidate queries. We exploit this property to propose a more efficient batch query processing technique to improve the performance substantially. The issue of finding semantically related results for queries with no-but-semantic-match problem has not been addressed in the context of semi-structured data before. Our contributions are as follows:
\begin{enumerate}
	\item We are the first to formulate the no-but-semantic-match problem in XML keyword search. 
	\item We propose two pruning methods and an efficient approach of processing the no-but-semantic-match query. 
 	\item Based on keywords the candidate queries have in common, we also propose a method to process multiple queries in a batch which improves the performance substantially.   
	\item We conduct extensive experiments which verify the effectiveness and efficiency of our solutions on two real datasets.   
\end{enumerate}

The rest of the paper is organized as follows: Section 2 discusses XML keyword search and presents the no-but-semantic-match problem. Section 3 presents the details of our pruning ideas and the efficient processing of the no-but-semantic-match query. Section 4 presents the batch query processing scheme to further improve the performance. The experiments are presented in Section 5. Section 6 reviews the related work. Finally, Section 7 concludes our paper.

\section{Background}
\subsection{Preliminaries}
An XML document is an ordered tree $T$ with labeled nodes and a designated root. 
All XML elements are treated as nodes containing information in $T$. 
There are parent-child and sibling relationships between the nodes. 
The depth of the tree is denoted as $d$, and the root node has a depth of 1. 
Each node $v$ in the tree $T$ is marked with a unique identifier in Dewey code, which describes the path from the root to the node $v$ as a sequence of numbers separated by a dot (``."). 
Sibling nodes have Dewey codes of equal length with a unique last number.

\begin{example}. Fig. \ref{fig:Motivating} shows an XML tree which contains information about staff and students of a university. 
The root node's Dewey code is 0. 
Dewey code 0.0.1 refers to a node containing information about a member of the university and the code prefix 0.0 refers to its parent node.
\end{example}

\textbf{Keyword Match Node:} A node $m$ in the tree $T$ is a match node for keyword $k_i$ if it contains $k_i$. 
e.g., the match nodes for keyword $k_1 = database$ presented in Fig. \ref{fig:Motivating} are: $m_1^1 = [0.0.2.0.0]$, $m_1^2 = [0.1.2.0.0]$, $m_1^3 = [0.2.3.0.0]$, and $m_1^4 = [0.2.4.0.1.0]$.

\textbf{Keyword Inverted List:} Each keyword $k_i$ corresponds to a list $S_i$ of entries and each entry corresponds to a node $m$ which contains $k_i$ in the tree $T$. 
e.g., the keyword inverted list for keyword $k_1 = database$ is $S_1=\{[0.0.2.0.0],[0.1.2.0.0],[0.2.3.$ $0.0],[0.2.4.0.1.0]\}$.

\textbf{Smallest Lowest Common Ancestor (SLCA):} Let $lca(m_1,...,m_n)$ returns the lowest common ancestor (LCA) of match nodes $m_1,...,m_n$. 
Then LCAs of query $q$ on $T$ are defined as $LCA(q)=\{v|v = lca(m_1,...,m_n),$ $m_i \in S_i ( 1 \leq i \leq n) \}$. 
SLCAs are a subset of LCAs which do not have other LCAs as child nodes and defined as $SLCA(q)$.

\begin{example}. In Fig. \ref{fig:Preliminaries}, for a keyword query $q$=\{$Jack$, $database$\}, there are 4 LCA nodes which are computed as: $LCA(q)$ = $\{lca([0.0.3.0],[0.0.2.0.0]),lca([0.0.3.0],[0.1$ $.2.0.0]),lca([0.1.1$ $.0.0],[0.1.2.0.0]),lca([0.2.1.0.0],[0.2.3.0$ $.0])\}$=$\{[0],[0,0],[0.1]$ $,[0.2]\}$. 
Since the LCA node $[0]$ is the ancestor node of $[0.0]$,$[0.1]$ and $[0.2]$, it is not an SLCA and should be removed. Therefore, $SLCA(q) = \{[0.0],[0.1],[0.2]\}$.
\end{example}

\textbf{Keyword Query and Subtree Result:} In XML data, a keyword query $q$ consists of a set of keywords $\{k_1,k_2, ... ,k_n\}$. 
A result $r = (v_{slca},\{m_1,m_2,...,m_n\})$ for $q$ is a subtree in $T$ which contains all keywords $k_i\in q$. 
Here, we consider $v_{slca}$, the root of the subtree, an SLCA node, i.e. $v_{slca}\in SLCA(q)$.
    
\textbf{Tightest SLCA Subtree Result:} For an SLCA node, there may exist several subtree results. 
This is because that under an SLCA node $v_{slca}$, we may find several match nodes $\{m_i^j\}$ for the keyword $k_i, (1 \leq i \leq n, 1 \leq j \leq n_i)$, where $n_i$ is the number of match nodes for $k_i$ under $v_{slca}$.
Let $m_i^{l_i}$ be the closest match node from $\{m_i^j\}$ to $v_{slca}$ for $k_i (1 \leq i \leq n)$, then we get the the tightest subtree result $r = (v_{slca},\{m_1^{l_1},..., m_i^{l_i},...,m_n^{l_n}\})$.

For example, for $SLCA(q) = [0.2]$ in Fig. \ref{fig:Preliminaries}, there are two subtree results, (c) and (d). 
The tightest subtree result is (c) $r =([0.2],\{[0.2.1.0.0],[0.2.3.0.0]\})$. 

We argue to return only the tightest SLCA subtree results to the user as these results match the user's search intention better than the results containing the sparsely distributed keyword match nodes under $v_{slca}$. That is, a result is more likely to be meaningful when the result subtree is more tight and cohesive (for survey \cite{HristidisKPS06}, \cite{FengLWZ10}).

\begin{table}
\centering
\caption{The list of symbols}
\begin{tabular}{ |l|l| }
\hline
\centering
$Symbol$ & $Meaning$ \\ \hline
\hline
$\lambda(q_0,q^\prime)$ & Similarity score of $q_0$ to $q^\prime$ \\ \hline
$\alpha$ & Tuning parameter \\ \hline
$\sigma^{min}$ & Threshold score \\ \hline
$\sigma(r,q^\prime,T)$ & Total score of a result \\ \hline
$q_0$ & User original query  \\ \hline
$q^\prime$ & A candidate query  \\ \hline
$\mathcal{Q}$ & A set of candidate queries \\ \hline
$\mathcal{S}$ & A set of Keyword Inverted lists \\ \hline
$\mathcal{B}$ & Candidate query batch \\ \hline
$\mathcal{R}$ & A set of results \\ \hline
$\mathcal{R}^*$ & A set of top-k results \\ \hline
$\Delta(r_1,r_2)$ & The score difference between $r_1$ and $r_2$ \\ \hline
$r$ & A result \\ \hline
$v_{slca}$ & A subtree result root \\ \hline
$m$ & A match node \\ \hline
$n$ & Number of keywords in a query \\ \hline
$m^l$ & Tightest match node \\ \hline
$\mathcal{P}$ & An execution plan \\ \hline
$c({\mathcal{B}})$ & Cost of an execution plan \\ \hline
$T$ & XML data \\ \hline
$d(r,T)$ & Number of edges in a result $r$ \\ \hline
$\theta(r,T)$ & Cohesiveness score of a result $r$ \\ \hline
$\mathcal{K}$ & A set of candidate keywords \\ \hline
$k$ & A query keyword \\ \hline
\end{tabular}
\vspace{-2ex}
\end{table}    

\subsection{Problem Statement} \label{sec:pstate}
\begin{definition}\textbf{(No-Match Problem)}
Given a keyword query $q_0$ = \{ $k_1$,$k_2$, ... ,$k_n$ \} on $T$, if $\exists$ $k_i\in q_0$ such that $S_i$ = $\emptyset$, we say that query $q_0$ has a no-match problem over $k_i$.      
\end{definition}

If a user submits a keyword query $q_0$ that has a no-match problem, traditional systems return an empty result set. 
However, the missing keyword $k_i$ that causes the no-match problem may have semantic counterparts in the data source $T$ which may produce results the user might be interested in, if the candidate keywords are sufficiently similar to $k_i\in q_0$. 
We use $\mathcal{K}_i$ to denote the list of candidate keywords that can be used to replace $k_i\in q_0$.

\begin{example}. Consider the keyword query $q_0$ = $\{Jack$, $lecturer$, $class\}$ presented in Example 1. 
It is easy to verify that the keywords $k_2 = lecturer$ and $k_3 = class$ cause a no-match problem for $q_0$. 
Candidate keywords that can be used instead of $k_2$ and $k_3$ for $q_0$ are: $\mathcal{K}_2$ = \{$academic$, $full$ $professor$ \} and $\mathcal{K}_3$ = \{$course$,$grade$,$event$\}.
\end{example}

\begin{definition}\textbf{(No-But-Semantic-Match Problem)} \\
Given a keyword query $q_0$ = \{ $k_1$,$k_2$, ... ,$k_n$ \} with a no-match problem on $T$, i.e., $\exists k_i \in q_0$ such that $S_i = \emptyset$, but $\mathcal{K}_i$ $\neq$ $\emptyset$, then we say that $q_0$ has a no-but-semantic-match problem over $k_i$.      
\end{definition}

The no-but-semantic-match problem is a special case of the no-match problem. The problem can be addressed in the following way: 
(a) find a candidate keyword list $\mathcal{K}_i$ that can be used to replace $k_i\in q_0$; 
(b) generate candidate queries $q^\prime$ for $q_0$ by replacing $k_i$ with $k^\prime_i\in \mathcal{K}_i$;  
(c) execute $q^\prime$ in the data source $T$ to produce the semantically related results $\mathcal{R}$ for $q_0$; 
(d) score and rank the results $r\in\mathcal{R}$ to return only the top quality results to the user for evaluation.

\begin{example}. Consider the keyword query presented in Example 1. 
The query $q_0$ has no-but-semantic-match problem over $k_2 =lecturer$ and $k_3 =class$. 
The semantic counterparts for $k_2$ and $k_3$ are: $\mathcal{K}_2$ = \{$academic$, $full$ $professor$\} and $\mathcal{K}_3$ = \{$course$,$grade$,$event$\}. 
These candidate keywords are combined with the rest of the keywords to generate semantically related candidate queries for $q_0$. 
The generated candidate queries are: $q_1$ = \{$Jack$, $academic$, $course$\}, $q_2$ = \{$Jack$, $academic$, $grade$\}, $q_3$ = \{$Jack$, $full$ $professor$, $course$\}, $q_4$ = \{$Jack$, $full$ $professor$, $grade$\}, $q_5$ = \{$Jack$, $academic$, $event$\}, and $q_6$ = \{$Jack$, $full$ $professor$, $event$\}. 
\end{example}

We use $\mathcal{Q}$ to denote the list of candidate queries. As mentioned before, the candidate queries $q^\prime\in\mathcal{Q}$ need to be executed against $T$ to produce the semantically related result set $\mathcal{R}$ for $q_0$. We know that these candidate queries $q^\prime$ are generated by replacing $k_i$ with $k^\prime_i\in\mathcal{K}_i$. 
However, not all candidate keywords $k^\prime_i\in\mathcal{K}_i$ are semantically similar to the user given keyword $k_i\in q_0$ and also, not all semantically related results $r\in\mathcal{R}$ are meaningful to the same degree. 
Therefore, we need to score the produced results $r\in\mathcal{R}$, denoted by $\sigma(r, q^\prime, T)$, as given as follows:

\begin{equation}
\sigma(r, q^\prime, T)=sim(q_0, q^\prime)\times coh(r, T)
\label{eq:score}
\end{equation}
where, $r$ is a result for the candidate query $q^\prime$, the $sim(q_0, q^\prime)$ measures the similarity of $q^\prime$ with $q_0$ and $coh(r, T)$ measures the cohesiveness of $r$ in $T$. 
The rank of a result $r\in\mathcal{R}$ is calculated as follows:

\begin{equation}
rank(r, T)=|\{r^{\prime}|\sigma(r^{\prime}, q^{\prime\prime}, T)>\sigma(r, q^\prime, T)\}|+1
\label{eq:rank}
\end{equation}  
 
\begin{definition}\label{de:topk} \textbf{Top-k Semantically Related Results.} Given a keyword query $q_0 = \{k_1,k_2,...,k_n\}$ on $T$, having the no-but-semantic-match problem, we want to discover $k$ results from $\mathcal{R}$ that maximizes the scoring function given in Eq. \ref{eq:score} or in terms of ranking the results $\{r|rank(r, T) \leq k \}$.   
\end{definition}
    
\section{Our Approach}
We propose a two phase approach to solve the no-but-semantic-match problem in XML data $T$. The schematic diagram of our approach is illustrated in Fig. \ref{fig:framework}. In the first phase, the semantic counterparts for the non-mapped keywords of the user query are extracted from the ontological knowledge base. Next, the candidate queries are generated by replacing the non-mapped keywords with their semantic counterparts and the similarities between the candidate queries and the original query are computed. In the second phase, the candidate queries are executed against the data source $T$. The results are scored based on Eq. \ref{eq:score} and finally, only the top-k results are presented to the user for evaluation. The results are scored based on the followings: (a) similarity of the candidate queries to the user given query; and (b) the cohesiveness of the results. 

As the candidate queries are generated using the ontological knowledge base, we use the \emph{ontological similarity} of the candidate query to the user given query as the measure of similarity for the first parameter. The details for computing this similarity is presented in section \ref{sec:ontology}.
The details for computing the \emph{cohesiveness} of the results is presented in section \ref{sec:coherency}. Since there are a number of candidate queries that should be executed against the data source $T$ and each candidate query may have several results that needs to be scored, we propose efficient pruning techniques to avoid unnecessary computations and terminate early. The pruning ideas and the details of our candidate query processing technique are presented in Section \ref{sec:pruning}. We also propose a \emph{batch} query processing technique to speed up the computations further by sharing the computations among the candidate queries, which is described in Section \ref{sec:multiple}.                 

\begin{figure}[tb]
\centering
 \includegraphics[width = 3.20in]{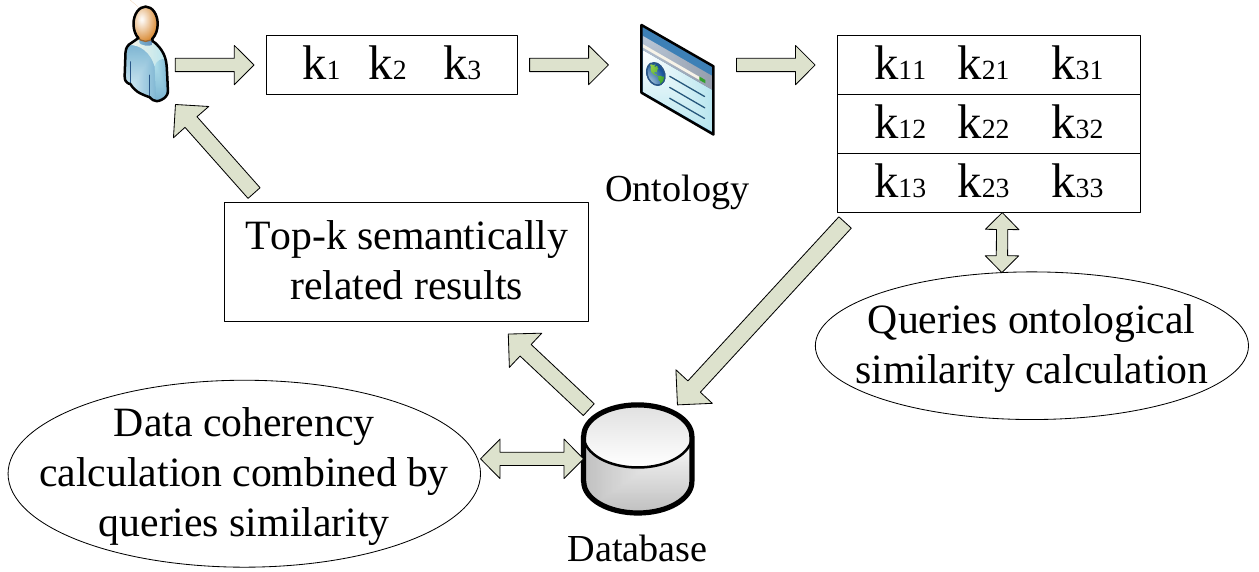}
	\caption{Schematic diagram of our approach for solving no-but-semantic-match problem in XML data}
	\label{fig:framework}
\end{figure}

\subsection{Candidate Queries}
\label{sec:ontology}
This section describes how to generate the candidate queries $\mathcal{Q}$ for the user query $q_0$ and compute their similarity to $q_0$. 

\subsubsection{Generating Candidate Queries}
\label{sec:generatingcandidatequeries}
To adhere closely to the user's intentions, the candidate keywords $k^{\prime}_i \in \mathcal{K}_i$ must be as close as possible to the non-mapped keywords $k_i\in q_0$. In this study we use WordNet, which is widely used in the literature \cite{BouquetKZ05} for finding semantic counterparts for $k_i \in q_0$. 
We categorize our semantic candidate keywords derived from WordNet into four groups \cite{Miller90}: (a) synonyms denoted as $Syn(k_i)$, (b) coordinate terms denoted as $Cot(k_i)$, (c) hyponyms denoted as $Hpo(k_i)$, and (d) hypernyms denoted as $Hpe(k_i)$. The candidate keyword list $\mathcal{K}_i$ for a keyword $k_i\in q_0$ contains all types of ontological counterparts as shown in Eq. \ref{eq:keywordList}.  

\begin{equation}
\label{eq:keywordList}  
\mathcal{K}_i = Syn(k_i) \cup Cot(k_i) \cup Hpo(k_i) \cup Hpe(k_i) 
\end{equation}

However, not all candidate keywords extracted from the ontological knowledge base are available in $T$. The keyword list has to be reduced to the candidates which have direct mapping in the data source. To do this, an inverted keyword list using hash indices can be queried in $\mathcal{O}(1)$ time. Finally, the candidate queries $\mathcal{Q}$ are generated by replacing the non-mapped keywords $k_i \in q_0$ with each of their semantic counterparts $k^{\prime}_i \in \mathcal{K}_i$ in turn.

\subsubsection{Measuring Candidate Query Similarity}
\label{sec:candidatequerysimilarity}
In order to measure the similarity between a candidate query $q^\prime\in\mathcal{Q}$ and the user's original query $q_0$, firstly we measure the individual similarity between the candidate keyword $k^{\prime}_i \in q^\prime$ and the corresponding non-mapped keyword $k_i \in q_0$ using Wu and Palmer's metric \cite{WuP1994}. This metric establishes the depths of both keywords and their least common subsumer (LCS) according to the WordNet structure and produces the degree of similarity between these two keywords, $SimWP(k_i,k^{\prime}_i)$,  as shown in Eq. \ref{eq:SimWP}. 

\begin{equation}
\label{eq:SimWP}  
SimWP(k_i,k^{\prime}_i) = \dfrac{ {2} \times {dep(LCS)}} { {dep(k_i)} + {dep(k^{\prime}_i)} }
\end{equation}

where $dep(k_i)$ returns the depth of the keyword $k_i$ in the WordNet structure. This metric is symmetric, i.e., $SimWP$ $(k_i,k^{\prime}_i)=SimWP(k^{\prime}_i, k_i)$. However, WordNet has a hierarchical structure. That is, the candidate keyword $k^{\prime}_i\in q^\prime$ could be a more special type (e.g., hyponyms) or a more general type (e.g., hypernyms) for the non-mapped keyword $k_i \in q_0$ in WordNet. Therefore, we incorporate the specialization/generalization aspect of $k_i\in q_0$ into the Wu and Palmer similarity metric as given as follows:
\begin{equation}
\label{eq:DSim}  
DSim(k_i,k^{\prime}_i) =\frac{dep(k^{\prime}_i)} {max(dep(k_i,k^{\prime}_i))} \times SimWP(k_i,k^{\prime}_i)    
\end{equation}

where $DSim(k_i,k^{\prime}_i)$ is the directional similarity of keyword $k_i\in q_0$ to keyword $k^{\prime}_i\in q^\prime$. The directional similarity $DSim(k_i,k^{\prime}_i)$ penalizes the more general keyword types of $k_i\in q_0$ by weighting $SimWP(k_i,k^{\prime}_i)$ with $\frac{dep(k^{\prime}_i)} {max(dep(k_i,k^{\prime}_i))}$. Finally, the similarity of $q^{\prime}$ to the original query $q_0$ is computed by considering all the replacements in $q^{\prime}$ as follows:   

\begin{equation}
\label{eq:querysimilarity}
\lambda(q_0,q^{\prime}) = \prod\limits_{i = 1}^ n {DSim(k_i \in q_0,k'_i \in q^{\prime})}
\end{equation}
where $n$ is the number of keywords in $q_0$ that have been replaced to generate $q^{\prime}$ ($1\leq n \leq |q|$).

\begin{example}. Consider the keyword query $q_0$ = \{$Jack$, $lecturer$, $class$\} presented in Example 1. Here, the second and the third keywords cause the no-match problem for $q_0$. The candidate keyword list for these two non-mapped keywords are: $\mathcal{K}_2$ = \{$academic:0.91$, $full$ $professor:0.84$\} and $\mathcal{K}_3$ = \{$course:1$, $grade:1$, $event:0.35$\}, where each candidate keyword is labeled with their corresponding \emph{DSim} scores. Now, the candidate queries are generated by replacing the non-mapped keywords in $q_0$ with their candidate keywords and scored as follows:\\ 
$q_1=\{Jack, academic, course\}$,$\lambda(q_0,q_1)$=$0.91 \times 1 = 0.91$,\\
$q_2=\{Jack, academic, grade\}$,$\lambda(q_0,q_2)$=$0.91 \times 1 = 0.91$,\\
$q_3=\{Jack,fullprof,course\}$,$\lambda(q_0,q_3)$=$0.84\times1=0.84$,\\
$q_4=\{Jack,fullprof,grade\}$,$\lambda(q_0,q_4)$=$0.84 \times 1=0.84$,\\        
$q_5=\{Jack,academic,event\}$,$\lambda(q_0,q_5)$=$0.91 \times 0.35$ $=0.31$ ,and 
$q_6=\{Jack,fullprof,event\}$,$\lambda(q_0,q_6)$=$0.84\times0.35=0.29$.
\end{example}

From the above, it is easy to verify that $q_1$ and $q_2$ are the most similar candidate queries to the original query.
\begin{figure*}[tb]
\centering
 \includegraphics[width = 5.5in]{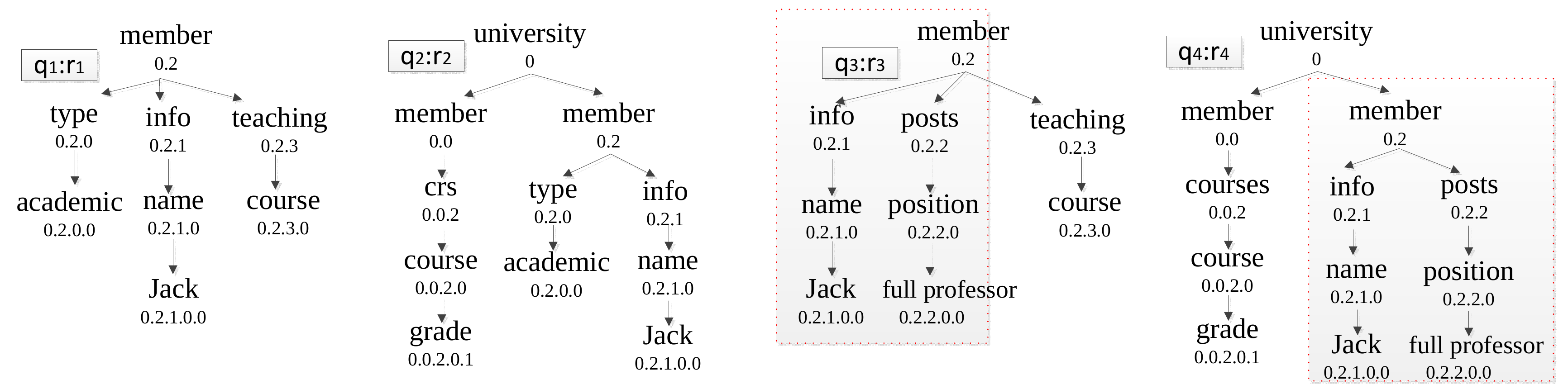}
	\caption{Subtree results for candidate queries $\mathcal{Q}$ in Example 7 after executing them against data given in Figure \ref{fig:Motivating}.}
	\label{fig:Coherency}
\end{figure*}

\subsection{Cohesiveness of Results}
\label{sec:coherency}
There can be potentially many candidate queries with a large number of results and not all results $r\in \mathcal{R}$ are meaningful to the same degree. In XML, if the match nodes in the result $r$ are near to each other, the result is considered to be more cohesive. Intuitively, when a result subtree is more cohesive, it is more likely to be relevant and meaningful (for survey \cite{HristidisKPS06}, \cite{FengLWZ10}). To measure the cohesiveness of a result $r$, we firstly compute the distance between each match node $m$ and the root $v_{slca}$ of the result subtree $r$. Then, we compute the overall distance of a result $r$ w.r.t. the data source $T$ as given as follows \cite{FengLWZ10}: 

\vspace{-2ex}
\begin{equation}
\label{eq:distance}
d(r,T) = \sum\limits{(l_m - l_{v_{slca}})} 
\end{equation}
where $r$ is the result subtree, $v_{slca}$ denotes the root of the result $r$, $m$ is the match node, $l_m$ is the level of $m$, and $l_{v_{slca}}$ is the level of the root in $r$. Clearly, the larger this distance is, the lower the cohesiveness score for the result $r$ should be. Therefore, we compute the cohesiveness of a result subtree $r$ w.r.t. the data source $T$ as given as follows \cite{FengLWZ10}: 
\begin{equation}
\theta(r,T) = \frac{1}{\log_{\alpha}(d(r,T) + 1) + 1}
\label{eq:coherency}
\end{equation}
where $\alpha$ is the tuning parameter by which the user can trade off between the similarity of the candidate queries and the cohesiveness of the results. If we set $\alpha$ to a larger value, the sensitivity to the cohesiveness of the results gets smaller. That is, the total score $\sigma$ of a result $r$ is more dependent on the similarity of the query $\lambda(q_0, q^\prime)$ than its cohesiveness.
     
\begin{example}. Consider the candidate keyword queries $\mathcal{Q}$ in Example 7. The result subtrees of these queries are illustrated in Fig. \ref{fig:Coherency}. Using $\alpha = 4$, we compute the cohesiveness of the results as follows:\\$d(r_1,T) = 7,  \theta(r_1,T) = \frac{1}{2.5}  = 0.4$,\\ 
$d(r_2,T) = 10, \theta(r_2,T) = \frac{1}{2.72} = 0.36$,\\ 
$d(r_3,T) = 8,  \theta(r_3,T) = \frac{1}{2.58} = 0.38$,\\ and 
$d(r_4,T) = 11, \theta(r_4,T) = \frac{1}{2.79} = 0.35$. 
\end{example}   

\subsection{Processing of Candidate Queries}
\label{sec:pruning}
A na\"ive approach to processing the no-but-semantic-match query $q_0$ first generates the candidate queries $\mathcal{Q}$ and then computes all semantically related results $r\in \mathcal{R}$ by executing the queries $q^\prime\in \mathcal{Q}$ against the data source $T$. 
Then it applies Eq. \ref{eq:querysimilarity} to the results $\mathcal{R}$ to establish the similarity of the corresponding candidate query $q^\prime$ to the original query $q_0$ and determines the cohesiveness in the data source $T$ according to Eq. \ref{eq:coherency}. 
The results are then sorted based on their total score $\sigma$ (Eq. \ref{eq:score}) to obtain the top-k ranked (Eq. \ref{eq:rank}) semantically related results.
 
Assume that the initial query $q_0=\{k_1,...,k_n\}$ has the no-but-semantic-match problem for all $k_i\in q_0$, $|\mathcal{Q}|$ is the maximal number of candidate queries produced for $q_0$, $d$ is the depth of the tree $T$, $|S|$ and $|S_1|$ are the maximal and minimal sizes of the inverted keyword lists for the semantic counterparts $k_i^\prime$, respectively and $|\mathcal{R}|$ is the maximal number of semantically related results for the candidate queries in $\mathcal{Q}$, then the complexity of the naive approach becomes $|\mathcal{Q}| \times nd|S_1|\log|S| + |\mathcal{R}| \log|\mathcal{R}|$. However, both $|\mathcal{Q}|$ and $|\mathcal{R}|$ could be potentially large, which makes the na\"ive approach impractical. We propose two efficient pruning techniques, called the \emph{inter-query} pruning and \emph{intra-query} pruning to significantly reduce the sizes of $\mathcal{Q}$ and $\mathcal{R}$, respectively.

\subsubsection{Inter-Query Pruning}
\label{sec:interquerypruning}

Assume the candidate queries $\mathcal{Q}=\{q_1, q_2,...,q_l, q_{l+1}, ...$ $q_{|\mathcal{Q}|}\}$ are sorted based on their similarities $\lambda(q_0, q_i)$ to the original query $q_0$. 
We obtain the following lemma. 

\begin{lemma}
\label{lem:InterPruning}
Assume $R^*$ is the $k$ results of $\mathcal{Q}^\prime=\{q_1, q_2,...$ $,q_l\}$ and $\sigma^{min}$ is the min-score of the results $R^*$. Then, we can stop processing the rest of the candidate queries $\mathcal{Q}^{\prime\prime}=\{q_{l+1}, ...q_{|\mathcal{Q}|}\}$ if $\lambda(q_0,q_{l+1})<\sigma^{min}$.
\end{lemma}
\begin{proof}
Assume that $r$ is a min-scored result in $R^*$ for $q^\prime\in\mathcal{Q}^\prime$, i.e., $\sigma^{min}=\sigma(r, q^\prime, T)$ and $r^\prime$ is a result of the candidate query $q_{l+1}$ whose similarity score is higher than any result of the queries $\mathcal{Q}^{\prime\prime}=\{q_{l+1}, ...q_{|\mathcal{Q}|}\}$. Now, assume that $\sigma(r^\prime,q_{l+1},T)>\sigma(r, q^\prime,T)$. We prove that this can not happen if $\lambda(q_0,q_{l+1})<\sigma^{min}$. To be scored higher than $r$, $r^\prime$ must satisfy the following: $\lambda(q_0,q_{l+1})>\frac{\sigma^{min}}{\theta(r^\prime, T)}$. However, the highest possible value of $\theta$ for any result in $T$ is 1. By putting this into the above, we get $\lambda(q_0,q_{l+1})>\sigma^{min}$, which contradicts the assumption. Therefore, $R^*$ consists of the top-$k$ semantically related results according to Def. \ref{de:topk} whose ranks are $\le k$.
\end{proof}

\begin{example}. Consider the candidate queries $\mathcal{Q}$ given in Example 7. 
If we want to present the top-1 result to the user, and after processing the candidate queries up to $q_4$ we get $\sigma^{min} = 0.91 \times 0.4 = 0.36$, then we do not need to process $q_5$ as $\lambda(q_0,q_5) = 0.31 < \sigma^{min}$. 
Thus, from $q_5$ to the end of the list of $\mathcal{Q}$, no queries can score higher than $\sigma^{min}$ and therefore, we can stop processing them. 
\end{example} 

\subsubsection{Intra-Query Pruning}
\label{sec:intraquerypruning}
Although the \emph{inter-query pruning} technique does not execute all of the candidate queries $q^\prime\in\mathcal{Q}$ against $T$, it employs the pruning technique only in the first phase of the framework. That is, once we start processing a candidate query $q^\prime$, we compute all of its results. 
Consider the candidate queries $\mathcal{Q}$ of $q_0$ given in Example 7 and assume that the user requests only the top-1 result for $q_0$. Also, assume that the candidate queries in $\mathcal{Q}$ are sorted based on their similarities with $q_0$ and we have already processed the candidate queries from $q_1$ to $q_5$. The current top-1 result is $r_1$ (see in Fig. \ref{fig:Coherency}) and $\sigma^{min}$ is $0.36$. Now, while processing the candidate query $q_3$, we can discard the result $r_3$ of $q_3$ (as shown in Fig. \ref{fig:Coherency}) while generating it. That is, while reading through the keyword inverted lists $S_{Jack}$, $S_{fullprofessor}$ and $S_{course}$ for $q_3$, we can partially compute $r_3$ consisting of the keywords $\{Jack, full$ $professor\}$ for $q_3$ only (as highlighted in Fig. \ref{fig:Coherency}), denoted by $r_3^p$, and compare the score 0.34 of $r_3^p$ with the current $\sigma^{min}$, we can decide that the complete $r_3$ consisting of keywords $\{Jack, full$ $professor, course\}$, denoted by $r_3^c$, can never outrank $r_1$ as $\theta(r_3^{c}, T)\le \theta(r_3^{p}, T)$. The same occurs while computing $r_4$ of $q_4$ and we can discard $r_4$ before generating the ultimate result. We call the above query pruning technique as the \emph{intra-query pruning}. 

\begin{algorithm}[tb]
\SetKwInOut{Input}{Input}
\SetKwInOut{Output}{Output}
\Input{User Query $q_0$, Tuning Parameter $\alpha$, Keyword Inverted Lists $\mathcal{S}$ }
\Output{Top-k Semantically Related Results $\mathcal{R}^*$}
 
 $\mathcal{Q}\gets generateCandidateQueries(q_0)$;\par	
 
 \While{$q^\prime\gets\mathcal{Q}\cdot$getNext()$\neq null$}{
	$q^\prime\cdot$sim$\gets\lambda(q_0, q^\prime)$\tcp*{\small{according to Eq.\ref{eq:querysimilarity} }}\par
	}
 $\mathcal{Q}\gets sortCandidates$($\mathcal{Q}$)\tcp*{\small{based on $q^\prime\cdot$sim}} \par

 $\mathcal{R}^*\gets null$ \tcp*{\small{$\mathcal{R}^*$ is a min heap}} \par 
 $\sigma^{min}\gets$MAXVAL \tcp*{\small{max value}}\par

\While{$q^\prime\gets\mathcal{Q}\cdot$getNext()$\neq null$}
{

	\If{$\mathcal{R}^*\cdot getSize()=k$ and $q^\prime\cdot sim<\sigma^{min}$}
		{
			\textbf{break} \tcp*{\small{inter-query pruning}}\par			
		} 
		
	\If{$\mathcal{R}^*\cdot getSize()=k$}
	{
		$root\gets \mathcal{R}^*\cdot root(); \sigma^{min}\gets root\cdot score$;\par
	}
	
	$\mathcal{S}^\prime\gets \{\}$;\par
	\ForEach{$k_i\in q^\prime$}
		{
		  $S_i\gets retriveKeywordInvertedList(k_i, \mathcal{S})$;\par
			$\mathcal{S}^\prime\gets \mathcal{S}^\prime\cup S_i$;\par
		}
	
	$\mathcal{R}^*\gets processQuery(\mathcal{R}^*, \sigma^{min}, \alpha, q^\prime, \mathcal{S}^\prime)$;\par 
	
}
\Return {$\mathcal{R^*}$}
 \caption{The Framework}
\label{alg:queryprocessingalgorithm}
\end{algorithm}
\subsubsection{The Framework}
Algorithm \ref{alg:queryprocessingalgorithm} presents the framework for processing the no-but-semantic-match query $q_0$ submitted by the user. First, it generates the candidate queries $\mathcal{Q}$ for $q_0$ as explained in Section \ref{sec:generatingcandidatequeries}, shown on line 1. The lines 2-4 compute the similarity between the candidate queries $\mathcal{Q}$ and the user query $q_0$ as explained in Section \ref{sec:candidatequerysimilarity} and sort them. Then, a min-heap is initialized with $\mathcal{R}^*$ to $null$ and the min-score $\sigma^{min}$ to $MAXVAL$ in lines 5-6. In lines 8-9, we stop processing the candidate queries in $\mathcal{Q}$ as soon as we find a query $q^\prime\in\mathcal{Q}$ if $|\mathcal{R}^*|=k$ and $q^\prime.sim<\sigma^{min}$. Otherwise, if $|\mathcal{R}^*|=k$, we update $\sigma^{min}$ by reading the root entry of the heap $\mathcal{R}^*$ and adding its score $root.score$ to $\sigma^{min}$ in lines 10-11. Then, for each keyword $k_i \in q^\prime$ we retrieve their corresponding inverted lists and pass it to the $processQuery$ method (which is explained in detail in the following section) in line 16 to retrieve the results of $q^\prime$ and insert the eligible results into $\mathcal{R}^*$. 

\subsubsection{The processQuery Method}
\label{sec:processquerymethod}
In order to process the no-but-semantic-match query $q_0$, we need to execute each candidate keyword query $q^\prime\in\mathcal{Q}$ against the data source $T$ in the $processQuery$ method. There are two benchmark algorithms in the literature to compute the keyword query results on XML data $T$ as given as follows: (a) scan eager \cite{XuP05} and (b) anchor based \cite{SunCG07} algorithms. However, these two benchmark algorithms are not readily available to implement our $processQuery$ method. These algorithms only find the root of the subtree results in $T$, but ignore the distribution of the keyword match nodes in the subtree. To implement our $processQuery$ method with these benchmark algorithms, we need to address the following issues which are specific to our problem: (a) finding the tightest nodes under the confirmed SLCA root $v_{slca}$ and (b) scoring the result partially based on its candidate query similarity and cohesiveness to apply \emph{intra-query} pruning. We propose two techniques to implement the $processQuery$ method based on the benchmark algorithms as follows: 

\begin{enumerate}
\item \textbf{S}can \textbf{E}ager based \textbf{Q}uery \textbf{P}rocessing (SE-QP) and
\item \textbf{AN}chor based \textbf{Q}uery \textbf{P}rocessing (AN-QP). 
\end{enumerate}

\textbf{SE-QP Algorithm.} Like scan-eager algorithm\cite{XuP05}, SE-QP firstly sorts the inverted lists of the keywords in $q^\prime$. Then, it picks a match node $m_1$ from the shortest inverted list $S_1$ and then, finds the closest match nodes to $m_1$ from other lists to compute the result root $v_{slca}$. However, the tightest subtree result computation and result scoring is delayed until we can confirm that this $v_{slca}$ can be an actual SLCA node. Therefore, the cursor of each of the inverted list is retained until we decide that this $v_{slca}$ cannot be an ancestor of any other result roots. Once we confirm that this $v_{slca}$ is an actual SLCA node, we compute the tightest subtree result for it and score the result. While scoring the result, we also apply \emph{intra-query} pruning here. Then, we advance the cursors of all of the lists and continue the above steps until we access all nodes in the list $S_1$.   

\begin{algorithm}[tb]
\SetKwInOut{Input}{Input}
\SetKwInOut{Output}{Output}
\Input{$\mathcal{R}^*, \sigma^{min}, \alpha, q^\prime, \mathcal{S}^\prime$}
\Output{Result Set:$\mathcal{R}^*$}

$sortLists(\mathcal{S}^\prime); r \gets null; v_{slca} \gets null$;\par 
\While{ $m_1 \gets getNext(S_1) \neq null$}
{ 
	$m_i \gets closest(m_1,S_i), \forall i \in [2,n]$;\par
	$v^{u}_{slca} \gets lca(m_1,...,m_n)$;
		
	\If{$r \neq null$ \textbf{and} $v_{slca} \not\prec_a v^{u}_{slca}$}
	{		
			
	 \For{$i = 1 \to  n$}
	 {
			$m^{li}_i \gets getTight(S_i,r.cursor_i)$;\par
			$d_i \gets getDist(v_{slca},m^{li}_i)$; $d\gets d+d_i$;\par 
			$r\cdot score\gets q^\prime\cdot sim\times\frac{1}{log_\alpha(d+1)+1}$;\par		
		
			\If{$r\cdot score<\sigma^{min}$}
			{
				\small{stop reading lists and jump to line 14.}				
			}
		}	
		$r \gets (v_{slca},\{m^{li}_i,\forall i \in [1,n]\})$;\par
		update top-k list $\mathcal{R}^*$ with $r$;\par			
	}
	$r \gets null$;\par
	\If{$v^{u}_{slca} \not\prec_a v_{slca}$ } 
	{
		$v_{slca} \gets v^{u}_{slca}; r.add(S_i.cursor,\forall i \in [1,n])$;\par
	} 
}
\If{$v^{u}_{slca} \not\prec_a v_{slca}$}
{
  $r \gets (v^u_{slca},\{m^{li}_i,\forall i \in [1,n]\})$;\par
	score $r$ and update $\mathcal{R}^*$ with $r$ if $r\cdot score>\sigma^{min}$; 
}

\Return {$\mathcal{R}^*$}
\caption{SE-QP}
\label{alg:SE-QP}
\end{algorithm}    

The SE-QP query processing technique is pseudocoded in Algorithm \ref{alg:SE-QP}. In line 1, we sort the inverted list of all keywords in $q^\prime$ and do the initialization. In lines 2-3, it finds the match node $m_1$ from the shortest list $S_1$ and the match nodes $m_i$ from other lists. In line 4, we find the potential root result $v^{u}_{slca}$ which then should be confirmed as an actual SLCA node. In line 5, we check $v^{u}_{slca}$ with the previous result root $v_{slca}$. If $v_{slca}$ is not an ancestor for $v^{u}_{slca}$, denoted as $v_{slca} \not\prec_a v^{u}_{slca}$, $v_{slca}$ is confirmed as the SLCA result root. Then, the corresponding tightest subtree result for $v_{slca}$ is retrieved. To do so, we scan each list $S_i$ by moving its cursor $r.cursor_i$ backward and forward to find the closest match nodes under $v_{slca}$, which is implemented in function $getTight$ of line 7. For each closest match node $m^{li}_i$, the distance of $m^{li}_i$ with $v_{slca}$ is computed by $getDist$ function in line 8 and the score of the result $r.score$ is computed partially in lines 8-9. We stop scanning other lists if the partial score cannot beat $\sigma^{min}$ (intra-query pruning) and jump to line 14, which is given in lines 10-11. Otherwise, we keep scanning all lists to compute the tightest subtree result and the ultimate score of $r$ for the $v_{slca}$. We update the min heap $\mathcal{R}^*$ by this result $r$ in line 13. Now, we update $v_{slca}$ with the current result root $v^{u}_{slca}$ if $v^{u}_{slca}\not\prec_a v_{slca}$ in lines 15-16. In lines 17-19, if the last result root node is an actual SLCA, the similar steps are conducted to score its tightest subtree result and if promising, is used to update $\mathcal{R}^*$.

\textbf{AN-QP Algorithm.}
Like scan-eager algorithm\cite{XuP05}, SE-QP performs worse when the inverted lists have similar sizes (e.g., match node distribution). Also, it incurs many redundant computations when the data distribution is skewed in $T$. For example, if the match nodes are mostly distributed in one part of the XML tree in an inverted list, SE-QP reads all the nodes in $S_1$ and computes their LCAs to finalize the corresponding SLCAs. However, lots of these nodes can be skipped because they are far from the nodes in other inverted lists and cannot create SLCA nodes.

In order to skip the non-promising match nodes, like \cite{SunCG07}, AN-QP considers only the \emph{anchor} match nodes for computing SLCA nodes. A set of match nodes $M=\{m_1, ...m_n\}$ for $q^\prime$ is said to be \emph{anchored} by a match node $m_a\in M$ if for each $m_i\in M\setminus \{m_a\}$, $m_i=closest(m_a, S_i)$, where $closest(m_a, S_i)$ returns the match nodes in the list $S_i$ which is closest to the node $m_a$\cite{SunCG07}. Unlike SE-QP, the \emph{anchor} match node $m_a$ is picked from among inverted lists (not necessarily from the shortest one) so that it can maximize the skipping of redundant computations. Similar to SE-QP, AN-QP first extracts the SLCA result root $v_{slca}$ and thereafter, finds the tightest subtree result under $v_{slca}$ and partially score the result to apply \emph{intra-query} pruning.

\begin{algorithm}[tb]
\SetKwInOut{Input}{Input}
\SetKwInOut{Output}{Output}
 \Input{$\mathcal{R}^*, \sigma^{min}, \alpha, q^\prime, \mathcal{S}^\prime$}
 \Output{Result Set:$\mathcal{R}^*$}
$r \gets null;v_{slca} \gets null;$\par 
$m_a \gets getAnchor(\{getNext(S_i), \forall i\in[1,n]\})$;\par

\While{ $m_a \neq null$ }
{	
  $m_i \gets closest(m_a,S_i), \forall i \in [1,n] \& i \neq a$;\par
	$v^{u}_{slca} \gets lca(m_1,...,m_n)$;\par
	
	\If{$r \neq null$ \textbf{and} $v_{slca} \not\prec_a v^{u}_{slca}$} 
	{		
			\For{$i = 1 \to  n$} 
			{
				$m^{li}_i \gets getTight(S_i,r.cursor_i)$;\par								
				$d_i \gets getDist(v_{slca},m^{li}_i)$; $d\gets d+d_i$\par 
				$r\cdot score\gets q^\prime\cdot sim\times\frac{1}{log_\alpha(d+1)+1}$;\par		
		
				\If{$r\cdot score<\sigma^{min}$}
				{
					stop reading lists and jump to line 15.					
				}
			}			
				$r \gets (v_{slca},\{m^{li}_i,\forall i \in [1,n]\})$;\par				
				update top-k list $\mathcal{R}^*$ with $r$; \par		
		}
		$r \gets null$;\par
	
		\If{$v^{u}_{slca} \not\prec_a v_{slca}$}
		{
			$v_{slca} \gets v^{u}_{slca}$; $r.add(S_i.cursor), \forall i \in [1,n]$;\par 
		}	
	
	$m_a \gets getAnchor(\{getNext(S_i), \forall i \in [1,n]\})$;\par

}
\If{$v^{u}_{slca} \not\prec_a v_{slca}$}
{
	$r \gets (v^u_{slca},\{m^{li}_i,\forall i \in [1,n]\})$;\par
	score $r$ and update $\mathcal{R}^*$ with $r$ if $r\cdot score>\sigma^{min}$; 
}	  

\Return {$\mathcal{R}^*$}
\caption{AN-QP}
\label{alg:AN-QP}
\end{algorithm}

\begin{figure*}

\centering
        \subfloat[Plan $\mathcal{P}_1$, $c(\mathcal{P}_1)=38$]{\label{LabelA}
  \includegraphics[width=1.7in,height=1in]{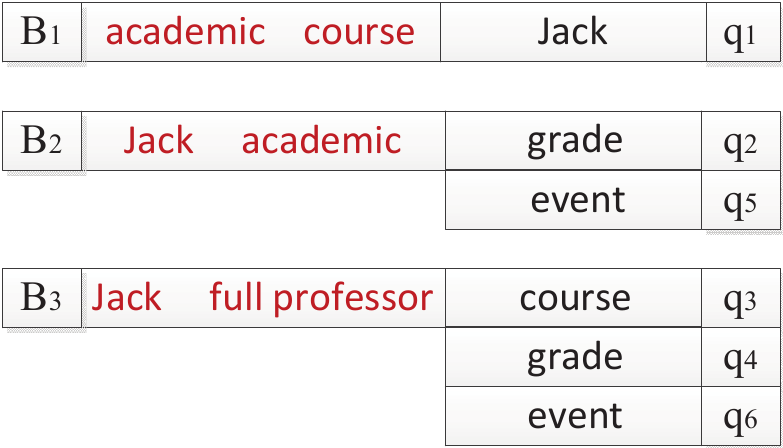}%
}
\hspace{3em}
\subfloat[Plan $\mathcal{P}_2$, $c(\mathcal{P}_2)=39$]{\label{labelB}
  \includegraphics[width=1.7in,height=1in]{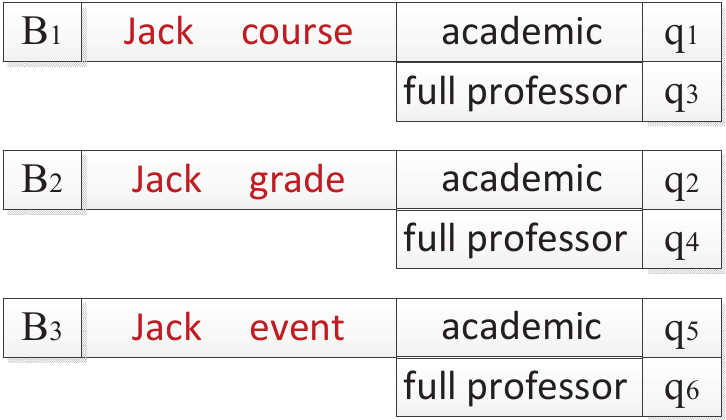}%
}
\hspace{3em}
\subfloat[Plan $\mathcal{P}_3$, $c(\mathcal{P}_3)=34$]{\label{labelC}%
  \includegraphics[width=1.7in,height=1in]{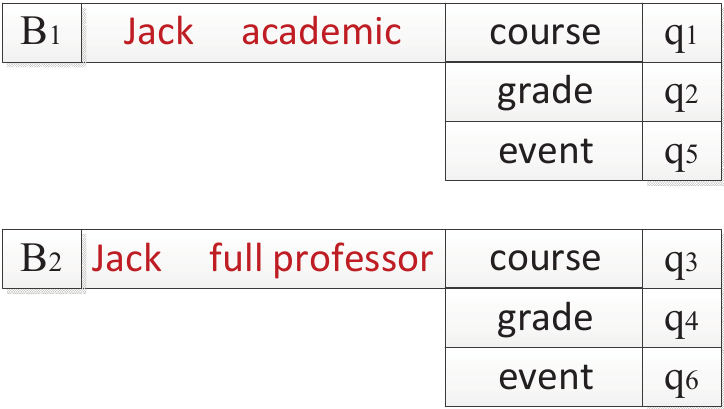}%
}		
        
				\caption{A part of the possible execution plans for processing the candidate queries in $\mathcal{Q}$.}
				\label{fig:plans}
\vspace{-3ex} 
\end{figure*}

The AN-QP technique is pseudocoded in Algorithm \ref{alg:AN-QP}. Line 2 finds the anchor node $m_a$ from the the inverted lists $\{S_1,...,S_n\}$, which is implemented in function $getAnchor$. The potential result root $v^{u}_{slca}$ is computed after finding the closest nodes $m_i\in S_i$ to $m_a$, $\forall i\in [1, n]$ and $i\neq a$ in lines 4-5. Here, if $v_{slca} \not\prec_a v^{u}_{slca}$, $v_{slca}$ is confirmed as the SLCA result root as given in line 6. Therefore, we retrieve the closest match nodes $m^{li}_i, \forall i\in [1, n]$ under $v_{slca}$. Similar to SE-QP,	the function $getTight$ of line 8, scans each list $S_i$ by moving its cursor $r.cursor_i$ backward and forward to compute the tightest subtree result for $v_{slca}$. Then, we find the distance of each $m^{li}_i$ with $v_{slca}$ by the function $getDist$ and add it to the total distance of $r$ in line 9. Then, we compute the score of the result $r.score$ partially in line 10. Here, we apply intra-query pruning if the partial score cannot beat $\sigma^{min}$ and jump to line 15, which is given in lines 11-12. Otherwise, we compute the ultimate score of the result by scanning all lists to retrieve the tightest subtree result under $v_{slca}$ and update $r$ with its corresponding data in line 13. In line 14, the promising result $r$ is inserted into $\mathcal{R}^*$. The current result root $v^{u}_{slca}$ is saved into $v_{slca}$, if $v^{u}_{slca}\not\prec_a v_{slca}$ in lines 16-17. We update the anchor node $m_a$ in line 18 for the next iteration. We check the last result root node and retrieve its tightest subtree result if it is an actual SLCA node in lines 19-20. Finally, we score it and update $\mathcal{R}^*$ with it if $r\cdot score>\sigma^{min}$ in line 21.  
                                           
\section{Batch Processing}
\label{sec:multiple}
This section investigates a more efficient method for processing the candidate queries $\mathcal{Q}$ obtained from the initial query $q_0$ with the no-but-semantic-match problem. As the candidate queries in $\mathcal{Q}$ are generated by replacing the keywords of the initial query, they usually share a subset of keywords.  It is possible to compute the results of these shared subsets of keywords among the queries in $\mathcal{Q}$ and then merge the results of the shared part with the exclusive part of each query $q^\prime \in \mathcal{Q}$ \cite{YaoLL013}. However, two challenges arise here: (a) finding the groups of queries, called \emph{batches}, which share a subset of keywords among them and can be executed efficiently; and (b) finding the tightest SLCA results for the shared part which can be ultimately merged with the exclusive part of each candidate query in a batch. 

\subsection{Constructing the Candidate Query Batch}
Given two candidate queries $q_1,q_2\in \mathcal{Q}$, assume that $\mathcal{K}^s(q_1, q_2)$ denotes the set of keywords shared by them, i.e., $\mathcal{K}^s(q_1, q_2)\subseteq q_1, q_2$. We can achieve the best performance by putting $q_1$ and $q_2$ in a batch if they share the maximal set of keywords, i.e., $|\mathcal{K}^s(q_1,q_2)| = |q_1|-1 = |q_2|-1$ \cite{YaoLL013}. A candidate query batch is defined as given below:

\begin{definition} A candidate query batch, denoted by $\mathcal{B}$, is a subset of $\mathcal{Q}$ such that the following conditions hold: (a) $\forall{q_1,q_2 \in \mathcal{B}}$, $|\mathcal{K}^s(q_1,q_2)| = |q_1|-1 = |q_2|-1$; (b) $\forall{q_1,q_2,q_3 \in \mathcal{B}}$, $\mathcal{K}^s(q_1,q_2)=\mathcal{K}^s(q_2,q_3)=\mathcal{K}^s(q_1,q_3)$; and (c) $1 \leq |\mathcal{B}| \leq |\mathcal{Q}|$.
\end{definition}

To execute the candidate queries in $\mathcal{Q}$, we have to construct the set of batches that can cover all queries in $\mathcal{Q}$. We call the set of candidate query batches that cover the queries in $\mathcal{Q}$ an \emph{execution plan}, which is defined below:     

\begin{definition} An execution plan, denoted by $\mathcal{P}$, is a set of candidate query batches such that: (a) $\mathcal{Q} = \bigcup_{i = 1}^{|\mathcal{P}|} \mathcal{B}_i\in\mathcal{P}$ and (b) $\forall \mathcal{B}_1,\mathcal{B}_2 \in \mathcal{P}$, $\mathcal{B}_1\cap \mathcal{B}_2 = \emptyset$.
\end{definition}

An execution plan $\mathcal{P}$ has its evaluation cost which is the summation of the execution costs of its constituent batches. The execution cost of a batch $\mathcal{B}$ directly depends on the inverted keyword lists which have to be accessed. Assume that $\mathcal{K}^u$ is the set of keywords of all candidate queries in a batch $\mathcal{B}$ that has not been covered by $\mathcal{K}^s$, i.e., $\bigcup_{\forall q_1\in\mathcal{B}} q_1\setminus\mathcal{K}^s$. We estimate the cost of executing $\mathcal{B}$, denoted by $c(\mathcal{B})$ as:

\begin{equation}
c(\mathcal{B}) = \sum\limits_{k_1 \in \mathcal{K}^s} {|S_{k_1}|} + \sum\limits_{k_2 \in \mathcal{K}^u} {(|min(S_{\mathcal{K}^s})|+|S_{k_2}|)}
\end{equation}
where $min(S_{\mathcal{K}^s})$ returns the shortest inverted keyword list size among the keywords in $\mathcal{K}^s$. However, there exist many plans for $\mathcal{Q}$ as shown in Fig. \ref{fig:plans}. The optimal plan has the least cost. Discovering this optimal plan is a combinatorial optimization problem as there are many ways of constructing the candidate query batches from $\mathcal{Q}$. Here, we propose a \emph{greedy approach} for discovering a sub-optimal plan which consists of the following steps: 
(a) retrieve the topmost similar query $q_1\in\mathcal{Q}$ to $q_0$; (b) construct all plausible batches for $q_1$ as follows: (i) remove a keyword $k_1\in q_1$ and construct $\mathcal{K}^s$ as $q_1\setminus k_1$; (ii) retrieve all $q_2\in\mathcal{Q}$ such that $\mathcal{K}^s\subset q_2$; (iii) insert $q_1$ and all $q_2$ into a plausible batch $\mathcal{B}_1$; (c) make the batch $\mathcal{B}_1$ as the actual batch that has the least unit cost $\frac{c(\mathcal{B}_1)}{|\mathcal{B}_1|}$; (d) remove all queries $\mathcal{B}_1$ from $\mathcal{Q}$; and (e) repeat the above steps until $\mathcal{Q}$ is empty.

\begin{example}. Consider the candidate queries $\mathcal{Q}$ presented in Example 7 and the sizes of the inverted lists as follows: $|S_{Jack}|=3$, $|S_{academic}|=1$, $|S_{full professor}|=1$, $|S_{course}|=6$, $|S_{grade}|=2$, $|S_{event}|=2$.\\
We start with the topmost query $q_1 \in \mathcal{Q}$ and remove one keyword at a time from $q_1$ to create the plausible candidate query batches as follows:\\  
$\mathcal{B}_{1.1}$ $(\mathcal{K}^s$ $=\{Jack,academic\},$ $\mathcal{K}^u$=$\{course,grade,event\})$\\
$\mathcal{B}_{1.2}$ $(\mathcal{K}^s=\{Jack,course\}$,$\mathcal{K}^u$=$\{academic,fullprofessor\})$\\   
$\mathcal{B}_{1.3}$ $(\mathcal{K}^s$ $=\{academic,course\},$ $\mathcal{K}^u=\{Jack\})$\\
We estimate the costs of the candidate query batches as follows:
$\frac{c(\mathcal{B}_{1.1})}{|\mathcal{B}_{1.1}|}=5.6 ,\frac{c(\mathcal{B}_{1.2})}{|\mathcal{B}_{1.2}|}=8.5, \frac{c(\mathcal{B}_{1.3})}{|\mathcal{B}_{1.3}|}=11$\\  
Therefore, $\mathcal{B}_{1.1}$ is the initial candidate query batch. After excluding the queries in $\mathcal{B}_{1.1}$ from $\mathcal{Q}$, the next topmost similar candidate query $q_3$ is considered and the following plausible batches are constructed:\\ 
$\mathcal{B}_{2.1}$ $(\mathcal{K}^s=\{Jack,full professor\}$,$\mathcal{K}^u$=$\{course,grade,event\})$\\
$\mathcal{B}_{2.2}(\mathcal{K}^s=\{Jack,course\},\mathcal{K}^u=\{full professor\})$\\
$\mathcal{B}_{2.3}(\mathcal{K}^s=\{full professor,course\},\mathcal{K}^u=\{Jack\})$\\
The costs of the above plausible batches are as follows:\\
$\frac{c(\mathcal{B}_{2.1})}{|\mathcal{B}_{2.1}|}=5.6 ,\frac{c(\mathcal{B}_{2.2})}{|\mathcal{B}_{2.2}|}=13, \frac{c(\mathcal{B}_{2.3})}{|\mathcal{B}_{2.3}|}=11$\\
Hence, $\mathcal{B}_{2.1}$ is the next candidate query batch. Now, if we exclude all queries in $\mathcal{B}_{2.1}$ from $\mathcal{Q}$, $\mathcal{Q}$ becomes empty and the process stops. Finally, plan $\mathcal{P}_3$ is our execution plan as shown in Fig. \ref{fig:plans}(c), which is sub-optimal.
\end{example} 

\subsection{Processing the Candidate Query Batch}

To process a candidate query batch $\mathcal{B}$, we need to compute the results of the shared part $\mathcal{K}^s$ first. Then, we need to merge these shared part results with the non-shared keywords $\mathcal{K}^u$ for $q^{\prime\prime}\in\mathcal{B}$. However, computing the shared part results that can be merged with the unshared part is a non-trivial problem. This is because, the SLCA result roots of the shared part $\mathcal{K}^s$ do not guarantee to be the SLCA result roots for $q^{\prime\prime}\in\mathcal{B}$. The result roots of $q^{\prime\prime}$ may ascend to higher levels in the tree $T$ when the shared part results are merged with the unshared part $q^{\prime\prime}\setminus\mathcal{K}^s$.

Consider the potential shared part results of $\{Jack,$ $academic\}$ for the batch $\mathcal{B}_1$ of plan $\mathcal{P}_3$ as presented in Fig. $\ref{fig:SharingResult}$. Now, when we merge these results with the keyword $course$ for $q_1\in\mathcal{B}_1$, $r^s_1$ contributes to the final SLCA result root, which is $v^{r_1}_{slca}=[0.2]$ (see in Fig. \ref{fig:Coherency}). However, when we merge these potential shared part results with the keyword $grade$ for $q_2\in\mathcal{B}_1$, $r^s_3$ contributes to the final SLCA result and the result root ascends to $v^{r_3}_{slca}=[0]$ (see in Fig. \ref{fig:Coherency}). This indicates that we need to retain $r^s_1$ as well as $r^s_3$ as the actual shared part results, though the root of $r^s_3$ is not the SLCA result root of the shared part $\{Jack, academic\}$, but the root of $r^s_1$ is. Here, we do not need to retain $r^s_2$ as $d(r^s_{2},T) = 7>d(r^s_{3},T) = 6$ and both $r^s_2$ and $r^s_3$ share the same root.

\begin{figure}
\centering
             \includegraphics[height=1.75in,width = 3in]{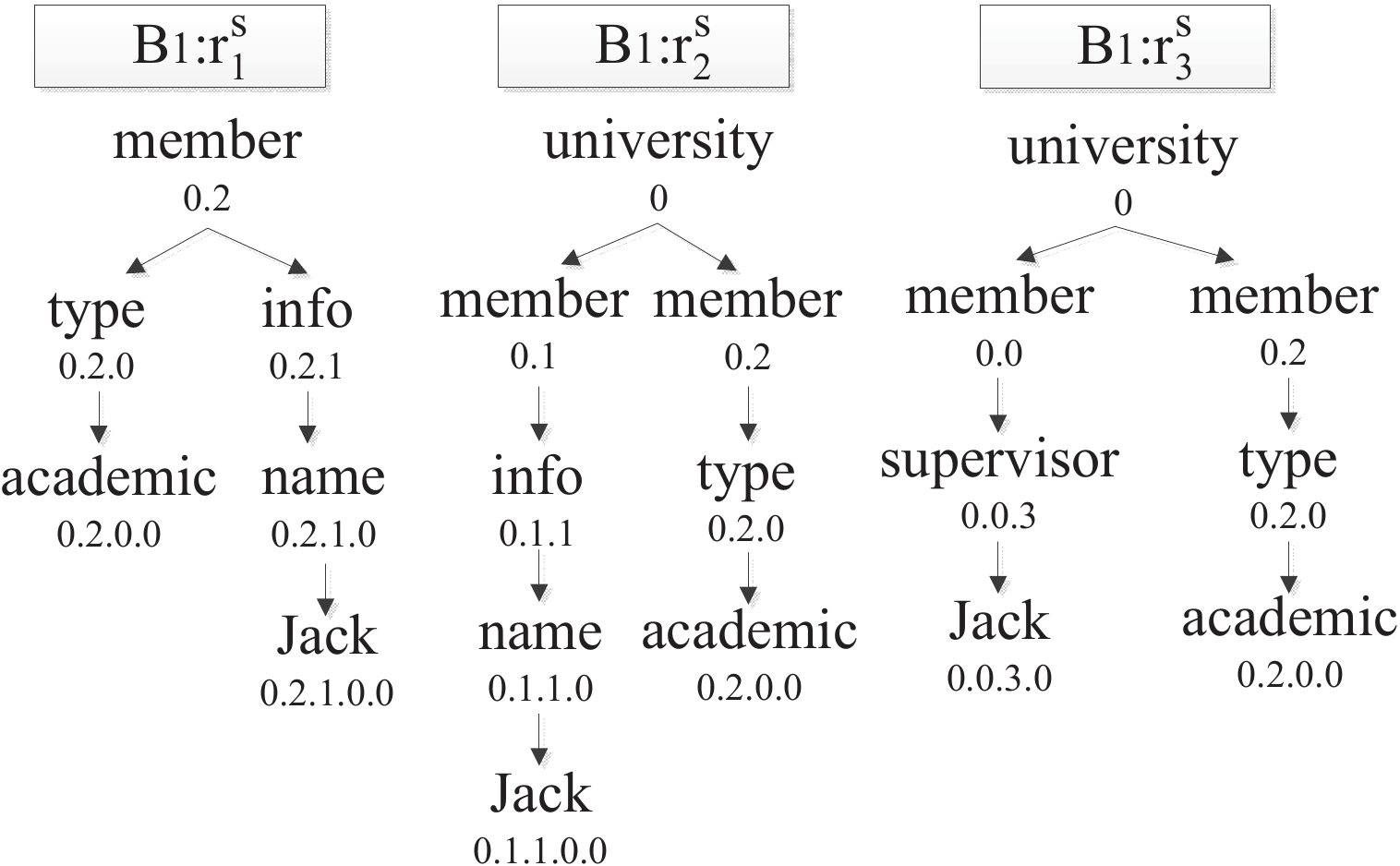}
						  	\caption{Shared part processing for query batch $\mathcal{B}_1$ of $\mathcal{P}_3$.}
							\label{fig:SharingResult}	
	\vspace{-3ex} 
\end{figure}

\begin{lemma}
Assume $q^{\prime\prime}\in \mathcal{B}$ and $v^s_{slca}$ is the SLCA result root of the shared part of $\mathcal{B}$. Then, the tightest match nodes of the final result of $q^{\prime\prime}$ are under $v^s_{slca}$ or under one of the ancestors of $v^s_{slca}$.         
\end{lemma}

Therefore, to process a candidate query batch shared part $\mathcal{K}^s$, we find the closest match nodes under the shared part result root $v^s_{slca}$ as well as under all of its ancestors.

\subsubsection{Inter and Intra-Batch Pruning}
Assume the candidate queries in a batch $\mathcal{B}_i$ are sorted based on their similarities to $q_0$ as follows: $\{q_1,q_2,...,q_{l-1}$ $,q_l,...,q_{|\mathcal{B}_i|}\}$. Also, $\mathcal{R}^s_{\mathcal{B}_i}$ is the set of shared part results for $\mathcal{B}_i$. Then, the lower bound of the overall \emph{distances} of the results $\mathcal{R}_{\mathcal{B}_i}$ of $\mathcal{B}_i$, denoted by $\underline{d}(\mathcal{R}_{\mathcal{B}_i},T)$, is $min\{d(r^s\in \mathcal{R}^s_{\mathcal{B}_i},T)\}$. Now, the upper bound of the \emph{cohesiveness} of $\mathcal{R}_{\mathcal{B}_i}$, denoted by $\overline\theta(\mathcal{R}_{\mathcal{B}_i},T)$ is computed using $\underline{d}(\mathcal{R}_{\mathcal{B}_i},T)$ in Eq. 8. The upper bound of the $\sigma$ of a result $r_l$ of query $q_l\in \mathcal{B}_i$ is: 

\begin{equation}
\overline{\sigma}(r_l,q_l,T) = \lambda(q_0,q_l) \times \overline\theta(\mathcal{R}_{\mathcal{B}_i},T)
\label{eq:BatchPruning}
\end{equation}
       
Assume $\mathcal{B}^{\prime}_i= \{q_1,q_2,...,q_{l-1}\}$ and $\mathcal{R}^*$ is the $k$ results of $\mathcal{Q}^\prime = \{\mathcal{B}_1,...,\mathcal{B}_{i-1},\mathcal{B}_{i}^\prime\}$. Also, $\sigma^{min}$ is the min-score for $\mathcal{R}^*$. Then, we can stop processing the candidate queries $\{q_l,...,q_{|\mathcal{B}_i|}\}\in\mathcal{B}_i$ if $\overline{\sigma}(r_l,q_l,T) < \sigma^{min}$. We call the above pruning technique as \emph{intra-batch} pruning if $l>1$, otherwise, we call it \emph{inter-batch} pruning (prune the entire batch).

\begin{algorithm}[tb]
\SetKwInOut{Input}{Input}
\SetKwInOut{Output}{Output}
 \Input{User Query $q_0$, Tuning Parameter $\alpha$, Keyword Inverted Lists $\mathcal{S}$}
 \Output{Top-k Semantically Related Result $\mathcal{R^*}$ }

Repeat Lines 1-6 of Algorithm 1.\par
\While{ $q^\prime \gets getNext(\mathcal{Q}) \neq null$}
{
	\If{$\mathcal{R}^*\cdot getSize()=k$ \textbf{and} $q^\prime\cdot sim<\sigma^{min}$}
	{
		\textbf{break} \tcp*{\scriptsize{inter-query pruning}}\par		
	} 		
	\For{ i = 1 $\to$ n  }
	{
		$\mathcal{K}^s \gets q^\prime \setminus k_i$\par
		$\mathcal{B} \gets getBatch(\mathcal{Q},\mathcal{K}^s)$\par
		$\mathcal{K}^u \gets getUnique(\mathcal{B},\mathcal{K}^s)$\par		
		$\mathcal{B}.add(\mathcal{B}_i, c(\mathcal{K}^s,\mathcal{K}^u))$\par											
	}	
	$\mathcal{B}^* \gets getMinCost(\mathcal{B})$; $\mathcal{S}^\prime \gets \{\}$;\par
			
	\small{retrieve  Inverted Lists $\forall k_i \in \mathcal{B}^*.\mathcal{K}^s$ into $S^\prime$};\par	
	$\mathcal{R}^s \gets sharedPartComputation(\sigma^{min}, \alpha, q^{\prime\prime}, \mathcal{S}^\prime)$\par
	
	\While{ $q^{\prime\prime} \gets getNext(\mathcal{B}^*) \neq null$ }
	{
		$\overline{r\cdot score} \gets q^{\prime\prime}\cdot sim \times \overline{coh(\mathcal{R}^{s}, T)}$; \par
	  \If{$\mathcal{R}^*\cdot getSize()=k$ \textbf{and} $\overline{r\cdot score} \textless \sigma^{min}$}
			{\textbf{break}\tcp*{\scriptsize{inter and intra-batch pruning}}}
		\If{$\mathcal{R}^*\cdot getSize()=k$}
	  {
			$root\gets \mathcal{R}^*\cdot root(); \sigma^{min}\gets root\cdot score$;\par
		}
		$k^u \gets q^{\prime\prime} \setminus \mathcal{B}^*.\mathcal{K}^s$;\par
		$S_1 \gets retrieveKeywordInvertedList(k^u, \mathcal{S})$;\par
		$\mathcal{R}^* \gets mergeResults(\mathcal{R}^*, \mathcal{R}^s, \sigma^{min}, \alpha, q^{\prime\prime}, S_1)$;\par		
	}	
	$\mathcal{Q} \gets \mathcal{Q} \setminus \mathcal{B}^*$\par
 }
\Return{$\mathcal{R^*}$}
 
\caption{BA-QP}
\label{alg:multipleprocessing}
\end{algorithm}

\subsubsection{The Framework}
Algorithm \ref{alg:multipleprocessing} presents the framework for \textbf{ba}tch \textbf{q}uery \textbf{p}rocessing technique (BA-QP) of the no-but-semantic-match problem. Similar to first 6 lines of Algorithm \ref{alg:queryprocessingalgorithm}, it generates the candidate queries, measures their similarity, sorts them and does the initializations in line 1 to apply \emph{inter-query} pruning first. That is, at each step of iteration, we take the topmost similar query $q^\prime \in \mathcal{Q}$ in line 2 and stop processing the candidate queries in $\mathcal{Q}$ as soon as we find a query $q^\prime\in \mathcal{Q}$ such that $|\mathcal{R}^*|=k$ and $q^\prime.sim<\sigma^{min}$ as given in lines 3-4. 

Otherwise, we construct the sub-optimal candidate query batch $\mathcal{B}^*$ for $q^\prime$ as given in lines 5-10. In line 12, the shared part $\mathcal{K}^s$ of the candidate query batch $\mathcal{B}^*$ is processed and its results are stored in $\mathcal{R}^s$. Then, for each query $q^{\prime\prime} \in \mathcal{B}^*$, the upper bound of its actual results' score $\overline{r.score}$ is computed in line 14. In lines 15-16, we apply \emph{inter and intra-batch} pruning where we stop processing the candidate query batch $\mathcal{B}^*$  if $|\mathcal{R}^*|=k$ and $\overline{r.score}<\sigma^{min}$. 
Otherwise, if $|\mathcal{R}^*|=k$, we update $\sigma^{min}$ by reading the root entry of the heap $\mathcal{R}^*$ in lines 17-18. In lines 19-20, the unshared keyword part $k^u$ of the batch query $q^{\prime\prime}$ and the corresponding inverted list $S_1$ are retrieved. The procedure $mergeResults$ in line 21 generates the final results for $q^{\prime\prime}$ by merging shared part results $\mathcal{R}^s$ with the unshared part $k^u$ and inserts the promising results into $\mathcal{R}^*$. Finally, the queries $q^{\prime\prime} \in \mathcal{B}^*$ are excluded from the $\mathcal{Q}$ as pseudocoded in line 22. The above steps continue until $\mathcal{Q}$ becomes empty.

\begin{algorithm}[tb]
\SetKwInOut{Input}{Input}
\SetKwInOut{Output}{Output}
 \Input{$\sigma^{min}, \alpha, q^\prime, \mathcal{S}^\prime$}
 \Output{Shared Result $\mathcal{R}^s$ }
$\mathcal{R}^s\gets $null;\par
$m_a \gets getAnchor(\{getNext(S_i), \forall i\in[1,n-1]\})$;\par
\While{ $m_a \neq null$ }
{	
  $m_i \gets closest(m_a,S_i), \forall i \in [1,n-1] \text{ and } i \neq a$;\par
	$v_{slca} \gets lca(m_1,...,m_{n-1})$;\par 
	$\mathcal{A} \gets v_{slca} \cup getAncestors(v_{slca})$;\par
		\While{$v_{slca}^s \gets getNext(\mathcal{A}) \neq null$}
		{		  
			\If{$v_{slca}^s \notin R^{s}$}
			{			  
				\For{$i = 1 \to  n-1$ }
				{					
					$m^{li}_i \gets getTight(S_i,S_i.cursor)$;\par
					$\small{d_i \gets getDist(v_{slca}^s,m^{li}_i);r.d\gets r.d+d_i}$;\par
					$r\cdot score\gets q^\prime\cdot sim\times\frac{1}{log_\alpha(r.d+1)+1}$;\par
					\If{$r\cdot score<\sigma^{min}$}
				  {
					  stop reading lists, jump to line 6.					 
				  }
				}
				$r \gets (v^s_{slca},\{m^{li}_i,\forall i \in [1,n]\})$;\par
				$\mathcal{R}^s\cdot insert$($r$);\par				
			}		
		}					
			
	$m_a \gets getAnchor(\{getNext(S_i), \forall i \in [1,n-1]\})$;\par
}	  

\Return{$\mathcal{R}^s$}
 \caption{BA-QP (sharedPartComputation)}
\label{alg:sharedprocessing}
\end{algorithm}

The details of the $sharedPartComputation$ method in line 12 of Algorithm \ref{alg:multipleprocessing} is presented in Algorithm \ref{alg:sharedprocessing}. In lines 2-5, we compute the result root $v_{slca}$. Then, we compute all the ancestors of $v_{slca}$ which is implemented in the function $getAncestors$ and put them in the potential result roots $\mathcal{A}$ in line 6. For each potential result root $v_{slca}^s\in\mathcal{A}$, we check if it is already processed and is in memory $\mathcal{R}^s$ in line 8. If $v_{slca}^s\not\in \mathcal{R}^s$, this result has not been processed and thus, we compute the closest nodes under $v_{slca}^s$ in the inverted lists $S_i, \forall i \in [1,n-1]$, their distance to $v_{slca}$, and their score (lines 9-12). In lines 13-14, we check if the score can beat $\sigma^{min}$ (intra-query pruning). In lines 15-16, $r$ is updated and inserted into $\mathcal{R}^s$ because it is promising. In line 17, we update the anchor match node $m_a$ for the next iteration.

The details of the $mergeResults$ method in line 21 of Algorithm \ref{alg:multipleprocessing} is presented in Algorithm \ref{alg:merging}. In lines 2-5, we compute the potential result root $v^u_{slca}$ which is yet to be confirmed. If the previous result root $v_{slca}$ is not an ancestor of $v^{u}_{slca}$, then $v_{slca}$ is confirmed as the SLCA result root in line 6. Thus, from $\mathcal{R}^s$, we retrieve the precomputed closest match nodes $m^{li}_i, \forall i \in [1,n-1]$ of the shared part which is implemented in function $retrieveNode$ in line 8 and the distance of the shared part which is implemented in the function $retrieveDist$ in line 9. Then, in lines 10-11 we compute the closest match node $m^{ln}_n$ for the unshared part and its distance to $v_{slca}$ and finally, add this distance to the total distance of $r$. In lines 13-15, we compute $r.score$ and if $r.score > \sigma^{min}$, the result is added to $\mathcal{R}^*$. We update $v_{slca}$ with the current result root $v^{u}_{slca}$ if $v^{u}_{slca}\not\prec_a v_{slca}$ in lines 16-17. Then, the anchor match node $m_a$ is updated in line 18. In lines 19-21, if the last result root is SLCA and is not considered, similar steps are taken and the corresponding result is inserted into $\mathcal{R}^*$ if it is promising.

\begin{algorithm}[tb]
\SetKwInOut{Input}{Input}
\SetKwInOut{Output}{Output}
 \Input{$\mathcal{R}^*, \mathcal{R}^s, \sigma^{min}, \alpha, q^{\prime\prime}, S_1$}
 \Output{Top-k Result $\mathcal{R}^*$ }
$r \gets null$;$v_{slca} \gets null$; \par
$m_a \gets getAnchor(\{getNext(S_1),getNext(S_2)\})$;\par
\While{ $m_a \neq null$ }
{	
  $m_i \gets closest(m_a,S_i), \forall i \in [1,2] \text{ and } i \neq a$;\par
	$v^{u}_{slca} \gets lca(m_1,m_2)$;\par	
	\If{$r \neq null$ \textbf{and} $v_{slca} \not\prec_a v^{u}_{slca}$}
	{		
			\If{$v_{slca} \in \mathcal{R}^s$}
			{
			  $\small{m^{li}_i\gets retrieveNode(r,\mathcal{R}^s), \forall i \in [1,n-1]}$;\par
				$d\gets retrieveDist(r,\mathcal{R}^s)$;\par
				$m^{ln}_n\gets getTight(S_1,r.cursor_1)$;\par 
				$d\gets d + getDist(v_{slca},m^{ln}_n)$; \par
				$r \gets (v_{slca},\{m^{li}_i,\forall i \in [1,n]\})$;\par
							
				$r\cdot score\gets q^{\prime\prime}\cdot sim\times\frac{1}{log_\alpha(d+1)+1}$;\par		
				\If{$r\cdot score>\sigma^{min}$}
				{
					update top-k list $\mathcal{R}^*$ with $r$; \par
				}			
			}
	
		\ElseIf{$v^{u}_{slca} \not\prec_a v_{slca}$}
		{
			$v_{slca} \gets v^{u}_{slca}$; $r.add(S_1.cursor)$;\par 
		}	
	}
	
	$m_a \gets getAnchor(\{getNext(S_i), \forall i \in [1,2]\})$;\par
}

\If{$v^{u}_{slca} \not\prec_a v_{slca}$}
{
	make $r$ by repeating lines 8-12;\par
	score $r$ and insert to $\mathcal{R}^*$ if promising; 
}	  	  

\Return{$\mathcal{R}^s$}
 \caption{BA-QP (mergeResults)}
\label{alg:merging}
\end{algorithm}

\begin{table*}
\caption{A Sample of test queries for IMDB and DBLP datasets}
\vspace{-1ex}
\begin{tabular}{ |l|l|l|l|l|l| }
\hline
\multicolumn{6}{ |c| }{Test Query Set} \\
\hline
\# & IMDB & $|\mathcal{Q}|$ & \# & DBLP & $|\mathcal{Q}|$ \\ \hline
\hline
$q_0.1$ & ghost, badgering, movie & 12 & $q_0.1$ & exigency, analysis, system & 252  \\ \hline
$q_0.2$ & battler, spanish, drama & 272 & $q_0.2$ & academic, fraudulence, threat & 4500 \\ \hline
$q_0.3$ & slump, federal, reserve, harshness, documentary & 285 & $q_0.3$ & information, ordination, track & 360 \\ \hline
$q_0.4$ & reproach, trespasser, fight, drama & 357 & $q_0.4$ & involvement, neuroscience, indicant, information & 350\\ \hline
$q_0.5$ & treasonist, zombie, shiver & 1295 & $q_0.5$ & mutter, alarm, analysis & 2079 \\ \hline
$q_0.6$ & research, outlander, universe & 3528 & $q_0.6$ & deceit, type, analysis & 8190 \\ \hline
$q_0.7$ & mass murder, horror, perfidy & 11900 & $q_0.7$ & online, trust, selling & 276  \\ \hline
$q_0.8$ & victory, exaltation, drama & 851 & $q_0.8$ & aftermath, type, analysis, system & 225 \\ \hline
$q_0.9$ & partiality, perfidy, fear, english & 442 & $q_0.9$ & psychopathy, symptom, visualization, science & 3276 \\ \hline
$q_0.10$ & criminal, overcharge, loneliness & 12240 & $q_0.10$ & interloper, search , analysis, system & 352  \\ \hline
$q_0.11$ & criminal, fear, spanish & 1152 & $q_0.11$ & trace, audit , system & 238 \\ \hline
$q_0.12$ & victory, exuberance, drama, fight & 437 & $q_0.12$ & trace, type, analysis & 255 \\ \hline
\end{tabular}
\vspace{-3ex} 
\label{tbl:sample}
\end{table*}

\section{Experiments}
This section evaluates the effectiveness and the efficiency of our approach for solving no-but-semantic-match problem in XML keyword search. To compare our results, we adapt and implement the closely related existing XOntoRank method \cite{FarfanHRW09} which uses ontology to enhance the XML keyword search on medical datasets. The adaptation is achieved as follows. We create a hash map for the ontologically relevant keywords to the original keywords by using ontological knowledge base. Then for each keyword $k_i \in q_0$ we look up the hash map and find the candidates and put them in the set of candidate keywords $\mathcal{K}$. Then we create Onto-DIL for each keyword $k_i \in q_0$ and all of its associated candidates in $\mathcal{K}$. Afterward, we compute the node score for each entry based on the relevance degree of the original keyword $k_i$ to the candidate keywords in $\mathcal{K}$. Finally, the inverted list is added to Onto-DIL $\mathcal{S}_{onto}$. After creating $\mathcal{S}_{onto}$, we compute LCAs by using Onto-DIL index and for each result $r$, the score is computed using the node score of its matched nodes. If the result score $r.score$ is better than the threshold $\sigma^{min}$, the result is added to top-k list $\mathcal{R}^*$. The above is pseudocoded in Algorithm \ref{alg:XOnto-QP}. We term this method as XO-QP in this paper.

\begin{algorithm}[tb]
\SetKwInOut{Input}{Input}
\SetKwInOut{Output}{Output}
 
\Input{User Query $q_0$, Data $T$, Ontology $O$}
\Output{Top-k Semantically Related Results $\mathcal{R}^*$}
$H \gets createHashMap(q_0,O)$;\par
$\sigma^{min} \gets 0$;\par
\ForEach{$k_i\in q_0$}
{
	 $\mathcal{K} \gets findRelevant(k_i,H)$;\par	 
	 $S^\prime \gets createDIL(k_i \cup \mathcal{K}, T)$;\par
	 $S^\prime.add(NS(k_i,k)), \forall k \in \mathcal{K}$;\par	
	 $\mathcal{S}_{onto} \gets \mathcal{S}_{onto} \cup S^\prime$;\par 
}

\While{inverted lists in $\mathcal{S}_{onto}$ $\neq$ null}
{ 
	$r \gets computeLCA(\mathcal{S}_{onto})$;\par
	$r.score \gets \sum{m.NS}, \forall m \in r$;\par
	\If{$r.score > \sigma^{min}$}
	{
		 update top-k list $\mathcal{R}^*$ with $r$;\par
		 $update(\sigma^{min})$;\par	
	}	
}

\Return {$\mathcal{R}^*$}

\caption{XO-QP}
\label{alg:XOnto-QP}
\end{algorithm}     

\vspace{-1ex}
\subsection{Settings}
\textbf{Datasets and Queries:} We evaluate our algorithms on two real datasets: (a) IMDB 170MB, that includes around 150,000 recent movies and TV series. (b) DBLP 650MB, which contains publications in major journals and proceedings. We use a wide range of queries to test the efficiency of our proposed methods for each dataset. For each test query, we choose keywords which satisfy the followings: (a) a keyword should be used often by the users and (b) a non-exiting keyword should have some semantic counterparts, which have direct mapping in the data source. The test queries have no-but-semantic-match problem. Table \ref{tbl:sample} presents a part of the test queries called sample queries for detailed analysis of efficiency and effectiveness of our methods. 
\noindent\\
\textbf{Environment:} All algorithms are implemented in $C\#$ and the experiments are conducted on a PC with 3.2 GHz CPU, 8 GB memory running 64-bit windows 7.   

\begin{figure*}

\centering
        \subfloat[]{\label{LabelA}
  \includegraphics[width=2.1in,height=1.3in]{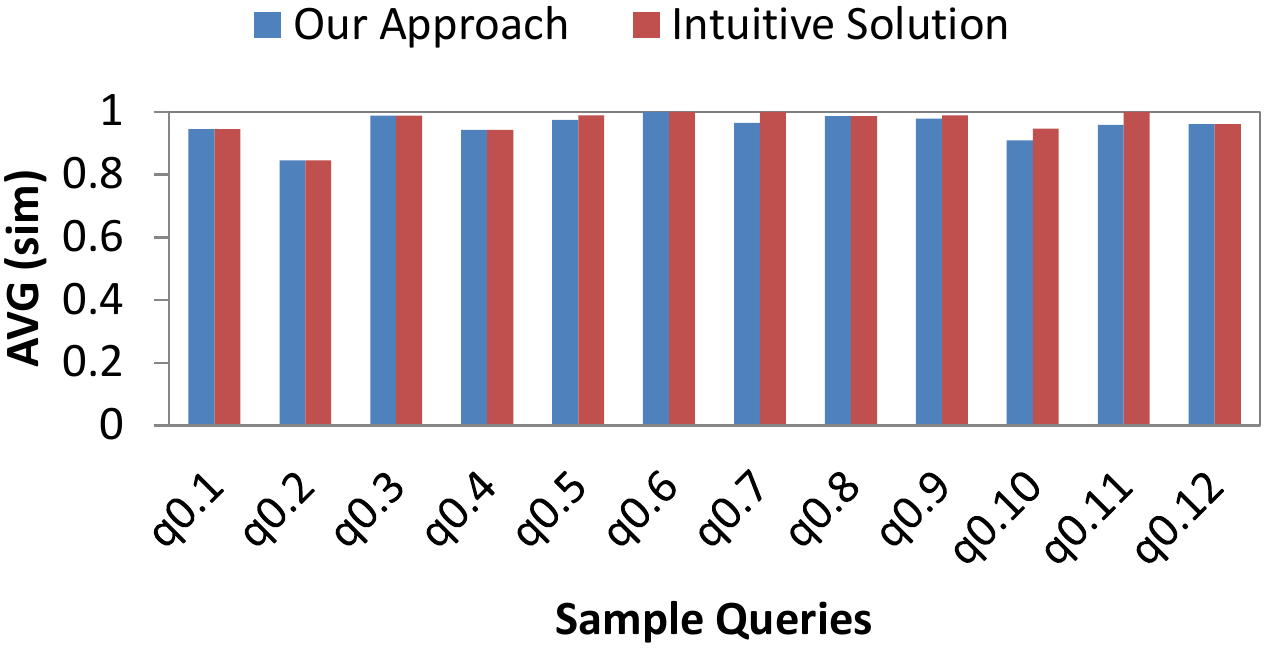}        
}
\hfill
\subfloat[]{\label{labelB}
  \includegraphics[width=2.1in,height=1.3in]{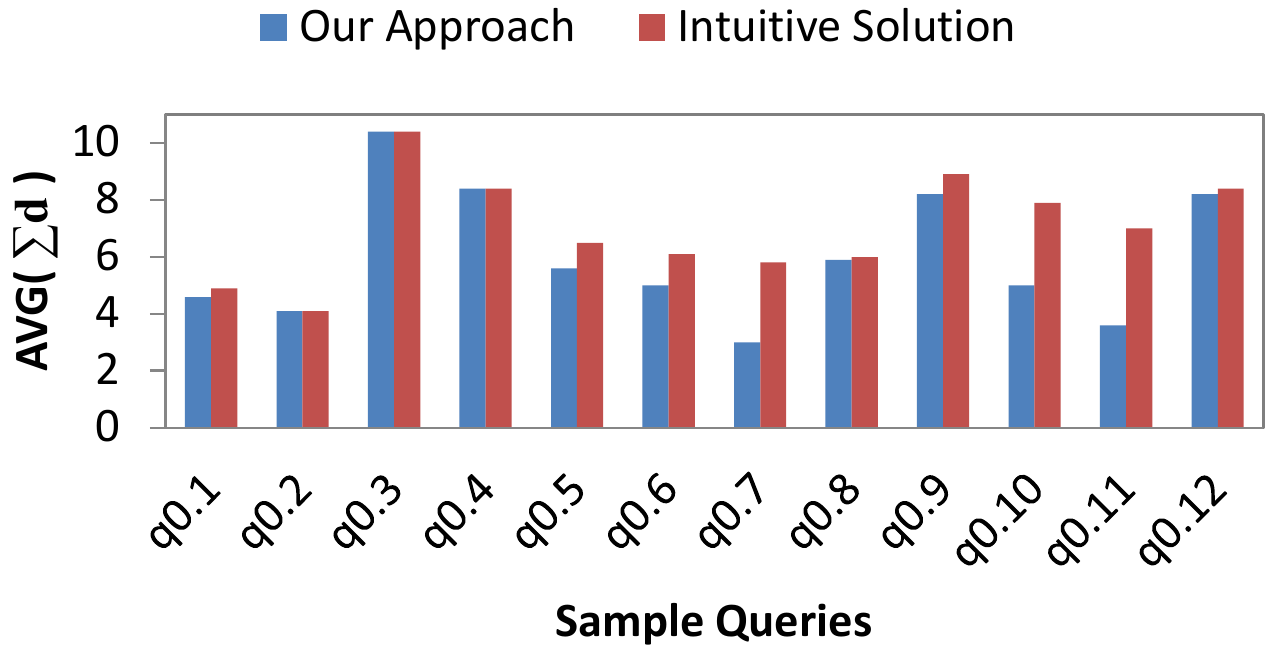}
}
\hfill
\subfloat[]{\label{labelC}%
  \includegraphics[width=2.1in,height=1.3in]{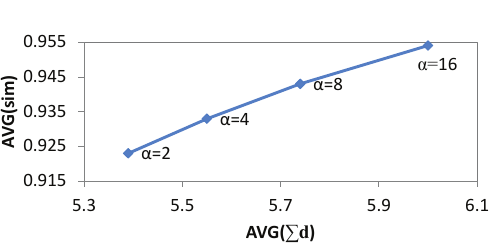}
}		
        
				\caption{(a) Average similarity of the candidate queries and (b) average distance of top-10 results for our approach and intuitive solution on IMDB; (c) Average query similarity versus average result distance with varying $\alpha$.}
				\label{fig:IMDBEffect}
\vspace{-1ex} 
\end{figure*}

\begin{figure*}

\centering
        \subfloat[]{\label{LabelA}
  \includegraphics[width=2.1in,height=1.3in]{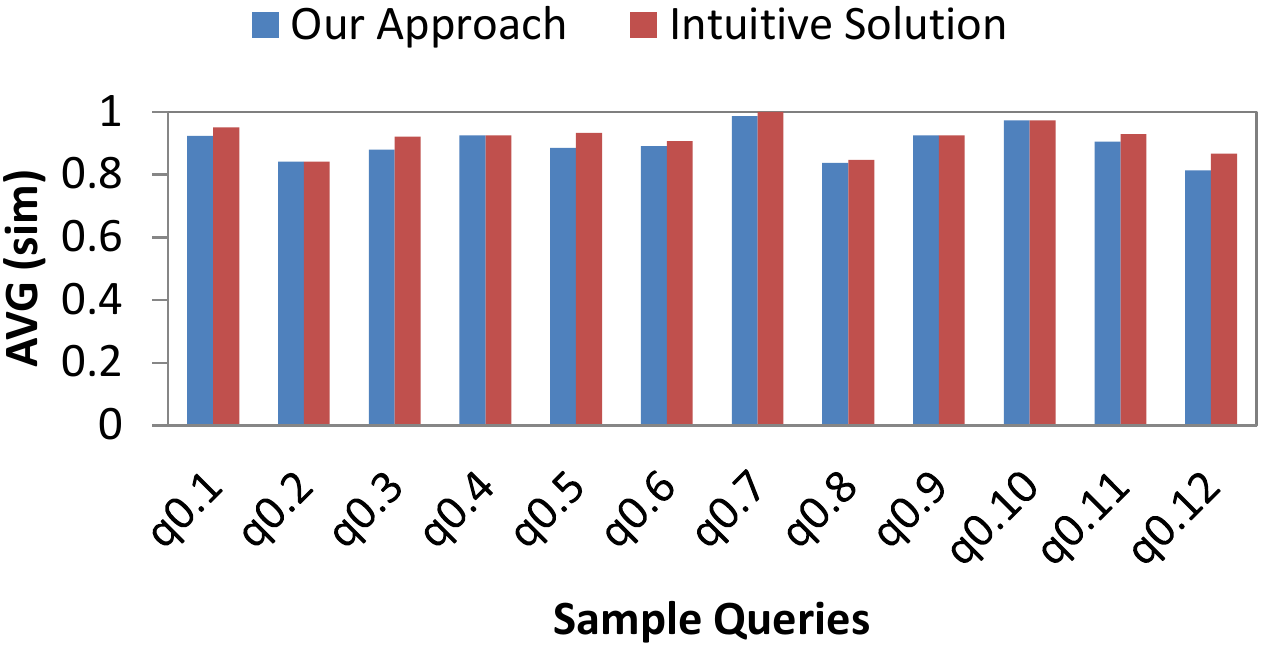}   
}
\hfill
\subfloat[]{\label{labelB}
   \includegraphics[width=2.1in,height=1.3in]{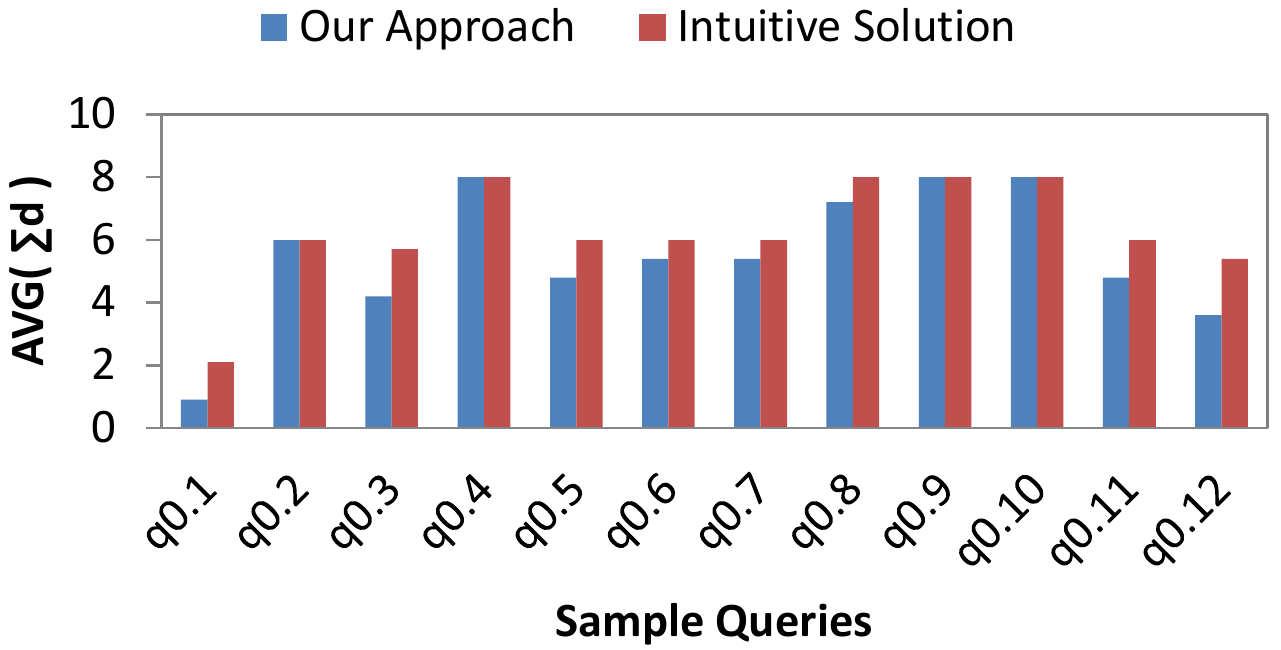}
}
\hfill
\subfloat[]{\label{labelC}%
   \includegraphics[width=2.1in,height=1.3in]{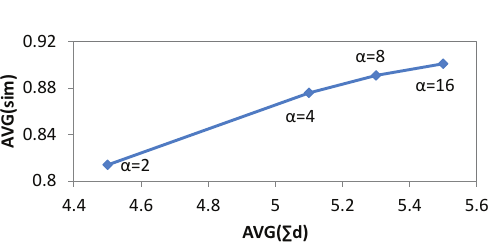}
}		
        
				\caption{(a) Average similarity of the candidate queries and (b) average distance of top-10 results for our approach and intuitive solution on DBLP; (c) Average query similarity versus average result distance with varying $\alpha$.}
				\label{fig:DBLPEffect}
\vspace{-1ex} 
\end{figure*}

\subsection{Effectiveness}
This section evaluates the effectiveness of our approach from different perspectives.
\subsubsection{Our Approach versus Intuitive Solution}
When a user issues a query and faces an empty result, she might try to change the initial query to obtain some results. However, being unfamiliar with the data source, the user is likely to think of a few synonyms of the keywords of the initial query. Assume the user is a kind of expert and succeeds in constructing the 10 topmost similar queries when changing the initial query. These queries may not produce good results but we consider this \emph{intuitive solution} as a benchmark to gauge the effectiveness of the technique we propose. Here, we compare the top-10 results retrieved by our approach which processes all possible candidate queries with the \emph{intuitive solution} which processes only the 10 topmost similar queries for all test queries given in Table 1. In Fig. $\ref{fig:IMDBEffect}$ (a) and Fig. $\ref{fig:DBLPEffect}$ (a), the average candidate query similarity of the top-10 results for the two approaches are presented for each dataset. Clearly for all the test queries, the average similarity of the results retrieved by our approach is very close to the average candidate query similarity of the results of the \emph{intuitive solution}. Fig. $\ref{fig:IMDBEffect}$ (b) and Fig. $\ref{fig:DBLPEffect}$ (b) show, on average, smaller number of edges (indicate better cohesiveness) for the top-10 results retrieved from our approach compared to the \emph{intuitive solution}. This shows that an ad-hoc approach, even if it is suggested by an expert, is unlikely to find results that are superior to those provided by the systematic method suggested in this study. In comparison with the \emph{intuitive solution}, our approach also does not retrieve results from the candidate queries that are far from the user given initial query in terms of candidate query similarity. In addition of it, our approach provides the flexibility of trading of the above two aspects, e.g., sometimes users may prioritize cohesive results over similarity to the original query. The tuning parameter $\alpha$ in our approach can be used to balance the weight of these two aspects. Fig. $\ref{fig:IMDBEffect}$(c) and Fig. $\ref{fig:DBLPEffect}$(c), show that as $\alpha$ becomes smaller, the similarity of the results decreases because more priority is given to result cohesiveness and therefore, the cohesiveness score of the results improves. That is, by setting $\alpha$ to smaller value, our approach effectively retrieves more cohesive results while the similarities of their contributing candidate queries are not very far from the initial query.                              

\begin{table*}
\centering
\caption{Ranking of $r_1$ is tolerant to $\alpha$ as it is better than $r_4$ in terms of both query similarity and result cohesiveness}
\begin{tabular}{ |l|l|l|l|l|l|l|l| }
\hline
$\alpha$ & $d(r_1,T)$ & $\lambda(q_0,q_1)\times\theta(r_1,T)$ & $\sigma(r_1,q_1,T)$ & $d(r_4,T)$ & $\lambda(q_0,q_4)\times\theta(r_4,T)$  & $\sigma(r_4,q_4,T)$ & $\Delta(r_1,r_4)$ \\ \hline
\hline
$2$ & 7 & $0.9167 \times 0.25$   & 0.2291  & 8 & $0.8462 \times 0.2398$  & 0.2029  & 0.0262\\ \hline
$3$ & 7 & $0.9167 \times 0.3456$ & 0.3168  & 8 & $0.8462 \times 0.3333$  & 0.282   & 0.0348\\ \hline
$4$ & 7 & $0.9167 \times 0.4$    & 0.3666  & 8 & $0.8462 \times 0.3868$  & 0.3273  & 0.0393\\ \hline
$8$ & 7 & $0.9167 \times 0.5$    & 0.4583  & 8 & $0.8462 \times 0.4862$  & 0.4114  & 0.0469\\ \hline
$16$ & 7 & $0.9167 \times 0.5714$& 0.5238  & 8 & $0.8462 \times 0.5578$  & 0.472   & 0.0518\\ \hline
\end{tabular}
\label{tbl:case1}
\end{table*} 

\begin{table*}
\centering
\caption{Trading off candidate query similarity and result cohesiveness in $r_2$ and $r_7$}
\begin{tabular}{ |l|l|l|l|l|l|l|l| }
\hline
\centering
$\alpha$ & $d(r_2,T)$ & $\lambda(q_0,q_2)\times\theta(r_2,T)$ & $\sigma(r_2,q_2,T)$ & $d(r_7,T)$ & $\lambda(q_0,q_7)\times\theta(r_7,T)$  & $\sigma(r_7,q_7,T)$ & $\Delta(r_2,r_7)$ \\ \hline
\hline
$2$ & 11 & $0.8462 \times 0.2181$  & 0.1845 & 7 & $0.7549 \times 0.25$    & 0.1887 &-0.0042\\ \hline
$3$ & 11 & $0.8462 \times 0.3065$  & 0.2593 & 7 & $0.7549 \times 0.3456$  & 0.2608 &-0.0015\\ \hline
$4$ & 11 & $0.8462 \times 0.3581$  & 0.303  & 7 & $0.7549 \times 0.4$     & 0.3019 & 0.0011\\ \hline
$8$ & 11 & $0.8462 \times 0.4555$  & 0.3854 & 7 & $0.7549 \times 0.5$     & 0.3774 & 0.008 \\ \hline
$16$ & 11 & $0.8462 \times 0.5273$ & 0.4462 & 7 & $0.7549 \times 0.5714$  & 0.4313 & 0.0149\\ \hline
\end{tabular}
\label{tbl:case2}
\end{table*} 

\subsubsection{Effect of Tuning Parameter}
Now, we provide two fine-grained case studies as follows: (1) the influence of the tuning parameter $\alpha$ on the ranking of the retrieved results and (2) trading off between query similarity and result cohesiveness in the final top-$k$ results based on $\alpha$.

\textbf{Case Study-1:} This case study demonstrates that our approach is tolerant to the settings of $\alpha$ if one result beats another one in terms of both query similarity and result cohesiveness. This is also expected as the user might explore the top cohesive results with better similarity first. Consider the queries given in Example 7. Here, we extract $r_1$ from $q_1$ and $r_4$ from $q_4$ as shown in Fig. \ref{fig:Coherency}. From Table \ref{tbl:case1}, we see that $r_1$ will always be ranked better than $r_4$ as $r_1$ has the higher overall score $\sigma$ than $r_4$ for all settings of $\alpha$. 

\textbf{Case Study-2:} This case study demonstrates how the user can trade off between query similarity and the result cohesiveness based on $\alpha$. Assume a user would like to explore the results with better cohesiveness first than those with higher similarity. Consider a candidate query $q_7=\{Jack,academic,position\}$ with $\lambda(q_0,q_7)=0.7549$ and a result $r_7$ from $q_7$ with $d(r_7,T)=7$. Now, we compare it with $r_2$ with $d(r_2,T)=11$ from $q_2$ with $\lambda(q_0, q_2)=0.8462$ as given in Example 7. From Table \ref{tbl:case2}, we observe that $r_7$ outranks $r_2$ for $\alpha=[2,3]$. Now, a user needs to set $\alpha>3$ to explore $r_2$ before $r_7$ in the result list by putting more emphasis on query similarity than result cohesiveness.

\subsubsection{Our Approach versus XO-QP}
In this section, we provide a detailed study of the top candidate query obtained from the top-10 results for different methods. The methods include our proposed method using the tuning parameter $\alpha = 2$ and $\alpha = 16$, intuitive solution (which only executes the top similar queries), XO-QP (which takes the result distance fixed), and XO-QP (which apply the result distance into ranking). In fact, the top-1 query, is the candidate query which has the biggest number of results among top-10 results. The number of results from a candidate query to the total number of top-k results is defined as $\eta(\mathcal{R}^*,q^\prime)$. Therefore, we provide in Table \ref{table:TopQueryIMDB} and Table \ref{table:TopQueryDBLP} the top-1 candidate queries on IMDB and DBLP datasets respectively. These queries have the maximum $\eta(\mathcal{R}^*,q^\prime)$ among the candidate queries which contribute to top-k results. Clearly, our method that set $\alpha$ to $2$ has the minimum $d(r,D)$ among other methods in both datasets. That's because if we set $\alpha$ to a lower number, we put more emphasis on result cohesiveness and therefore, the results with better cohesiveness are ranked higher. However, when we look at our method using $\alpha = 16$, the similarity of the top query in most cases is better comparing to our method using $\alpha = 2$. In the intuitive solution, the top candidate query has the maximum similarity in most cases. Since the intuitive solution only generates top similar queries to retrieve the semantically related results, the similarity of the top candidate query is the maximum in most cases. The XO-QP using fixed distance for ranking the results have the worst cohesiveness score in most cases. That's because XO-QP uses XRank\cite{GuoSBS03} for retrieving results and therefore, does not limit the results to SLCAs. Moreover, when we do not reflect the distance of the results into the ranking, more sparse results may be ranked higher and inserted to Top-$k$ list. However, when we reflect the distance into ranking of XRank, then the cohesiveness of the results improves.    
                    
\begin{table*}
\caption{Top candidate queries for different methods in IMDB }
\begin{tabular}{ |l|l|l|l|l|l| }
\hline
Original Query & Method & Top-1 Candidate Query & $\eta(\mathcal{R}^*,q^\prime)$ & $\lambda(q_0,q^\prime)$ & $d(r\in q^\prime,T)$ \\ \hline
\hline
$q_0.4$ & Our Approach($\alpha=2$)  & blame, fight, drama, intruder & 0.1 & $0.94$ & 8  \\ \hline
$q_0.4$ & Our Approach($\alpha=16$) & blame, fight, drama, intruder & 0.1 & $0.94$ & 8 \\ \hline
$q_0.4$ & Intuitive Solution  & blame, fight, drama, intruder & 0.1 & $0.94$ & 8  \\ \hline
$q_0.4$ & XO-QP (No Dist) & blame, fight, drama, squatter & 0.5 & 0.92 &  12 \\ \hline
$q_0.4$ & XO-QP (Dist) & blame, fight, drama, intruder & 0.3 & 0.94 &  11 \\ \hline
$q_0.7$ & Our Approach($\alpha=2$)  & massacre, panic, betrayal & 0.2 & $0.87$ & 3  \\ \hline
$q_0.7$ & Our Approach($\alpha=16$) & slaughter, treachery, revulsion & 0.1 & 1 & 6 \\ \hline
$q_0.7$ & Intuitive Solution & slaughter, treachery, revulsion  & 0.1 & 1 & 6  \\ \hline
$q_0.7$ & XO-QP (No Dist) & uxoricide, apprehension, betrayal & 0.1 & 0.8 & 8 \\ \hline
$q_0.7$ & XO-QP (Dist) & hit, panic, betrayal & 0.3 & 0.8 & 3 \\ \hline
$q_0.9$ & Our Approach($\alpha=2$)  & call, betrayal, fear, english & 0.1 & $0.85$ & 4  \\ \hline
$q_0.9$ & Our Approach($\alpha=16$) & fancy, treachery, fear, english & 0.1 & 1 & 8 \\ \hline
$q_0.9$ & Intuitive Solution  & fancy, treachery, fear, english & 0.1 & 1 & 8  \\ \hline
$q_0.9$ & XO-QP (No Dist) & bias, treachery, fear, english & 0.5 & 0.93 & 12 \\ \hline
$q_0.9$ & XO-QP (Dist) & tilt, betrayal, fear, english & 0.4 & 0.93 & 11 \\ \hline
$q_0.10$ & Our Approach($\alpha=2$)  & murderer, extortion, blood & 0.3 & $0.76$ & 3  \\ \hline
$q_0.10$ & Our Approach($\alpha=16$) & crook, extortion, desolation & 0.1 & 0.94 & 7 \\ \hline
$q_0.10$ & Intuitive Solution & crook, extortion, desolation & 0.1 & 0.94 & 7  \\ \hline
$q_0.10$ & XO-QP (No Dist) & outlaw, extortion, desolation & 0.3 & 0.94 & 9 \\ \hline
$q_0.10$ & XO-QP (Dist) & outlaw, extortion, desolation & 0.3 & 0.94 & 9 \\ \hline
$q_0.12$ & Our Approach($\alpha=2$)  & enthusiasm, drama, fight, triumph & 0.1 & 1 & 8  \\ \hline
$q_0.12$ & Our Approach($\alpha=16$) & enthusiasm, drama, fight, triumph & 0.1 & 1 & 8 \\ \hline
$q_0.12$ & Intuitive Solution & enthusiasm, drama, fight, triumph & 0.1 & 1 & 8  \\ \hline
$q_0.12$ & XO-QP (No Dist) & madness, drama, fight, triumph & 0.6 & 0.94 & 12 \\ \hline
$q_0.12$ & XO-QP (Dist) & madness, drama, fight, win & 0.4 & 0.87 & 10 \\ \hline
\end{tabular}
\label{table:TopQueryIMDB}
\end{table*}             

\begin{table*}
\caption{Top candidate queries for different methods in DBLP }
\begin{tabular}{ |l|l|l|l|l|l| }
\hline
Original Query & Method & Top-1 Candidate Query & $\eta(\mathcal{R}^*,q^\prime)$ & $\lambda(q_0,q^\prime)$ & $d(r\in q^\prime,T)$ \\ \hline
\hline
$q_0.1$ & Our Approach($\alpha=2$)  & emergency, research, system  & 0.4 & $0.88$ & 3  \\ \hline
$q_0.1$ & Our Approach($\alpha=16$) & emergency, research, system  & 0.4 & $0.88$ & 3 \\ \hline
$q_0.1$ & Intuitive Solution & emergency, research, system  & 0.4 & $0.88$ & 3  \\ \hline
$q_0.1$ & XO-QP (No Dist) & pinch, breakdown, system & 0.6 & 0.94 &  6 \\ \hline
$q_0.1$ & XO-QP (Dist) & pinch, breakdown, system & 0.6 & 0.94 &  6 \\ \hline
$q_0.7$ & Our Approach($\alpha=2$)  & online, content, exchange & 0.1 & $0.52$ & 0  \\ \hline
$q_0.7$ & Our Approach($\alpha=16$) & online, faith, merchandising & 0.1 & 1 & 6 \\ \hline
$q_0.7$ & Intuitive Solution & online, faith, merchandising & 0.1 & 1 & 6  \\ \hline
$q_0.7$ & XO-QP (No Dist) & online, faith, trading & 1 & 1 & 6 \\ \hline
$q_0.7$ & XO-QP (Dist) & online, faith, trading & 1 & 1 & 6 \\ \hline
$q_0.8$ & Our Approach($\alpha=2$)  & impact, order, analysis, system & 0.1 & 0.69 & 0  \\ \hline
$q_0.8$ & Our Approach($\alpha=16$) & impact, order, analysis, system & 0.1 & 0.69 & 0 \\ \hline
$q_0.8$ & Intuitive Solution & wake, variety, analysis, system & 0.1 & 0.94 & 8  \\ \hline
$q_0.8$ & XO-QP (No Dist) & wake, variety, analysis, system & 1 & 0.94 & 8 \\ \hline
$q_0.8$ & XO-QP (Dist) & wake, variety, analysis, system & 1 & 0.94 & 8 \\ \hline
$q_0.10$ & Our Approach($\alpha=2$)  & whole, study, analysis, system & 0.1 & 0.24 & 0  \\ \hline
$q_0.10$ & Our Approach($\alpha=16$) & intruder, hunt, analysis, system & 0.1 & 1 & 8 \\ \hline
$q_0.10$ & Intuitive Solution & intruder, hunt, analysis, system & 0.1 & 1 & 8  \\ \hline
$q_0.10$ & XO-QP (No Dist) & intruder, lookup, analysis, system & 0.3 & 1 & 8 \\ \hline
$q_0.10$ & XO-QP (Dist) & intruder, lookup, analysis, system & 0.3 & 1 & 8 \\ \hline
$q_0.11$ & Our Approach($\alpha=2$)  & gesture, testing, system & 0.1 & 0.72 & 0  \\ \hline
$q_0.11$ & Our Approach($\alpha=16$) & gesture, testing, system & 0.1 & 0.72 & 0 \\ \hline
$q_0.11$ & Intuitive Solution & hint, testing, system & 0.1 & 0.9 & 6  \\ \hline
$q_0.11$ & XO-QP (No Dist) & hint, search, system & 0.7 & 0.9 & 6 \\ \hline
$q_0.11$ & XO-QP (Dist) & hint, search, system & 0.7 & 0.9 & 6 \\ \hline
\end{tabular}
\label{table:TopQueryDBLP}
\end{table*}             
\subsection{Evaluation of Quality}
In this section, we evaluate the quality of our approach and compare it with XO-QP. Here, we select some sample queries with no-but-semantic-match problem for both datasets to conduct a comprehensive user study and thereafter, evaluate the overall quality of our approach. In order to carry out a fair user study, we select the users among both experts who have worked in XML keyword search areas and naive users who are graduate computer science students. To do the study, we present to the users with the original queries and their top-10 candidate queries/results retrieved from the top result list $\mathcal{R}^*$. After that, we ask the users to assess the quality of each candidate query with regards to their semantic similarity to the original query by scoring the candidate queries/results using Cumulated Gain metric \cite{JarvelinK02}. They score each candidate query/result from 0 to 5 points (5 means the best and 0 means the worst).

\begin{figure*}

\centering
        \subfloat[IMDB]{\label{IMDB}
  \includegraphics[width=3.2in,height=1.8in]{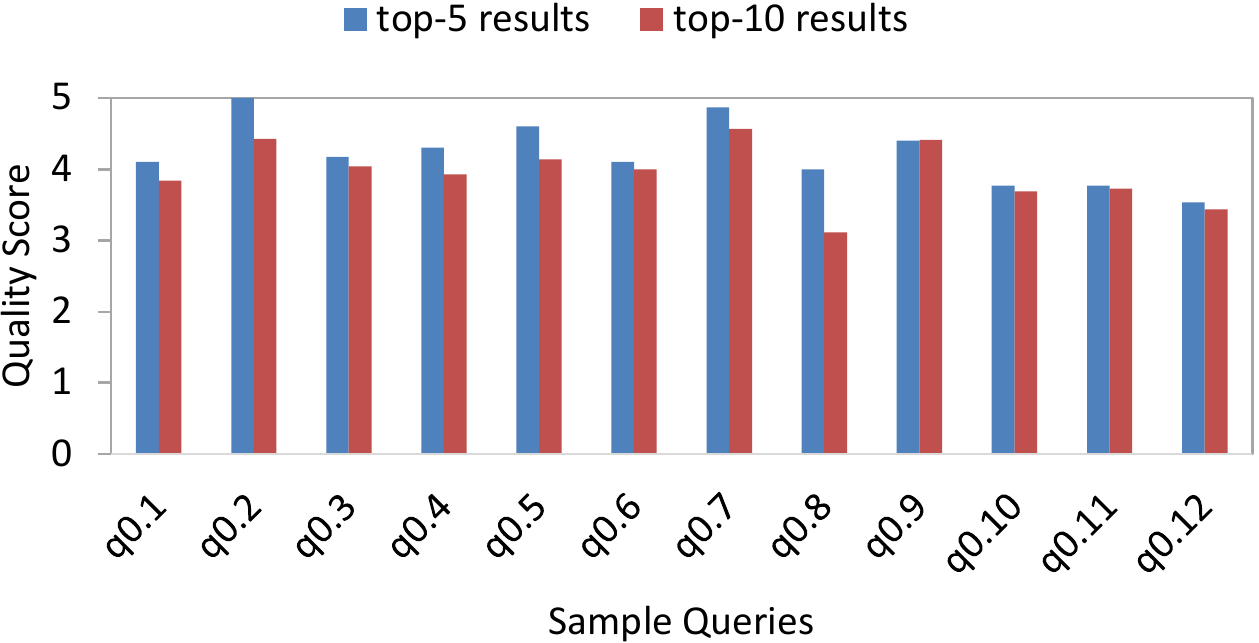}  
}
\hfill
\subfloat[DBLP]{\label{DBLP}
    \includegraphics[width=3.2in,height=1.8in]{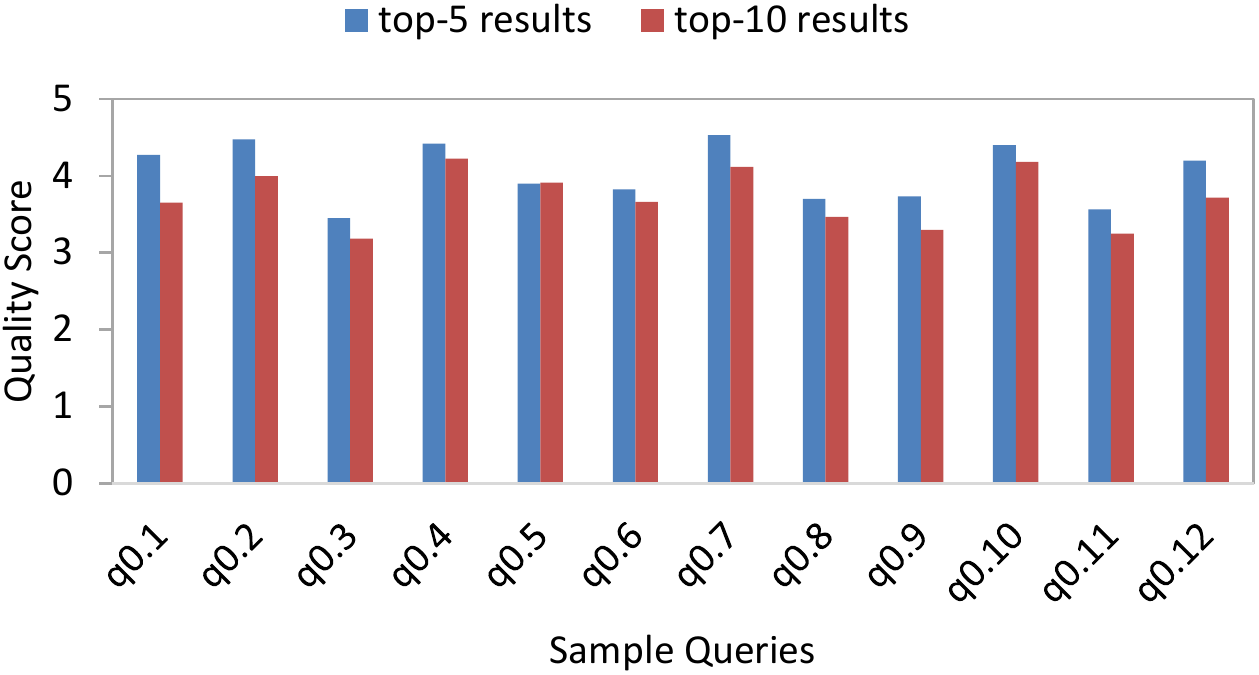}
}        
				\vspace{-1ex}
				\caption{Average quality of results in our approach: top-5 vs. top-10.}
				\label{fig:userstudy}
\end{figure*}

\begin{figure*}

\centering
        \subfloat[IMDB]{\label{IMDB}
  \includegraphics[width=3.2in,height=1.8in]{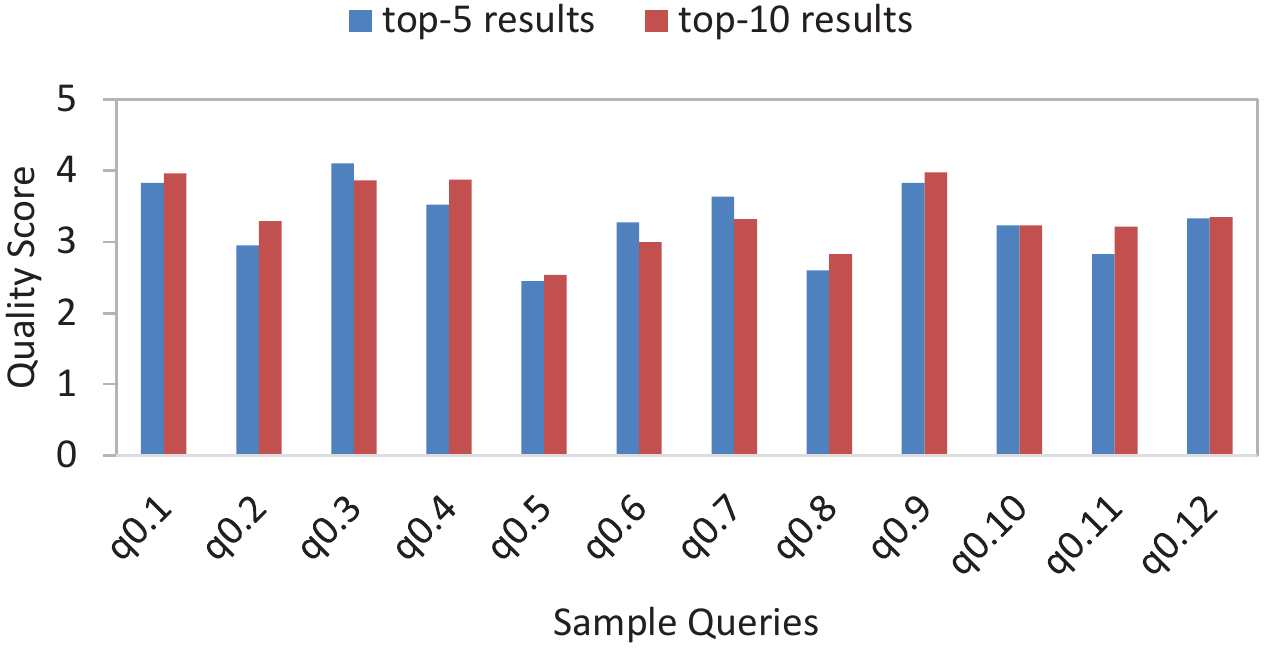}  
}
\hfill
\subfloat[DBLP]{\label{DBLP}
    \includegraphics[width=3.2in,height=1.8in]{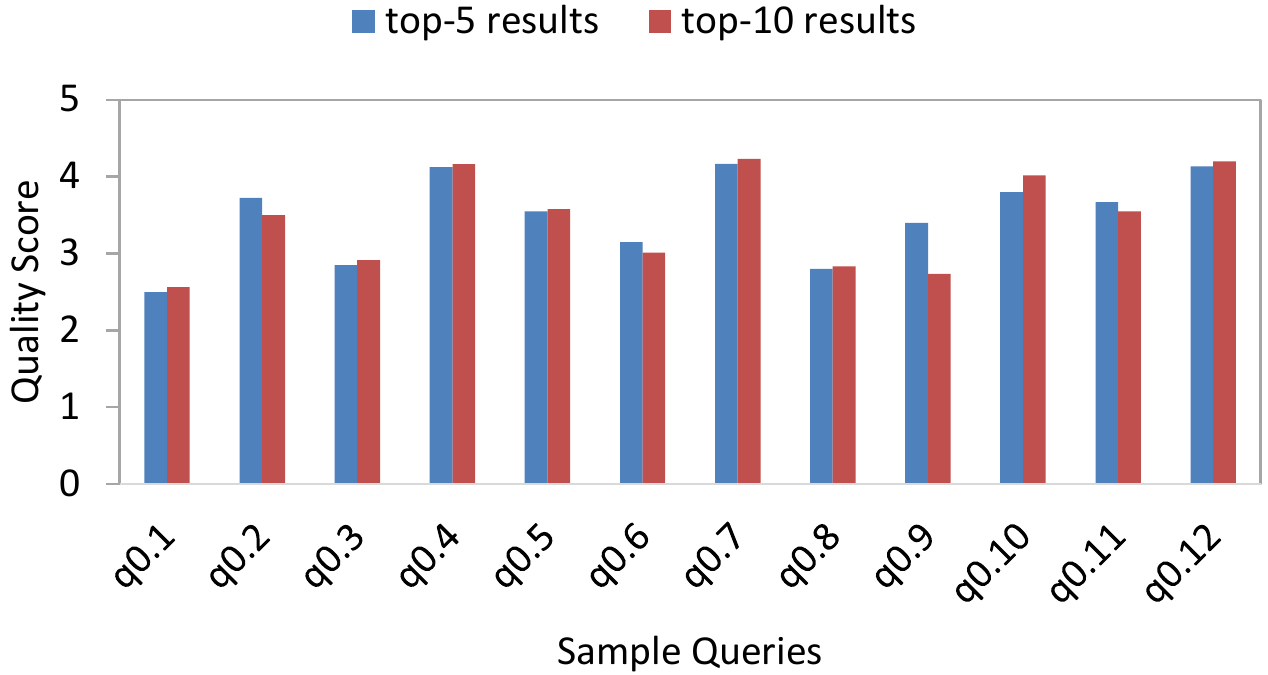}
}        
				\vspace{-1ex}
				\caption{Average quality of results in XO-QP: top-5 vs. top-10.}
				\label{fig:userstudyxonto}
\end{figure*}

\subsubsection{Ranking Scheme Comparison}
The average quality scores of the top-5 and top-10 queries /results for our approach and the existing counterpart XO-QP are presented in Fig. \ref{fig:userstudy} and Fig.\ref{fig:userstudyxonto} respectively. From Fig. \ref{fig:userstudy}, we observe that our proposed method suggests reasonable results for the no-match query for both top-5 and top-10 results. Also, we see that the average quality of top-5 results are always better than the average quality of top-10 results which indicates that our ranking function successfully ranks more similar and meaningful results higher than the rest of the results. However, we observe that in many cases for the existing method XO-QP, the average quality of results for top-10 is higher than top-5 results (for instance for $q_{0.2},q_{0.4},q_{0.11}$ in IMDB and for $q_{0.3},q_{0.7},q_{0.10}$ in DBLP) as shown in Fig. \ref{fig:userstudyxonto}. This indicates that XO-QP sometimes ranks some less similar and meaningful results higher. This is probably because of adapting XRank scheme into XO-QP. We conclude that our approach addresses the no-but-semantic-match problem better than the existing XO-QP method.                  

\begin{figure*}

\centering
        \subfloat[IMDB]{\label{IMDB}
  \includegraphics[width=3.2in,height=1.8in]{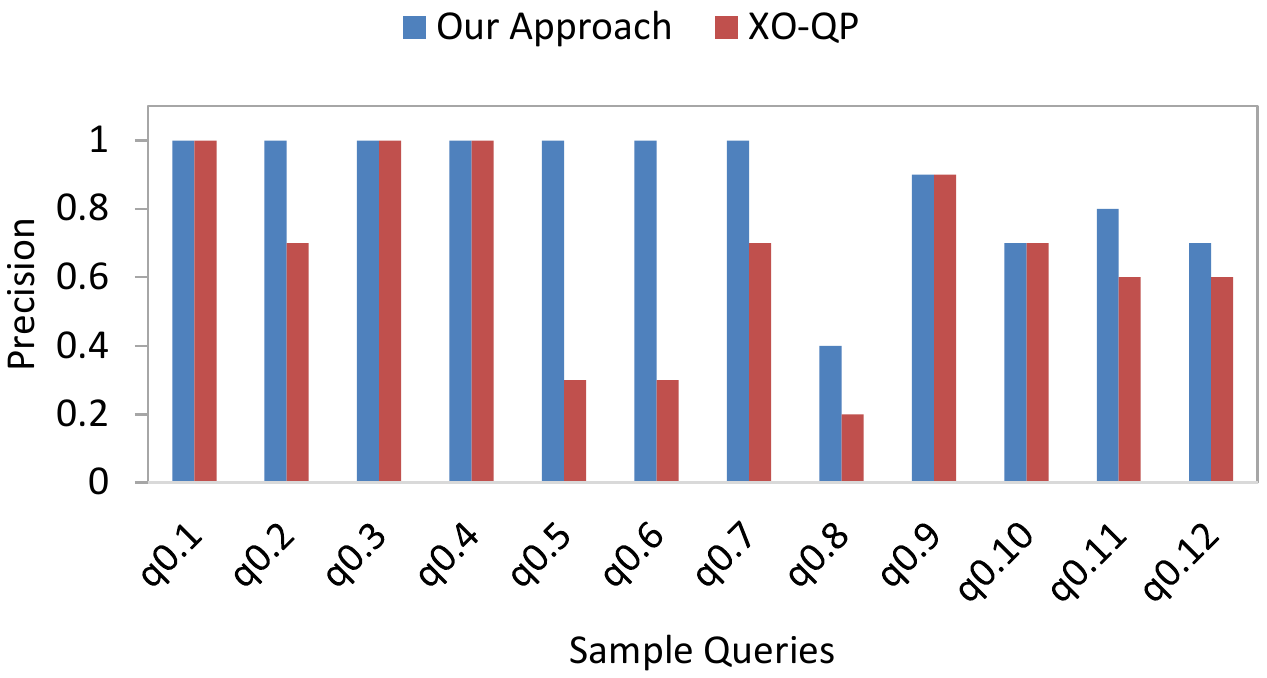}  
}
\hfill
\subfloat[DBLP]{\label{DBLP}
    \includegraphics[width=3.2in,height=1.8in]{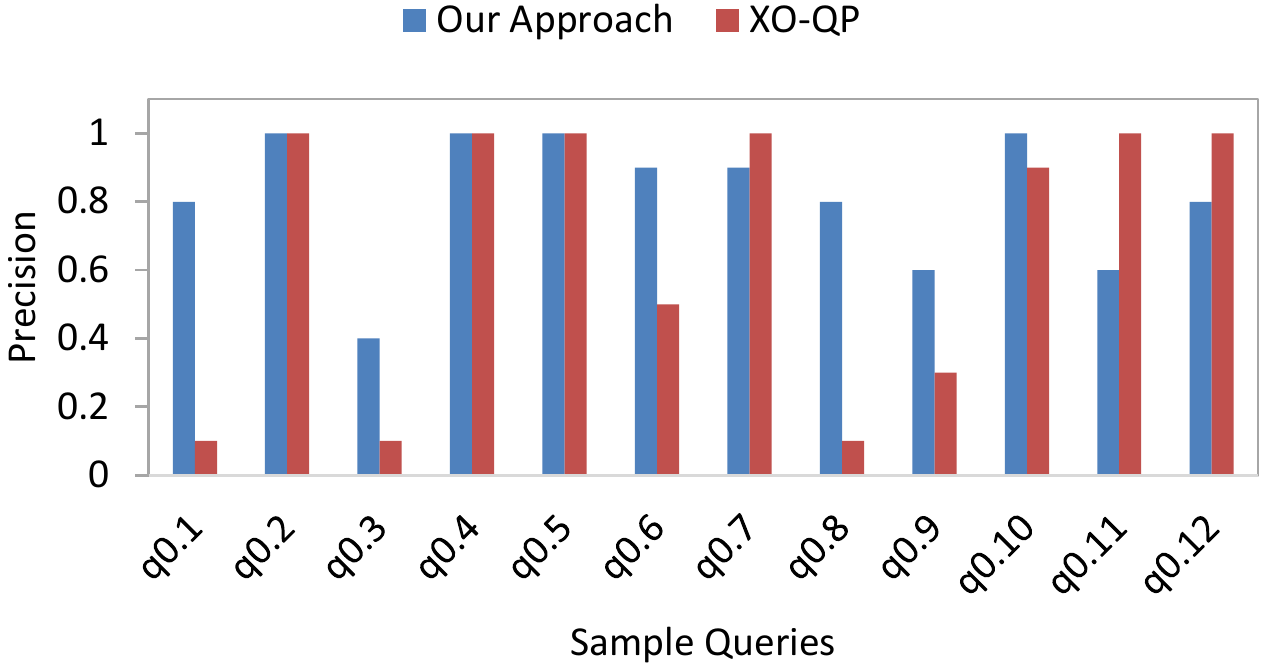}
}        
				\vspace{-1ex}
				\caption{Precision comparison for computing top-10 results.}
				\label{fig:precisioncomparison}
\end{figure*}

\subsubsection{Precision}
In this section, we compare the precision of our proposed method with the XO-QP method. In order to make a comparison, we count the number of meaningful results in top-10 results. A result is regarded as meaningful if the average quality score of the assessors is not less than 3. We see that the precision of our approach and XO-QP is presented on both IMDB and DBLP as shown in Fig. \ref{fig:precisioncomparison}. In IMDB, we see that the precision of our proposed method is better than XO-QP. This is because we use both semantic similarity and cohesiveness scores to rank the results and retrieve the tightest SLCA results with the maximum similarity to the original query keywords. Also in DBLP, we observe that the precision of our method is better than XO-QP in many cases specially in $q_{0.1}, q_{0.3}$ and $q_{0.8}$. That's because our method can effectively retrieve the most similar results that are cohesive and make a meaningful combination in terms of data cohesiveness. However, in some cases like $q_{0.11}$ and $q_{0.12}$, we see that the precision of our method is smaller than XO-QP because it uses less similar keywords in the candidate query to retrieve more cohesive results. But the precision is not largely deteriorated as we see in XO-QP in many cases.                  
\begin{figure*}

\centering
        \subfloat[IMDB]{\label{IMDB}
  \includegraphics[width=3.3in,height=2in]{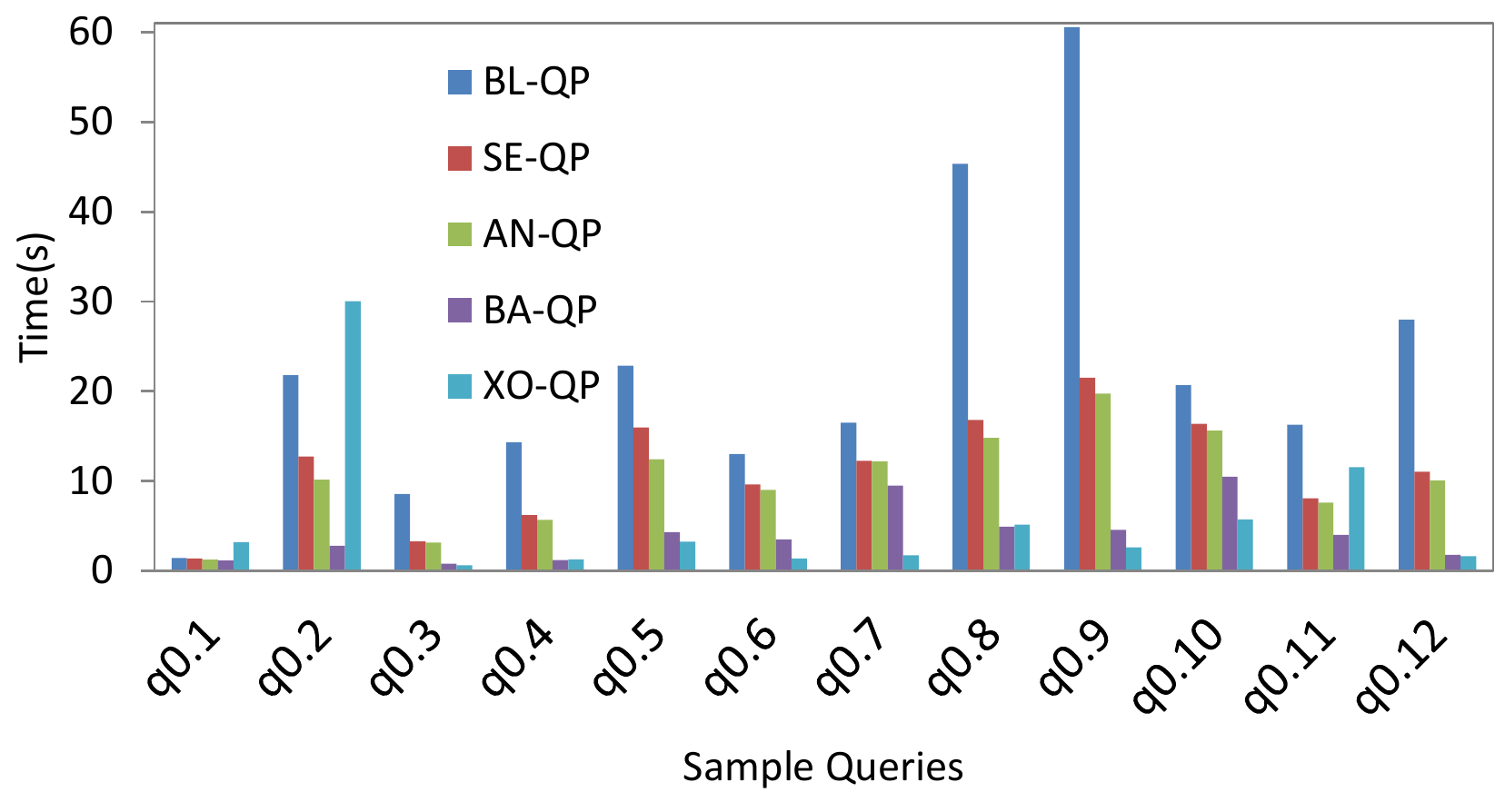}  
}
\hfill
\subfloat[DBLP]{\label{DBLP}
    \includegraphics[width=3.3in,height=2in]{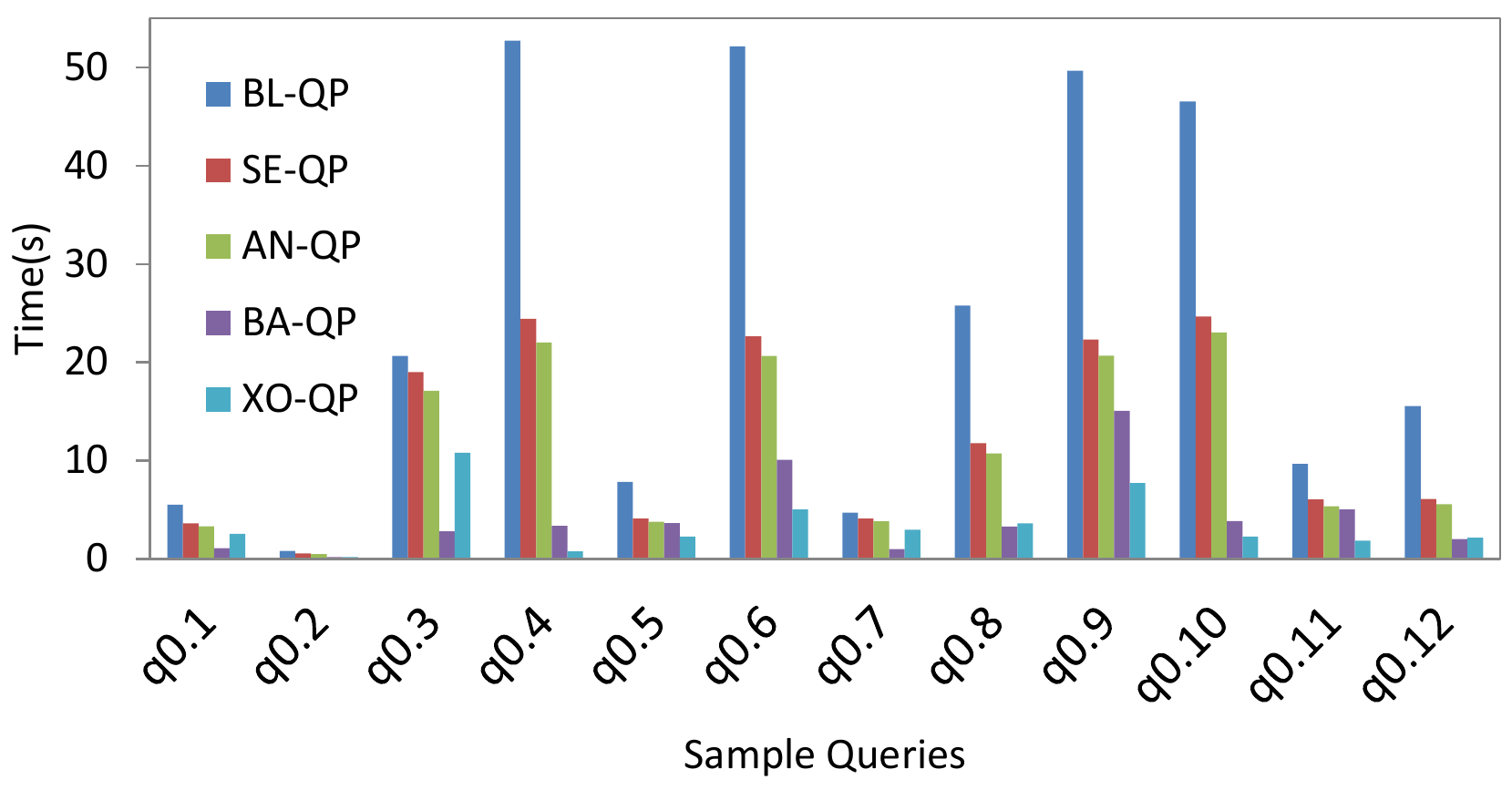}
}        
				\vspace{-1ex}
				\caption{Processing time of the sample queries for computing top-10 results.}
				\label{fig:ProcessingTime}
\vspace{-3ex} 
\end{figure*}

\subsection{Efficiency}
\label{sec:efficiency}
To demonstrate the efficiency of our approach, we use a baseline method which applies inter query pruning only, and uses scan eager for query processing. We call this baseline query processing as BL-QP. The purpose of this baseline is to demonstrate the efficiency of the intra-query pruning scheme in our proposed methods. Overall, we compare the performance of the following methods: (a)BL-QP, (b) SE-QP, (c) AN-QP, and (d) BA-QP, and (e) XO-QP.

\subsubsection{Processing Time}
Fig. $\ref{fig:ProcessingTime}$ shows the response time of computing the top-10 semantically related results for the sample queries (presented in Table \ref{tbl:sample}) on the IMDB and DBLP datasets. According to the results, BA-QP achieves the best performance. The reason behind this is that we apply inter and intra batch prunings in BA-QP in addition of inter and intra-query prunings. Also, we share the partial results of the shared keywords among the candidate queries in a batch. On average, BA-QP consumed 30 percent time of the time needed by the SE-QP and AN-QP methods. The difference, however, depends on the number of candidate queries. In Fig. $\ref{fig:ProcessingTime}$ (a), the response time for processing the test query $q_0.1$ in BA-QP is very close to those of SE-QP and AN-QP due to the small number of candidate queries for $q_0.1$ which is 12 and insufficient number of batches (in this case we have only one batch) to apply inter-batch pruning. Moreover, if the cost of merging the shared part $\mathcal{K}^s$ with the unshared part $\mathcal{K}^u$ is relatively high in most of the batches, BA-QP may not outperform the other methods as we see in $q_0.5$ on DBLP. In IMDB $q_0.5$, however, AN-QP outperforms SE-QP by a large margin because the data distributions in most of the inverted lists are skewed. In such case, many nodes are skipped in AN-QP, which improves the performance. Clearly, BL-QP has the worst performance on most cases because it does not apply the intra query pruning, therefore, incurs unnecessary computations. The XO-QP, however, is comparable to BA-QP	in many cases because we set the node score threshold to $0.9$ which means that only the relevant nodes to the keywords with the similarity degree bigger than $0.9$ are selected. If we set this value to a smaller number, XO-QP performance deteriorates and gets closer to AN-QP and SE-QP. Moreover, XO-QP builds a special inverted list called XOnto-DIL for the keywords which takes additional time and space for its creation while our proposed methods do not use such indexing and use the normal DIL for query processing.      

\begin{figure}
\centering       
			
        \includegraphics[width=2.5in,height=1.2in]{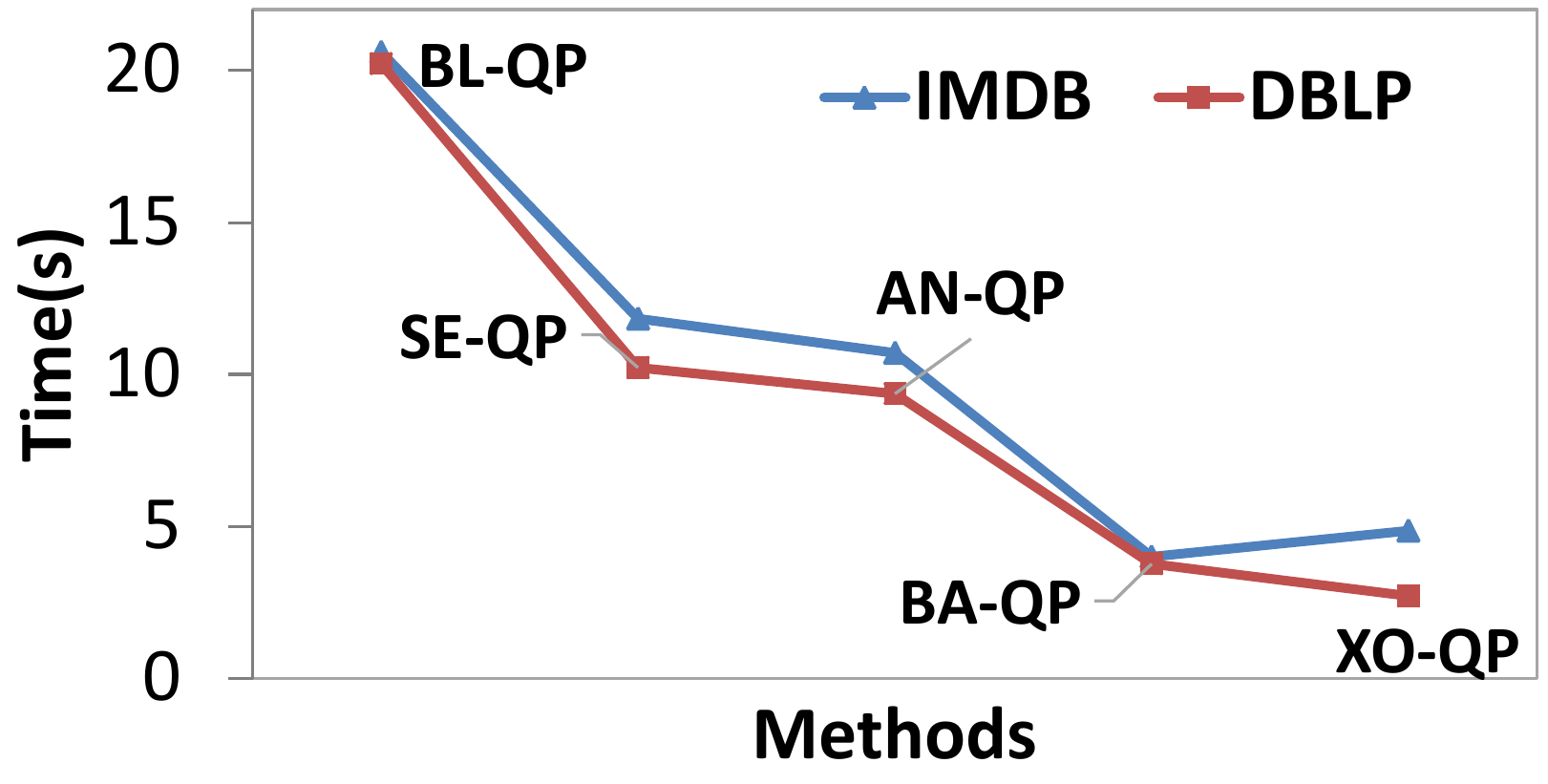}
				\vspace{-1ex}
        \caption{Average processing time on all test queries}     
				
\label{fig:TotalProcessingTime}
\vspace{-3ex}
\end{figure}

In Fig.\ref{fig:TotalProcessingTime}, the average processing time of different methods for all test queries are presented. Clearly the baseline method which does not use intra query pruning spends the maximum time for processing on both datasets. The SE-QP and AN-QP spend less time comparing to BL-QP because we apply inter and intra query pruning and therefore, the processing terminates early when we reach to the global $\sigma^{min}$. Moreover, there is a narrow improvement in the AN-QP over SE-QP due to using the anchor node processing which skips many redundant computations and expedites the efficiency. The BA-QP processing improves sharply on both datasets because in BA-QP, we execute queries in batch and share the computations among the queries in the batch and thereafter, we reach to the global $\sigma^{min}$ before SE-QP and AN-QP can do. Also, we can apply inter and intra-batch pruning which expedite its efficiency. In XO-QP, however, the performance is close to BA-QP on both datasets because we build the XOnto-DIL on the nodes with the relevance degree no less than $0.9$. In such condition, only the highly relevant nodes are selected to replace the keywords, therefore, the processing time does not grow considerably due to small number of candidate keywords. In contrast, by setting the threshold to lower numbers, the processing time will grow exponentially and gets closer to SE-QP. Furthermore, our methods do not need to build special indexing, therefore, avoid using additional space and offline processing time for building such indexing.                                          
\subsubsection{Effect of Query Length }

In this experiment, we choose test queries that have at least 100 candidate queries. At each step, we set their length to a number of settings ($\{3,4,5,6\}$) and compute the average response time of all queries when processing 100 candidate queries while setting $k$ to 10. In each step, for the test queries that have smaller number of keywords, we add additional keywords to increase their length. We also use the same queries to conduct experiments and compare the results with XO-QP method. In Fig. $\ref{fig:QueryLength}$ the effect of growing the length of queries on the response time is analyzed. The response time of BA-QP is almost fixed compared to XO-QP, BL-QP, SE-QP, and AN-QP on IMDB. Similarly in DBLP, BA-QP has the slowest growth in response time when the number of keywords increases while BL-QP, SE-QP and AN-QP show a big jump in processing times. Also, XO-QP shows a jump in the processing time when the query length increases. Thus, the performance of XO-QP is sensitive to the query length parameter. The sharpest increase in the processing time is related to BL-QP because it does not apply intra-query pruning, therefore, all the inverted lists are accessed during the processing and this incurs many useless computations specifically when the query length increases.  

In summary, we can conclude that BA-QP method is not sensitive to the number of keywords due to sharing the computations while the performance of SE-QP and AN-QP is highly sensitive to the query length.   

\begin{figure}

\centering
        \subfloat[IMDB]{\label{IMDB}
  \includegraphics[width=1.5in,height=1.2in]{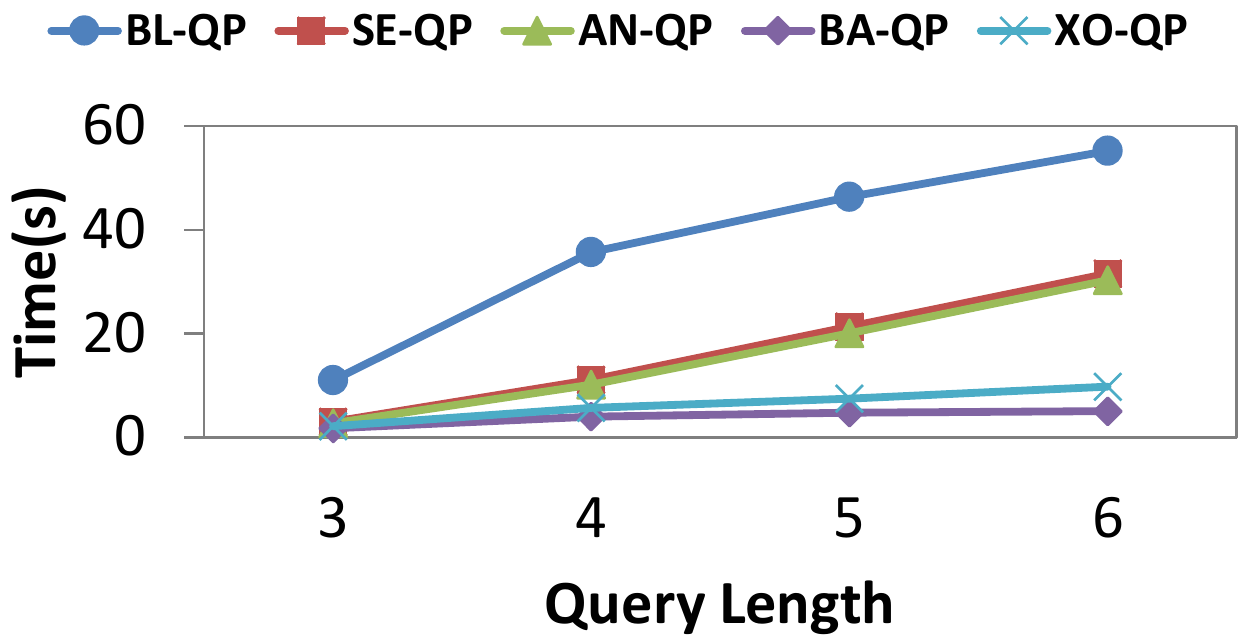} 
}
\hfill
\subfloat[DBLP]{\label{DBLP}
    \includegraphics[width=1.5in,height=1.2in]{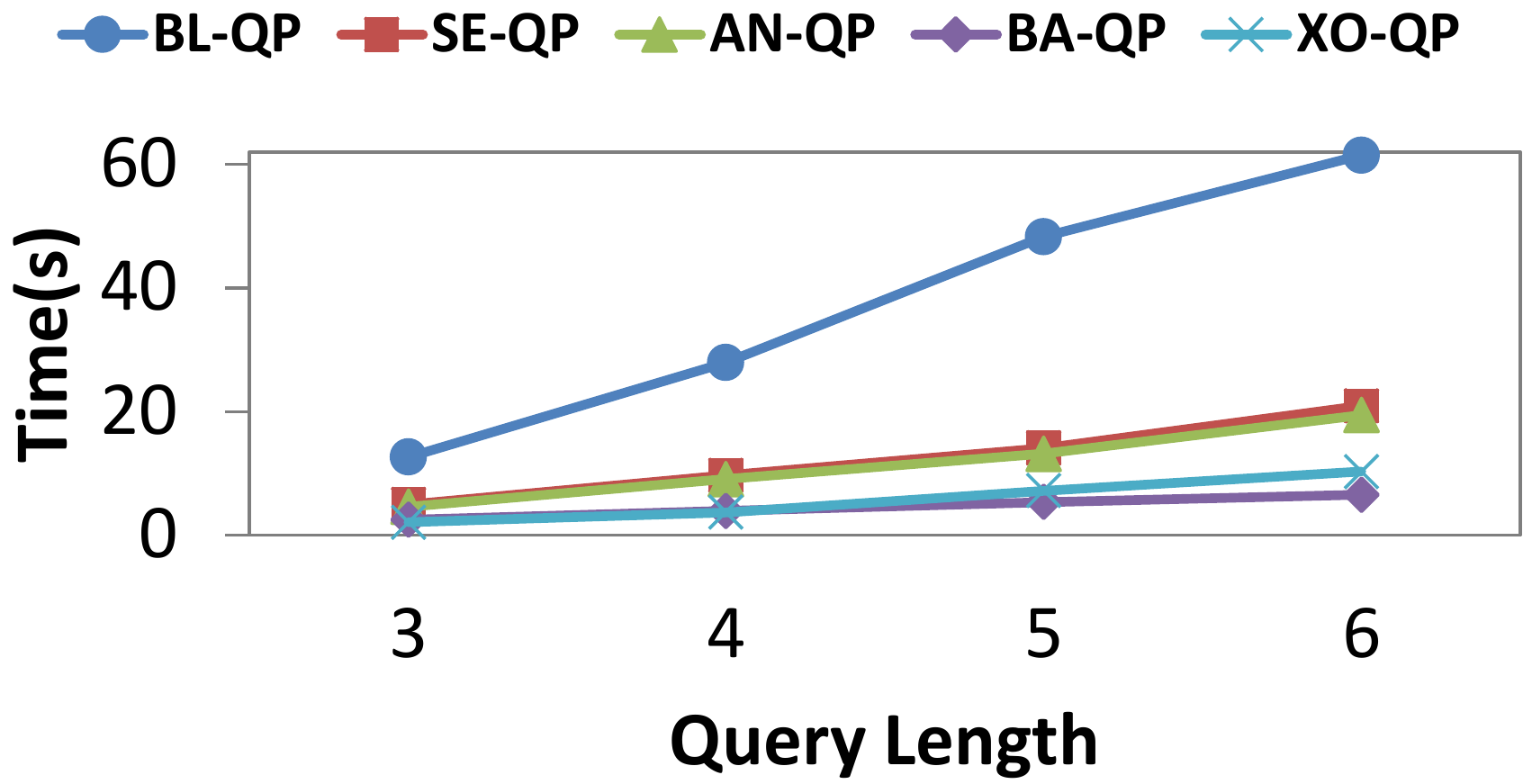}
}        
				\vspace{-1ex}
				\caption{Average processing time of the test queries with varying query length (i.e., number of keywords).}
				\label{fig:QueryLength}
\vspace{-3ex} 
\end{figure}

\vspace{-2ex}
\subsubsection{Effect of Top-$k$ List Size $k$}
In this experiment, we vary the top-$k$ size $k$ and compute the average processing time. In Fig. $\ref{fig:Top-k}$, we observe the effect of top-$k$ size $k$ on the average processing time for different methods. Clearly, the BA-QP and XO-QP methods are not that sensitive to the value of $k$. By increasing the value of $k$, processing times of BA-QP and XO-QP show only a small increase or almost fixed. On the other hand, for SE-QP and AN-QP methods, any increase on $k$, leads to a considerable jump on the average processing time. However, this increase is not big when we change $k$ from 10 to 20. When $k$ is selected as a bigger number, $\sigma^{min}$ will be smaller and this causes the inter query pruning in SE-QP and AN-QP to be less effective. That is, the application of inter-query pruning is delayed in SE-QP and AN-QP for big $k$. 
However in BA-QP, we execute a candidate query batch by sharing the computations among the queries in it and thereafter, we reach to the global $\sigma^{min}$ before SE-QP and AN-QP can do. Also, we can apply inter and intra-batch pruning with this early found global $\sigma^{min}$ which helps BA-QP to expedite its efficiency. 
       
\begin{figure}

\centering
        \subfloat[IMDB]{\label{IMDB}
   \includegraphics[width=1.5in,height=1.2in]{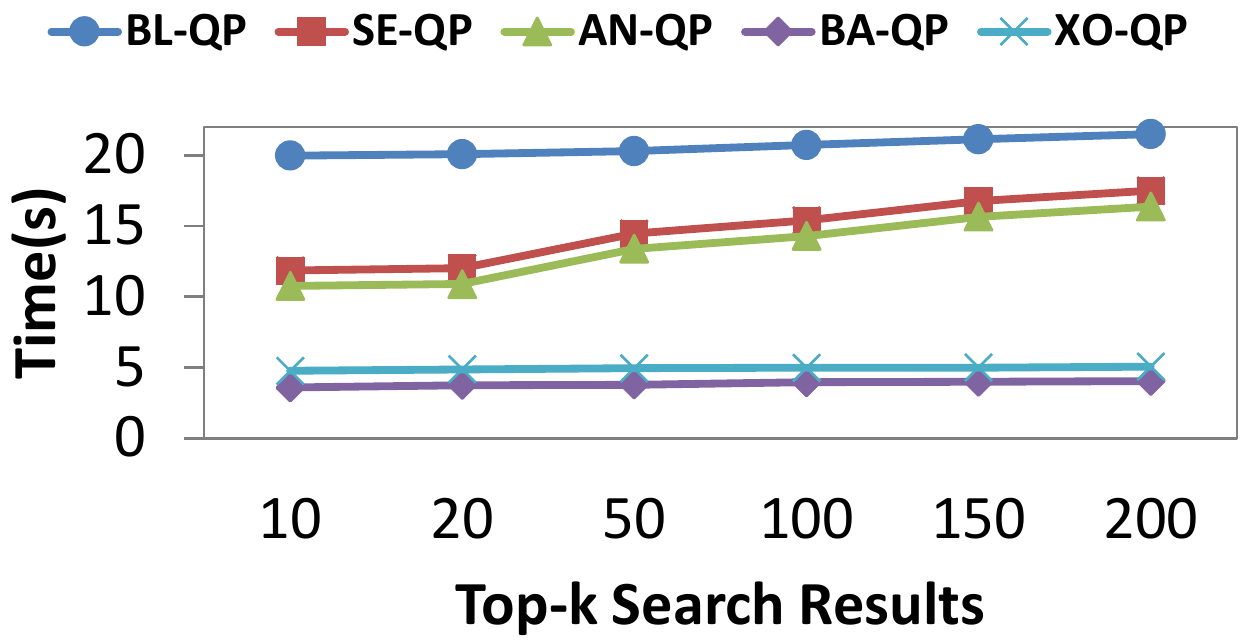}
}
\hfill
\subfloat[DBLP]{\label{DBLP}
     \includegraphics[width=1.5in,height=1.2in]{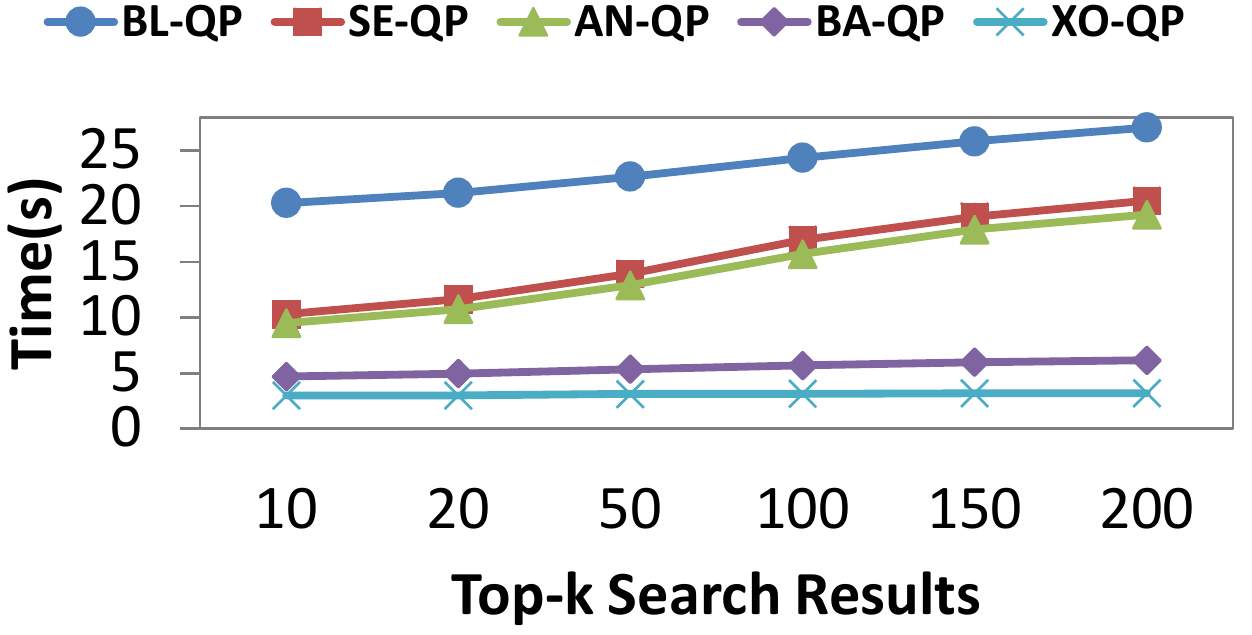}
}        
				\vspace{-1ex}
				\caption{Average processing time of the test queries with varying $k$ in top-k list.}
				\label{fig:Top-k}
\vspace{-3ex}
\end{figure}

\subsubsection{Effect of Tuning Parameter $\alpha$}
In this experiment, we vary $\alpha$ between a number of settings ($\{2,4,8,16\}$) and compute the average processing time of the test queries when $k = 10$. Fig. $\ref{fig:Tuningparameter}$ presents the effect of choosing different values for $\alpha$. Larger values of $\alpha$ reduce the sensitivity to data cohesiveness. This usually leads to a more effective intra-query pruning and decreases the processing time. On the contrary, as $\alpha$ decreases, the possibility for the partial results scores to be smaller than the $\sigma^{min}$ also decreases. In this case, the intra-query pruning becomes less effective and therefore, the processing time increases, specifically, in SE-QP and AN-QP. In BL-QP, however, there is no significant difference in the processing time when $\alpha$ changes to smaller number. This is because BL-QP does not apply intra query pruning and therefore, the processing time is not affected that much by $\alpha$.   

\begin{figure}

\centering
        \subfloat[IMDB]{\label{IMDB}
   \includegraphics[width=1.5in,height=1.1in]{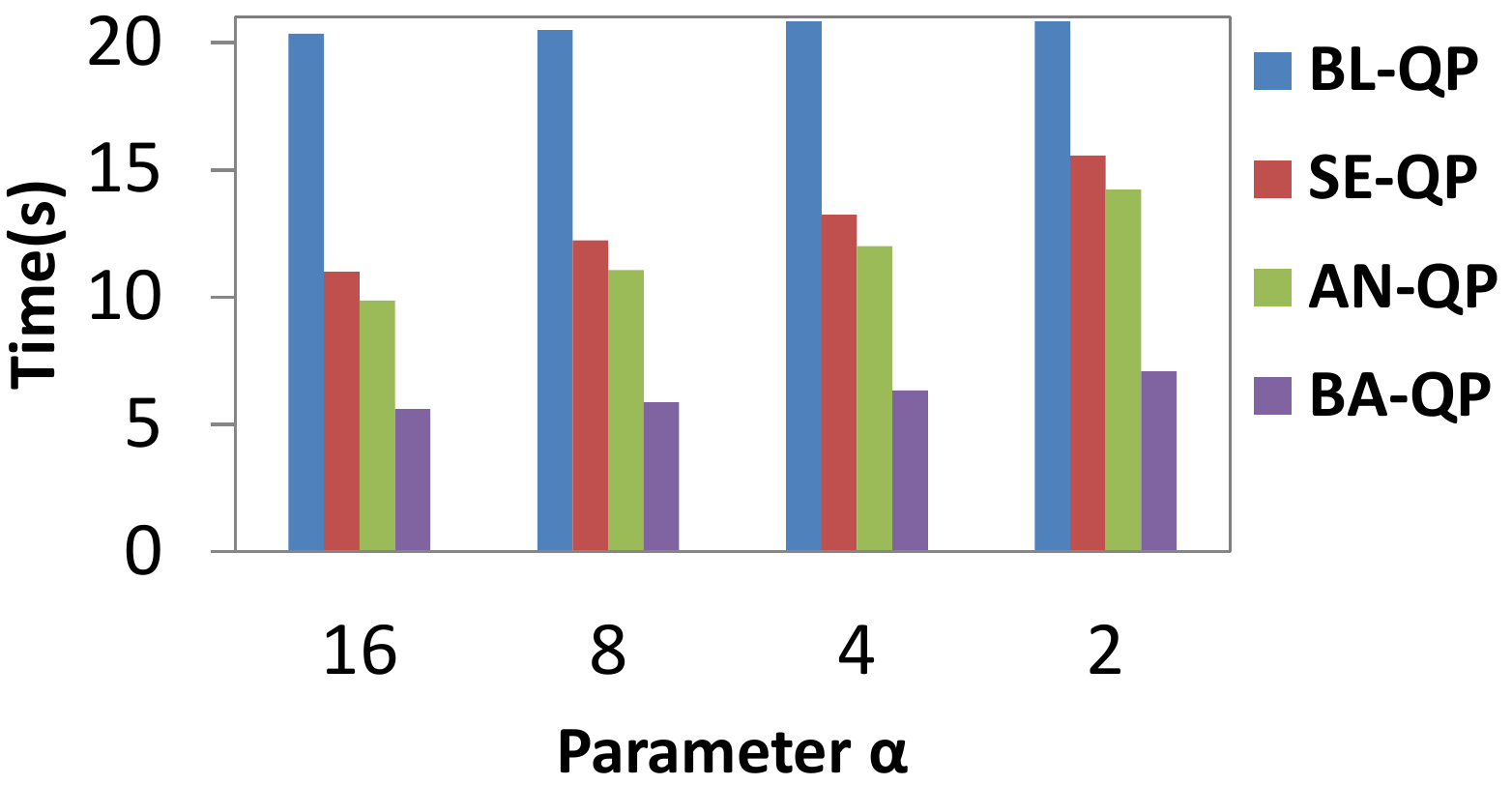}
}
\hfill
\subfloat[DBLP]{\label{DBLP}
     \includegraphics[width=1.5in,height=1.1in]{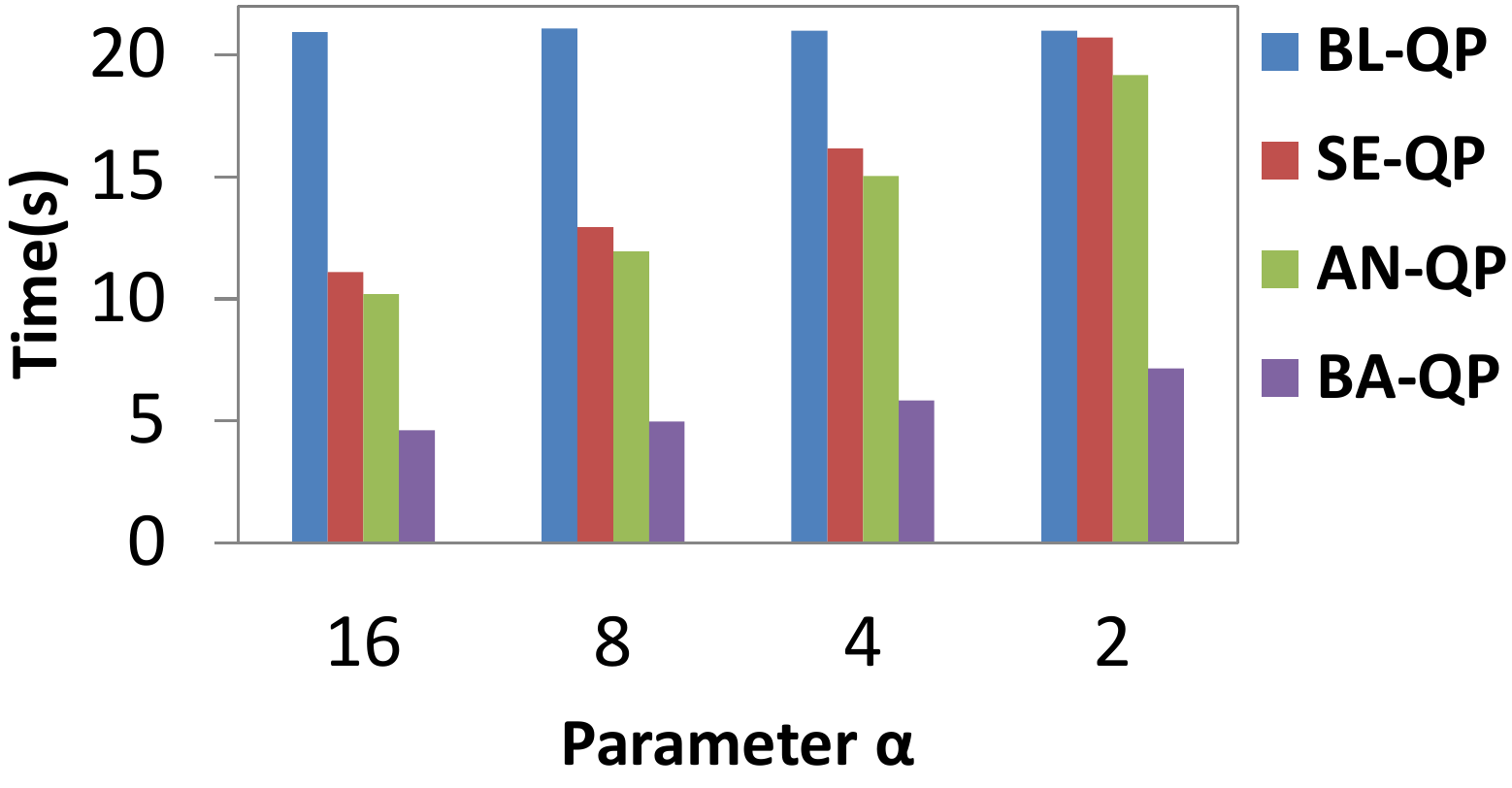}
}        
				\vspace{-1ex}
				\caption{Average processing time of the test queries with varying tuning parameter $\alpha$.}
				\label{fig:Tuningparameter}
\vspace{-3ex} 
\end{figure}

\subsubsection{Effect of Candidate Queries Number $|\mathcal{Q}|$}
For this experiment, we choose test queries which have at least 200 candidate queries. At each step we process a certain number ($\{40,80,120,160,200\}$) of their candidate queries and compute the average processing time when $k = 10$. In Fig. $\ref{fig:QueriesNumber}$, the effect of growing the number of candidate queries $|\mathcal{Q}|$ on the response time is analyzed. The response time of BA-QP method grows slowly compared to BL-QP, SE-QP, and AN-QP which show a sharp rise in each step. In BL-QP, the growth in the query processing time is the maximum among the methods because it does not apply intra query pruning, therefore, it incurs many unnecessary computations during execution of the candidate queries. We can conclude that BA-QP is not that sensitive to $|\mathcal{Q}|$. This is because BA-QP not only shares the computations among the candidate query batch but also applies inter and intra batch pruning with the early-found global $\sigma^{min}$ to expedite its performance.

\begin{figure}

\centering
        \subfloat[IMDB]{\label{IMDB}
   \includegraphics[width=1.5in,height=1.2in]{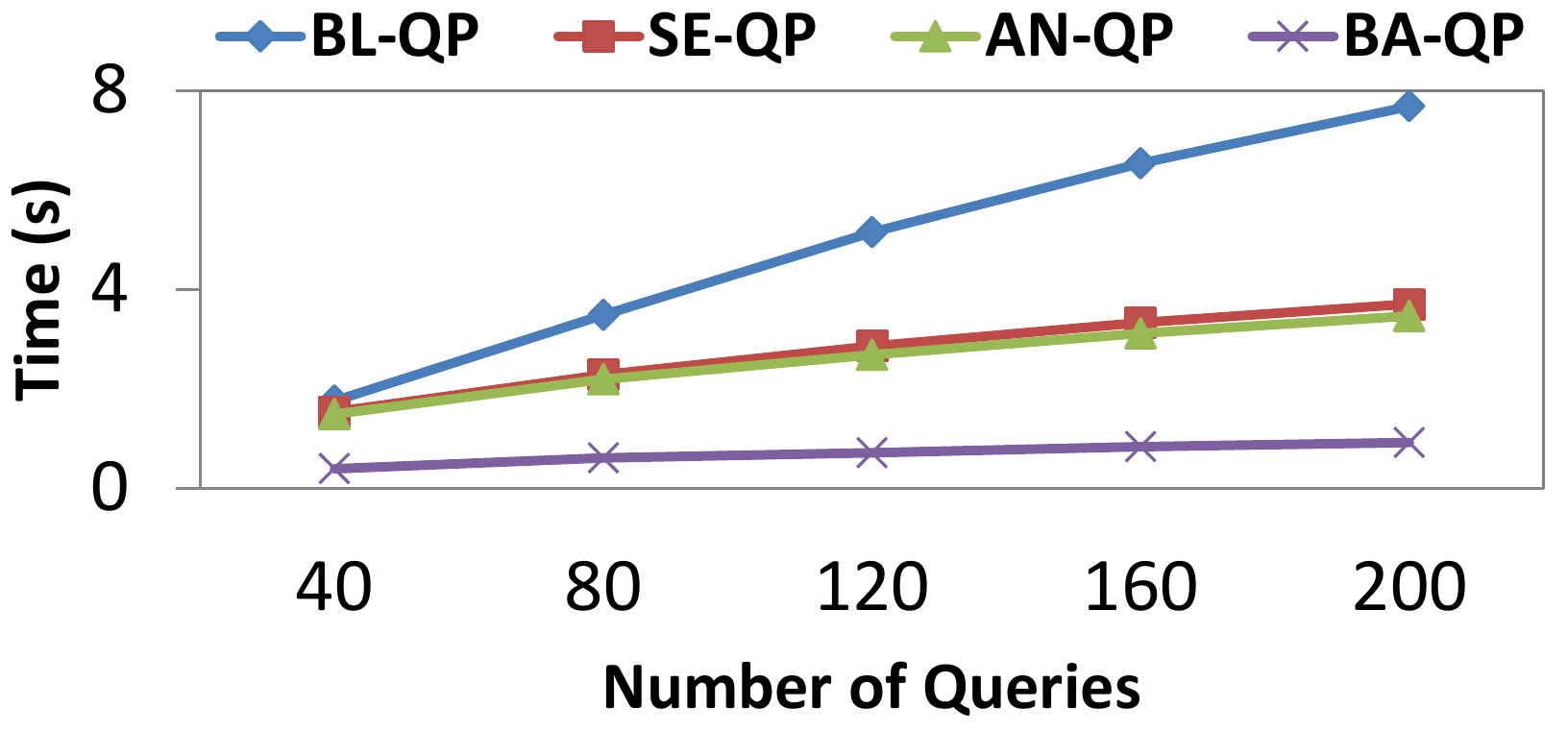}
}
\hfill
\subfloat[DBLP]{\label{DBLP}
     \includegraphics[width=1.5in,height=1.2in]{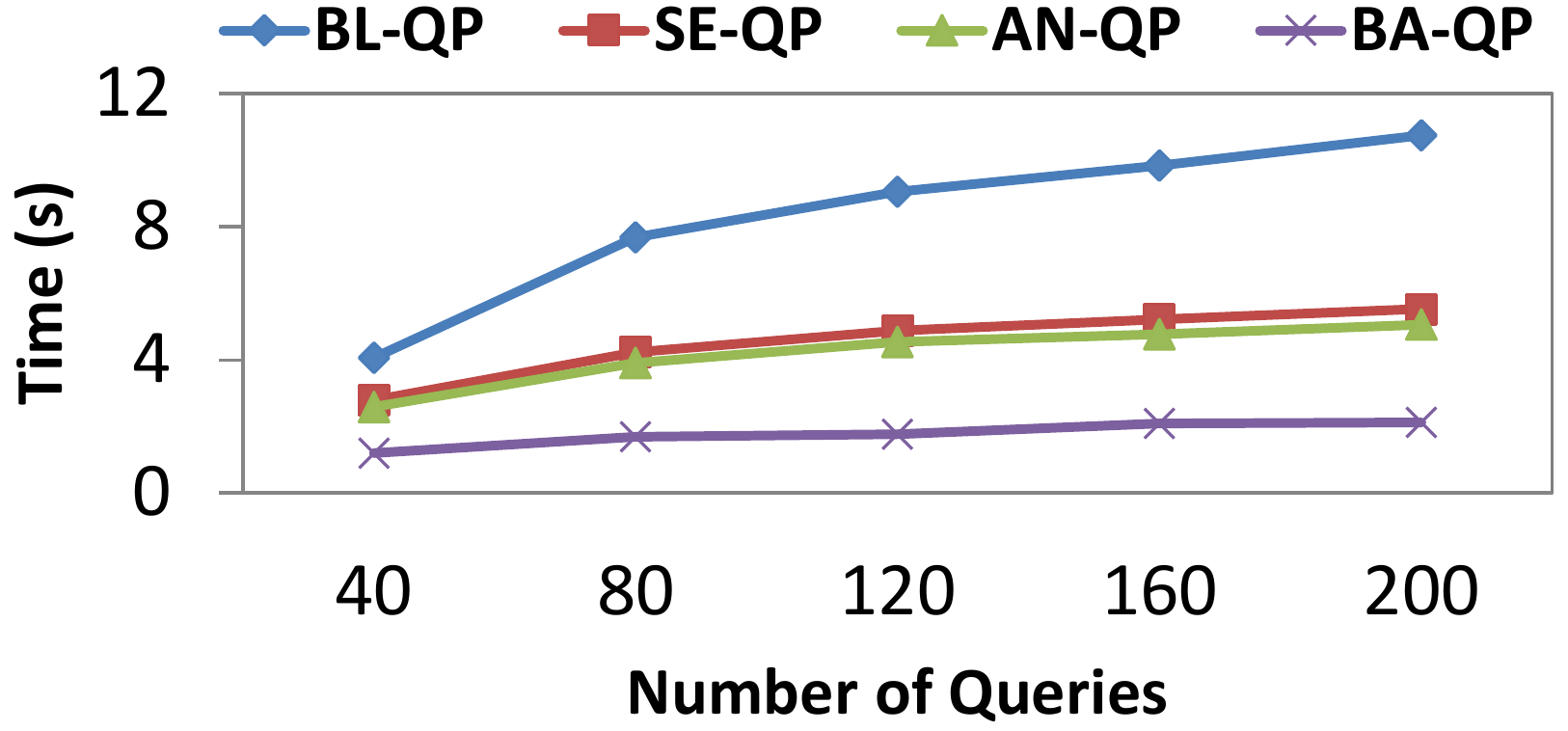}
}       
				\vspace{-1ex}
				\caption{Processing time with varying number of candidate queries $|\mathcal{Q}|$.}
				\label{fig:QueriesNumber}
\vspace{-3ex} 
\end{figure}
      
\begin{figure*}

\centering
        \subfloat[IMDB]{\label{IMDB}
   \includegraphics[width=3.3in,height=2.1in]{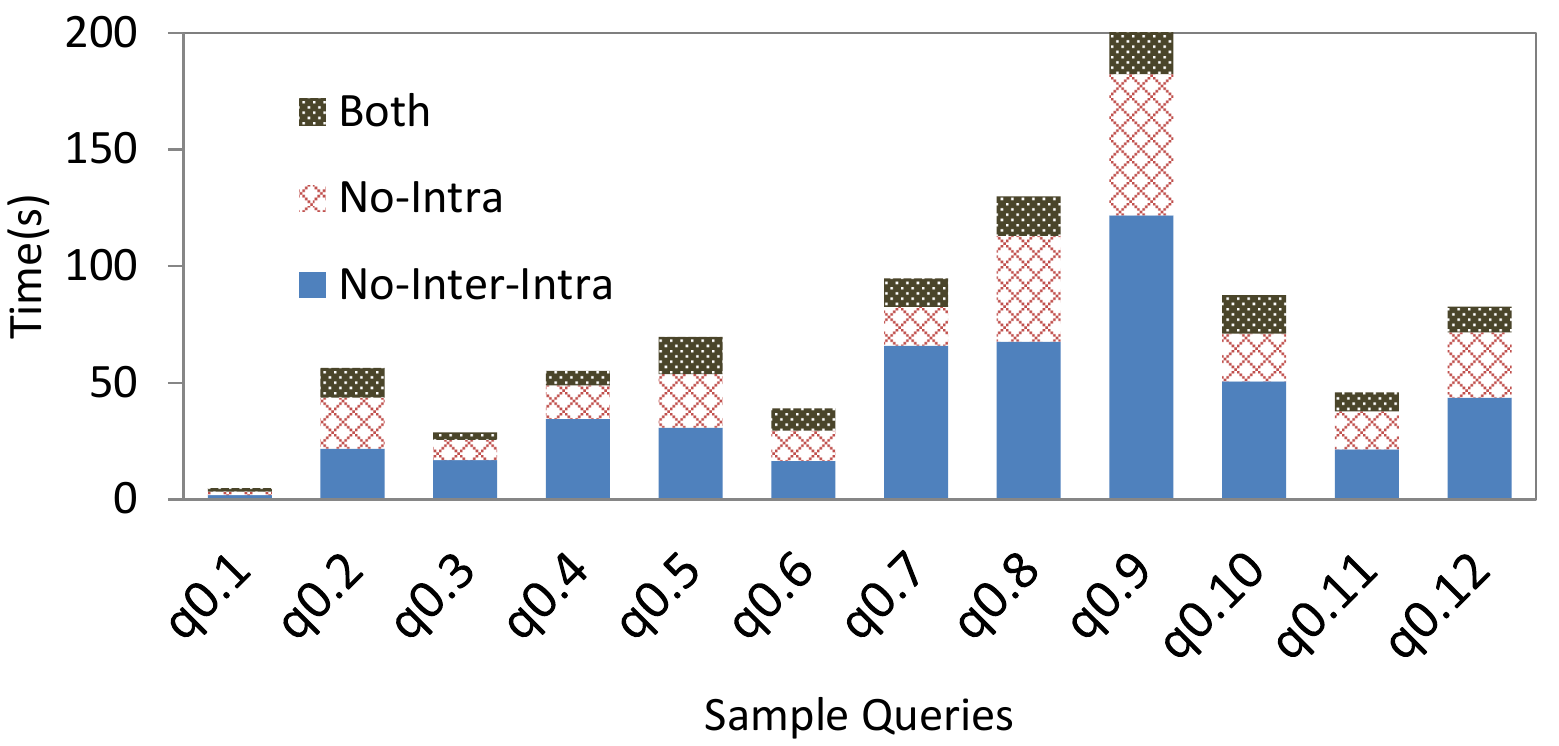}
}
\hfill
\subfloat[DBLP]{\label{DBLP}
     \includegraphics[width=3.3in,height=2.1in]{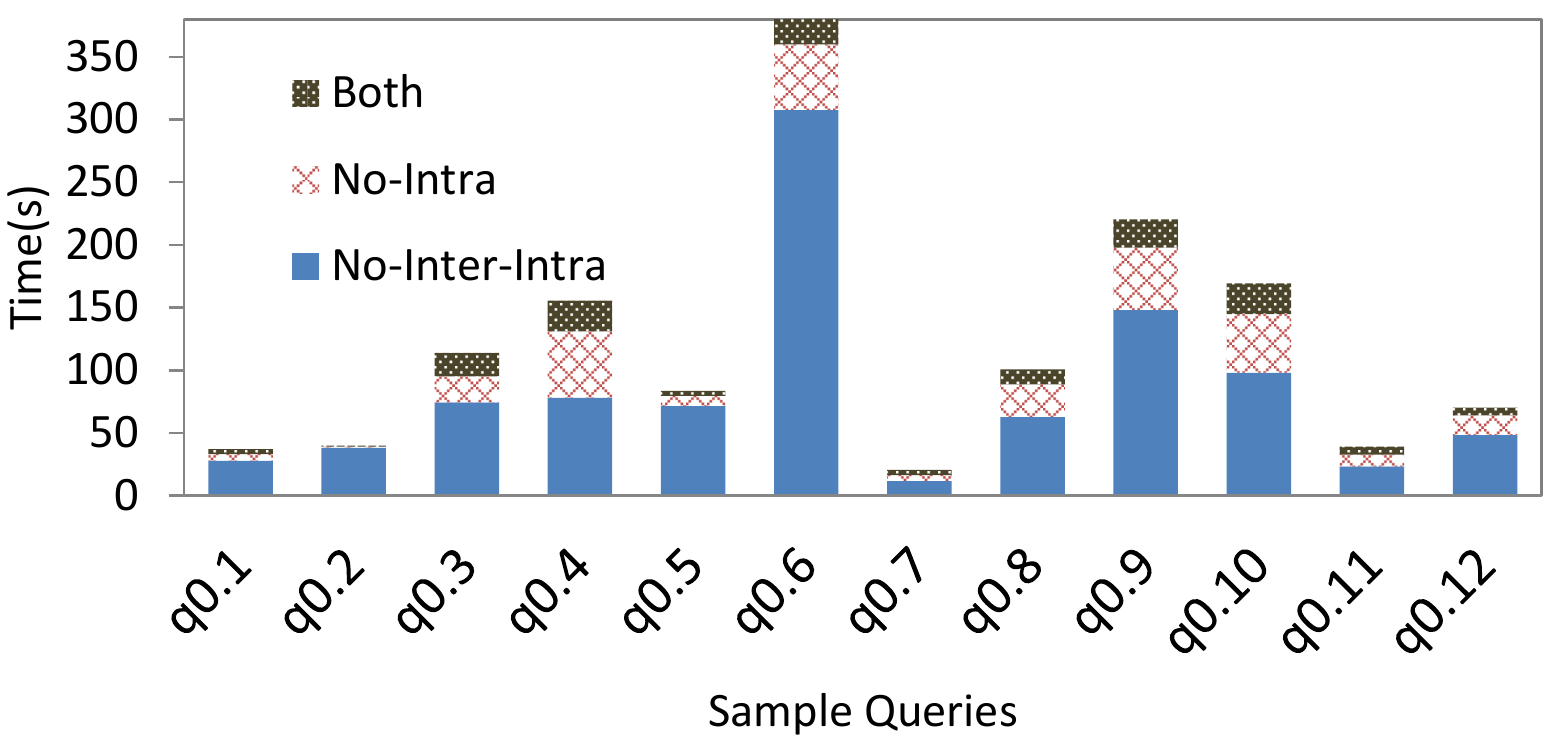}
}      
				\vspace{-1ex}
				\caption{Inter and Intra Query Pruning Improvement on Sample Queries.}
				\label{fig:PruningTime}
\vspace{-3ex} 
\end{figure*}

\subsubsection{Pruning Improvement}
In this section, we show the pruning effect on the processing time. From Fig. \ref{fig:PruningTime}, the sample queries processing time are shown for 3 scenarios: (a) when there is no pruning for processing the queries, (b) when only the inter query pruning is implemented, and (c) when both inter and intra pruning methods are implemented. The processing time are measured using SE-QP method. Clearly, the inter query pruning shows the most effective method to cut the processing time on most of the sample queries. If the number of candidate queries is large and the breaking point occurs when most of the queries are not executed, then inter query has the best performance. e.g., in the sample queries $q_{0.9}$ and $q_{0.10}$ on IMDB, the number of executed queries are $\frac{225}{442}$ and $\frac{4918}{12240}$ respectively and in the sample queries $q_{0.6}$ and $q_{0.9}$ on DBLP, the number of executed queries are $\frac{912}{8190}$ and $\frac{1287}{3276}$ respectively. Therefore in these cases, most of the queries are not executed by using inter query pruning and the processing time reduced sharply. The intra query pruning is more effective when the number of query keywords is bigger or the inverted list that is not accessed due to pruning is big sized. In such condition, some inverted lists are not accessed when the result is not able to beat the $\sigma^{min}$ and this expedites the processing time. For example, in the sample queries $q_{0.8}$ and $q_{0.9}$ on IMDB and $q_{0.4}$ and $q_{0.9}$ on DBLP, the query keywords are from 4 to 5 keywords and include some big sized inverted lists that are not accessed, therefore the processing time reduced considerably.\par 

\begin{figure}
\centering        
				
        \includegraphics[width=2.3in,height=1.4in]{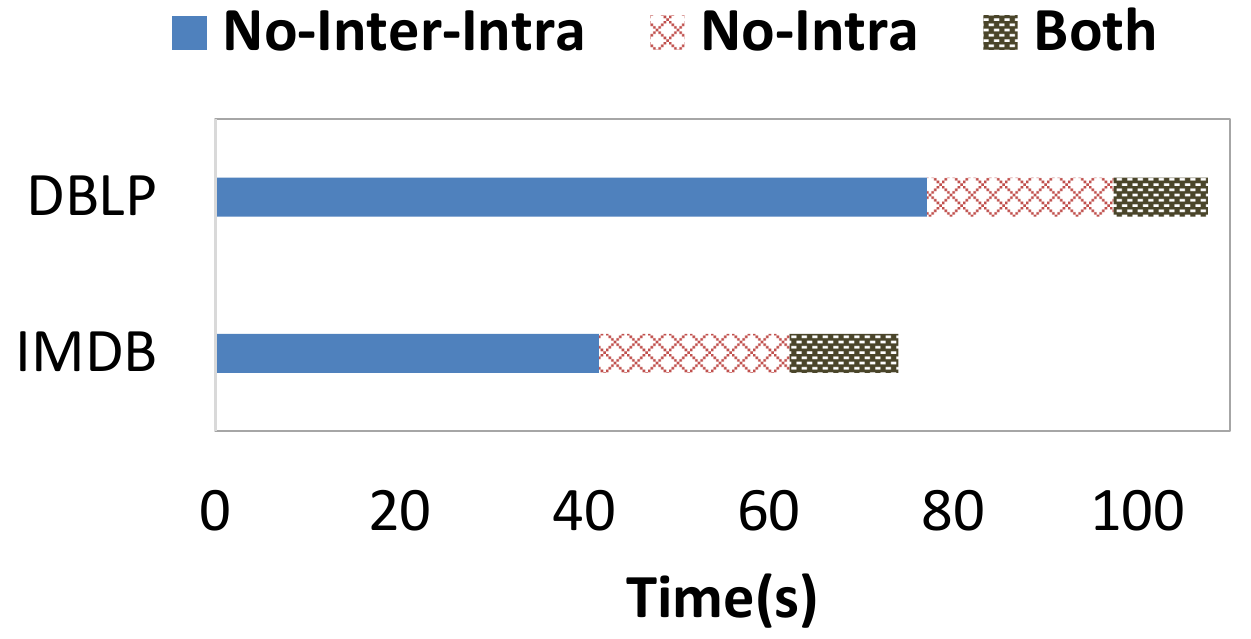}
       	\vspace{-1ex}
     		\caption{Average Pruning Improvement.}
\label{fig:TotalPruningTime}
\end{figure} 
      
Fig. \ref{fig:TotalPruningTime} presents the average processing time for the set of test queries for 3 scenarios: (a) no pruning is implemented, (b) only inter query pruning is implemented, and (c) both pruning techniques are implemented. Clearly, the processing time for the case with no pruning is maximum on both datasets. This shows that inter query pruning has the most tangible effect on the processing time by avoiding to execute the queries that cannot contribute to the $\mathcal{R}^*$. The efficiency improvement on IMDB and DBLP is 2 and 3 times respectively. After that, intra query pruning expedites the efficiency by avoiding to access all inverted lists when the result cannot beat $\sigma^{min}$.\par

Fig. \ref{fig:PruningTimeBAQP} presents the sample queries processing time for BA-QP in 2 scenarios: (a) when no batch pruning is implemented, (b) when batch pruning is implemented. In most cases, the processing time for the method which uses pruning has decreased. The batch pruning becomes more effective when the query length increases as shown in $q_0.3$ and $q_0.12$ on IMDB or in $q_0.4$ and $q_0.8$ on DBLP. Moreover, when the candidate queries contain some big inverted lists and we reach to the global $\sigma^{min}$ early, the improvement is more considerable as in $q_0.1$ and $q_0.2$ on IMDB and in $q_0.1$ and $q_0.7$ on DBLP.\par

\begin{figure*}

\centering
        \subfloat[IMDB]{\label{IMDB}
   \includegraphics[width=3.3in,height=2.1in]{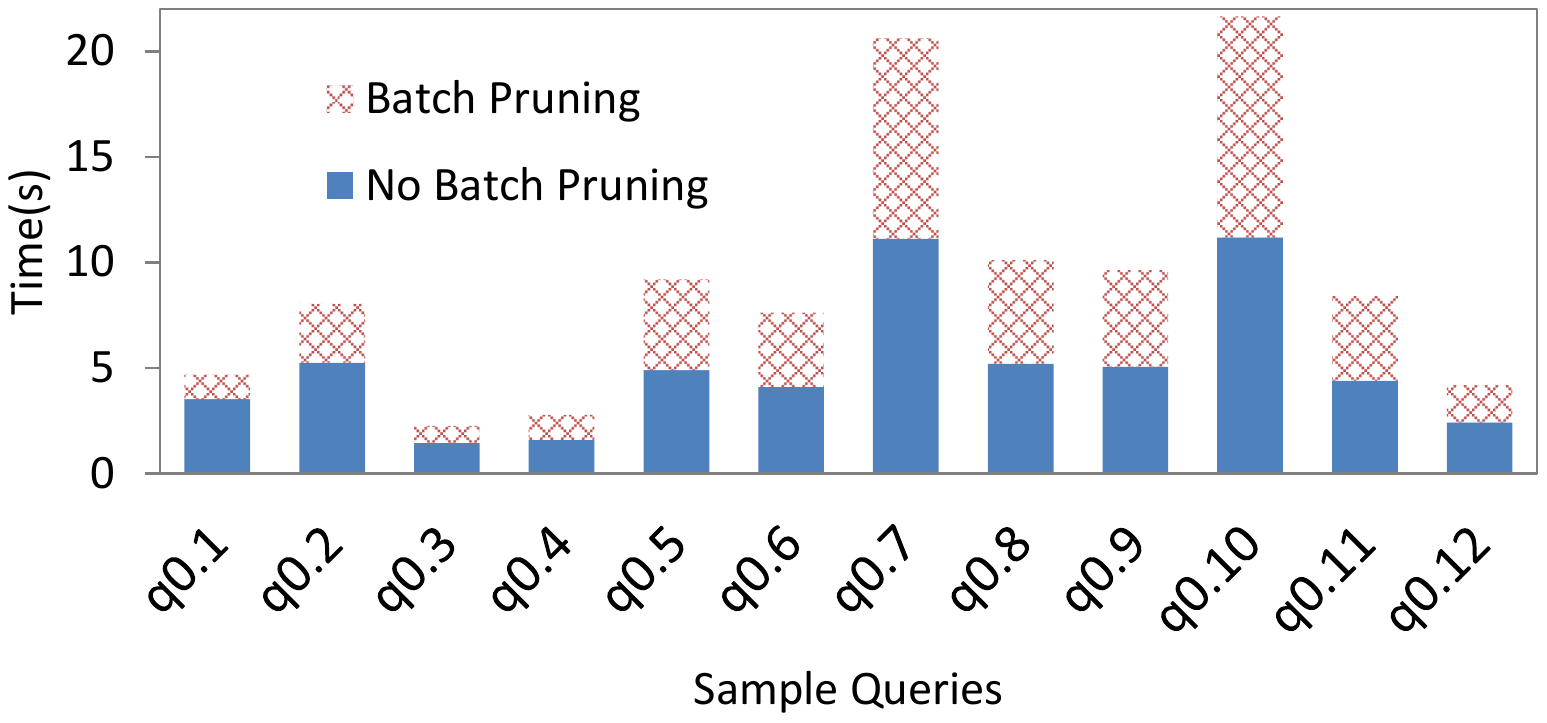}
}
\hfill
\subfloat[DBLP]{\label{DBLP}
     \includegraphics[width=3.3in,height=2.1in]{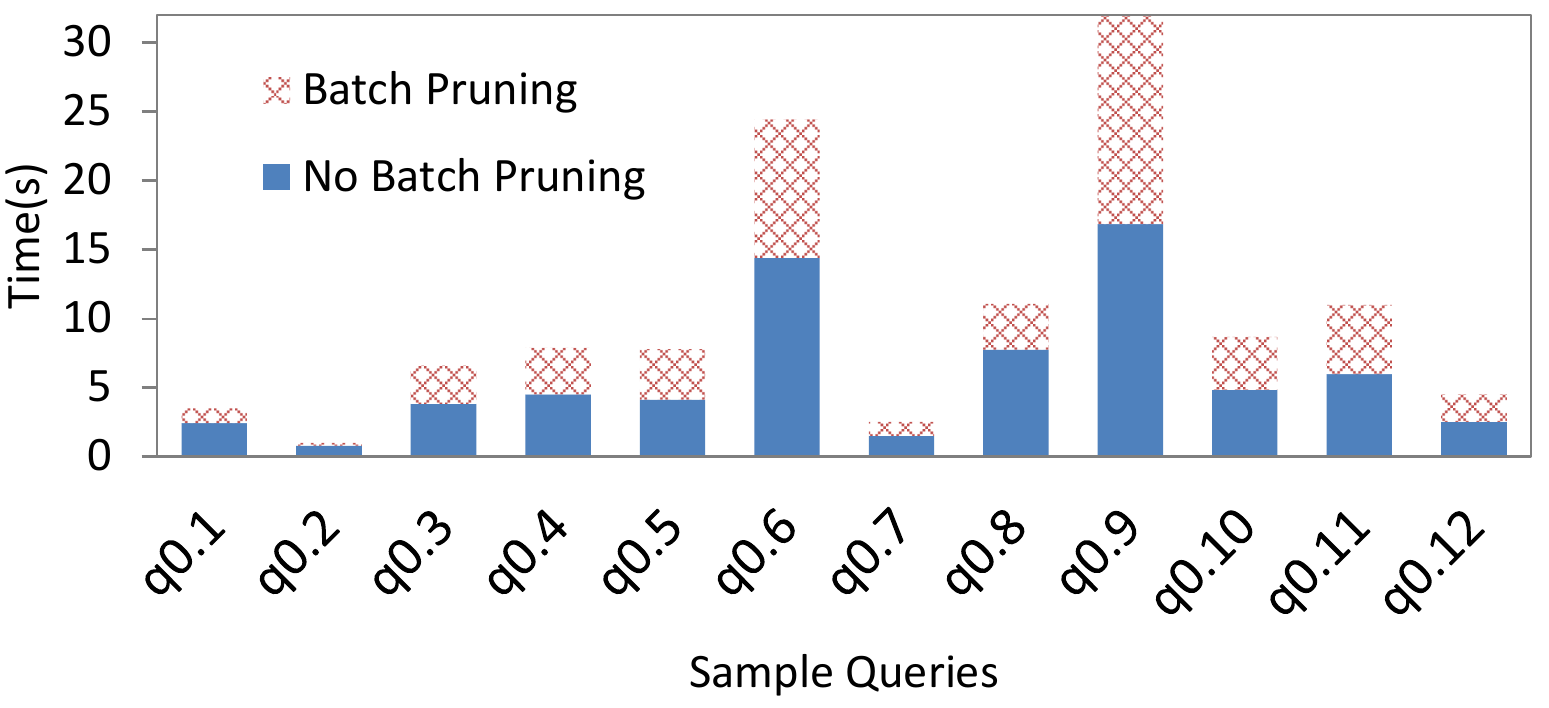}
}        
				\vspace{-1ex}
				\caption{Batch Pruning Improvement on Sample Queries.}
				\label{fig:PruningTimeBAQP}
\vspace{-3ex}
\end{figure*}

Fig. \ref{fig:TotalPruningTimeBAQP} shows the average processing time of the test queries for BA-QP on 2 cases: (a) when no batch pruning is implemented, (b) when batch pruning is implemented. We observe from the picture that the batch pruning improved the processing time on both datasets, however, the improvement on DBLP is more considerable. The improvement is achieved because we reach to the global $\sigma^{min}$ earlier by applying batch pruning, therefore it expedites the performance.                        

\begin{figure}
\centering        
				
        \includegraphics[width=2.3in,height=1.4in]{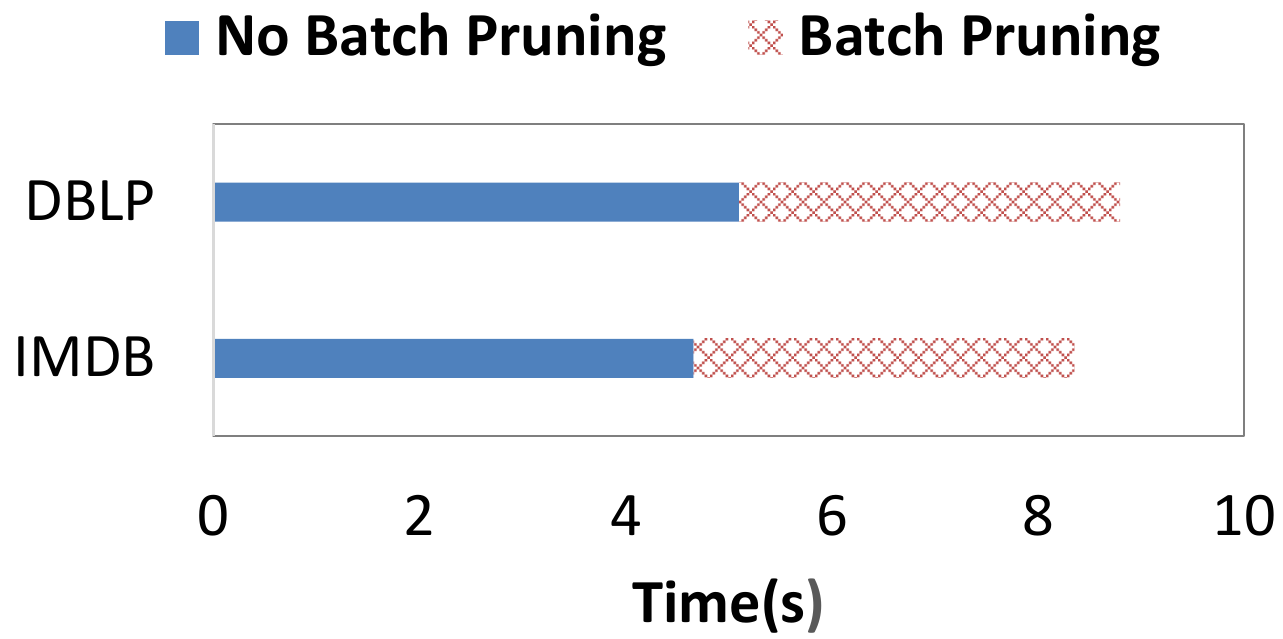}
				\vspace{-1ex}
     		\caption{Average Batch Pruning Improvement.}
\label{fig:TotalPruningTimeBAQP}
\vspace{-3ex}
\end{figure} 

\section{Related Work}
\textbf{Failed Queries.}
When a user queries a data source, the result may be empty or otherwise below expectation. This problem known as failed queries has inspired a broad range of research in the database community (e.g. \cite{ChuYCMCL1996},\cite{HuhMA10},\cite{ZhouGBN07}). In the context of relational databases, the problem has been studied by Nambiar and Kambhampati \cite{NAMBIARK06}, Muslea \cite{MUSLEA04} as well as Muslea and Lee \cite{MUSLEAL05}. Nambiar and Kambhampati \cite{NAMBIARK06} presented approximate functional dependencies to relax the user original query and find tuples similar to the user query. Muslea \cite{MUSLEA04}, as well as Muslea and Lee \cite{MUSLEAL05} used machine learning techniques to infer rules for generating replacement queries. Amer-Yahia, Cho and Srivastava \cite{AMERCS02}, Brodianskiy and Cohen \cite{BRODIANSKIYC07} as well as Cohen and Brodianskiy \cite{COHENB09} studied query relaxation for XML data. The studies proposed to discover the constraints in the queries that prevent results from being generated and remove them so that a result can be produced. All these investigations focus on the modification of the user original query constraints on the content level rather than the semantic analysis of the original query constraints. Hill et al. \cite{HillTGC10}, used the ontology information to relax structured XML queries. Farfan et al. \cite{FarfanHRW09}, proposed XOntoRank system to address the ontology-aware XML keyword search of electronic medical records. Unlike their work which uses SNOMED ontology for enhancing the search on medical records, we address the general no-match problem on XML data and use a general ontological knowledge base like a thesaurus or dictionary to solve the problem for general documents.          
 
\textbf{Query Expansion.}
Query expansion has widely studied in many works such as \cite{XuC1996},\cite{SchenkelTW05},\cite{KimK05},\cite{KimKJ06}. Schenkel, Theobald, and Weikum \cite{SchenkelTW05} proposed XXL which combines the keyword search with structural conditions and semantic similarity to increase the quality of results. Kim and Kong \cite{KimK05} suggested a query expansion technique that uses an ontology algorithm to map a target DTD to ontology. This scheme is successful for expanding the queries minimally. Kim, Kong, and Jeon \cite{KimKJ06} developed a web XML document search engine that applies ontology-DTD match algorithm for remote documents. However, in all of the above works, the focus is on structured queries. In our work, we find some semantic counterparts for specific non-mapped keywords for replacement, therefore, query expansion is not useful in our case. 
        
\textbf{Recommendation Systems.}
Users are often interested in items similar to those they have visited before or to content that has been looked up by similar users. These items are presented by the recommendation systems. Akbarnejad et al. \cite{AkbarnejadCEKMOPV10} and Chatzopoulou, Eirinaki and \cite{ChatzopoulouEP09} proposed query recommendation based on a prediction of the items that user is interested in. Yao et al. \cite{YaoCHH12} proposed to exploit structural semantics for query reformulation. Meng, Cao and Shao \cite{MengCS14} used the semantic relationships between keywords and keyword queries to suggest a set of keyword queries from the query log. However, the semantic relationship is interpreted as the co-occurrence of the keywords and no ontological analysis is carried out. Moreover, the work focuses on extracting similar queries from the query log using data mining techniques without processing the results. Drosou and Pitoura \cite{DrosouP13} presented a database exploration framework which recommends additional items called ``You May Also Like" results. However, the recommended results are compiled based on the results of the original query and there is no focus on semantic connection between the original query and recommended results.

\textbf{Mismatch Problem.} 
Sometimes the system shows erroneous mismatch results for a user query which is called mismatch problem. Bao et al. \cite{BaoZLZLJ15} proposed a framework to detect the keyword queries that lead to a list of irrelevant results on XML data. They detect a mismatch problem by analyzing the results of a user query and inferring the user's intended node type result based on data structure. Based on this, they are able to suggest queries with relevant results to the user. Unlike the current study, Bao et al. investigate ways of producing relevant results instead of finding results for no-match queries.        

\textbf{Query Cleaning.}
Sometimes the empty result is caused by typographical errors. Pu and Yu \cite{PuY08} and Lu et al. \cite{LuWLL11} investigated a way of suggesting queries that have been cleaned of typing errors. Unlike our study, these authors do not tackle the problem of non-mapped keywords.

\textbf{Ontology-based Querying.}
Many studies have used ontology information for searching the semantic web \cite{CorbyDFG06}, \cite{LittleSL08}, \cite{MeiMP06}. Studies by Aleman-Meza \cite{Aleman-MezaHSSA05},  Cakmak and {\"{O}zsoyoglu} \cite{CakmakO08} as well as Wu, Yang and Yan \cite{WUYY13} used ontology information to find frequent patterns in graphs. Wu, Yang and Yan \cite{WUYY13} proposed an improved subgraph querying technique by ontology information. They revised subgraph isomorphism by mapping a query to semantically related subgraphs in terms of a given ontology graph. Our work generates substitute queries for the user given keyword query by extracting the semantically related keywords from the ontological knowledge base and thereafter, produce semantically related results to the user query instead of returning an empty result set to the user. 

\section{Conclusion}
This paper investigates ways of efficiently building substitute queries against XML data sources when the user given keyword query fails to produce any result as one or more of its keywords do not exist in the data source. Our approach depends on an ontological knowledge base for a discovery of semantically related keywords to generate the substitute queries, which can be executed against the data source to produce the semantically related results for the user's original query. As the number of substitute queries can be potentially large and also, not all semantically related results are meaningful to the same degree, we propose efficient pruning techniques to reduce the number of substitute queries and return only the top-$k$ semantically related results. We develop two query processing algorithms to evaluate the substitute queries against the data source based on our pruning techniques. We also develop a batch processing technique that exploits the shared keywords among the substitute queries to expedite the performance further. The extensive experiments with two real datasets validate the effectiveness and efficiency of our approach. 

\section{Acknowledgment} 
This work is supported by the Australian Research Council discovery grants DP140103499 and DP160102412.

\end{document}